\documentclass[12pt]{article}

\usepackage[utf8]{inputenc}
\usepackage[T1]{fontenc}
\usepackage[margin=1in]{geometry}
\usepackage{amsmath,amssymb,amsfonts,amsthm}
\usepackage{algorithm}
\usepackage{algpseudocode}
\usepackage{graphicx}
\usepackage{xcolor}
\usepackage{url}
\usepackage{bbm}
\usepackage{natbib}
\usepackage{subfigure}
\usepackage{booktabs}
\usepackage{authblk}
\usepackage{setspace}
\usepackage{appendix}
\usepackage{hyperref}
\usepackage{comment}

\graphicspath{{Fig_pdf/}}

\newtheorem{theorem}{Theorem}
\newtheorem{proposition}[theorem]{Proposition}
\newtheorem{corollary}{Corollary}
\newtheorem{lemma}{Lemma}
\newtheorem{assumption}{Assumption}

\newtheorem{example}{Example}
\newtheorem{remark}{Remark}

\DeclareMathOperator{\proj}{proj}
\newcommand{\RR}{\mathbb{R}}
\newcommand{\Sp}[1]{Sp}

\newcommand{\RankS}[3]{S^{\##1}_{#2,#3}} 
\newcommand{\ThreS}[3]{S^{(#1)}_{#2,#3}}

\newcommand{\GenS}[2]{\tilde S_{#1,#2}}
\newcommand{\Evid}[3]{O^{#1}_{#2,#3}} 
\newcommand{\RankEvid}[3]{O^{\##1}_{#2,#3}} 
\newcommand{\LLR}{GLR}

\title{Online Multivariate Changepoint Detection: Leveraging Links With Computational Geometry}

\author[1]{Liudmila Pishchagina}
\author[2]{Gaetano Romano}
\author[2]{Paul Fearnhead}
\author[1]{Vincent Runge}
\author[1,3,4]{Guillem Rigaill\thanks{Corresponding author: guillem.rigaill@inrae.fr}}

\affil[1]{Université Paris-Saclay, CNRS, Univ Evry, Laboratoire de Mathématiques et Modélisation d'Evry, 91037, Evry-Courcouronnes, France.}
\affil[2]{School of Mathematical Sciences, Lancaster University, Lancaster LA1 4YF, UK}
\affil[3]{Université Paris-Saclay, CNRS, INRAE, Univ Evry, Institute of Plant Sciences Paris-Saclay (IPS2), 91405, Orsay, France.}
\affil[4]{Université Paris-Saclay, AgroParisTech, INRAE, UMR MIA Paris-Saclay, 91120, Palaiseau, France.}
\date{}

\begin{document}

\maketitle

\begin{abstract}
The increasing volume of data streams poses significant computational challenges for detecting changepoints online. Likelihood-based methods are effective, but a naive sequential implementation becomes impractical online due to high computational costs. 
We develop an online algorithm that exactly calculates the likelihood ratio test for a single changepoint in $p$-dimensional data streams by leveraging a fascinating connection with computational geometry. This connection straightforwardly allows us to exactly recover sparse likelihood ratio statistics: that is assuming only a subset of the dimensions are changing.
Our algorithm is straightforward, fast, and apparently quasi-linear. A dyadic variant of our algorithm is provably quasi-linear, being $\mathcal{O}(n\log(n)^{p+1})$ for $n$ data points and $p$ less than $3$, but slower in practice. These algorithms are computationally impractical when $p$ is larger than $5$, and we provide an approximate algorithm suitable for such $p$ which is $\mathcal{O}(np\log(n)^{\tilde{p}+1}), $ for some user-specified $\tilde{p} \leq 5$. We derive statistical guarantees for the proposed procedures in the Gaussian case, 
and confirm the good computational and statistical performance, and usefulness, of the algorithms on both empirical data and NBA data.
\end{abstract}


\section{Introduction}
In recent years, many methods have been proposed for detecting one or multiple changepoints offline or online in data streams. The reason for such a keen interest in changepoint detection methods lies in its importance for various real-world applications, including bioinformatics \cite[]{olshen2004circular, Picard2005}, econometrics \cite[]{bai2003computation, Aue_monitoring}, medicine \cite[]{Bosc2003, Staudacher2005ANM}, climate and oceanography \cite[]{DucrRobitaille2003,Reeves2007}, finance \cite[]{Andreou, Fryzlewicz_2014}, autonomous driving \cite[]{galceran2017multipolicy}, entertainment \cite[]{Rybach, Radke}, computer vision \cite[]{ranganathan2012pliss} or neuroscience \cite[]{jewell2020fast}. 
Whether this is online (i.e. observations are processed as they become available) or offline (i.e. observations are processed only once the whole data is available) the ever-increasing size of the data stream raises computational challenges. Ideally, a changepoint procedure should be linear or quasi-linear in the size of the data and, for online applications, should also be able to process new observations sequentially, and as quickly as they arrive. A common idea for both offline and online parametric (multiple) changepoint model estimation, is to optimise a likelihood function or calculate likelihood ratio statistics \cite[]{siegmund1995using, Auger}. From a computational perspective this optimisation can either be done (1) heuristically using, for example, the Page-CUSUM statistic \cite[]{page1954continuous} in an online context or Binary Segmentation and its variants \cite[]{scott1974cluster,Fryzlewicz_2014} in an offline context, or (2) exactly using dynamic programming techniques as in the Segment Neighbourhood, or Optimal Partitioning algorithms \cite[]{Auger,jackson2005algorithm}. 
Importantly, these later algorithms used with an appropriate penalty enjoy good statistical properties \cite[]{Yao,lavielle2000least, Lebarbier2005,garreau2018consistent} and have shown good properties in numerous applications \cite[]{lai2005comparative,liehrmann2021increased}.

Significant speed-ups are possible by reducing the set of changepoint locations that need to be checked at each iteration of the dynamic programming recursion, a technique commonly known as pruning \cite[]{Killick, rigaill2015pruned, Maidstone, romano2022fast}. Two types of rules have been proposed that prune changepoint locations that will never contribute to the full optimisation cost, achieving significant speed-ups without sacrificing optimality. Inequality pruning rules, as in the PELT algorithm \cite[]{Killick}, apply to a wide range of changepoint models but are efficient only if the number of true changepoints increases quickly with the size of the data. 
Functional pruning, as in the pDPA, FPOP, or FOCuS algorithms \cite[]{rigaill2015pruned, Maidstone, romano2022fast}, applies mostly to univariate models, but is efficient even when the number of true changepoints is small compared  
to the size of the data.
Interestingly, these algorithms are sequential in nature and empirically quasi-linear. This quasi-linearity was recently proven for the online detection of a single changepoint \cite[]{romano2022fast}.
The pDPA, FPOP, and FOCuS algorithms rely on a functionalisation of the problem, that is they consider the likelihood as a function of the last segment parameter. A direct extension of these ideas to a higher dimension ($p>1$) has proven difficult because it involves computing the intersection and subtraction of many convex sets in $\mathbb{R}^p$ \cite[]{runge2020finite,pishchagina2023geometric}.

In this paper, we study online models with one changepoint, where the likelihood can be algebraically written as a log-likelihood from the $p$-variate natural exponential family. 
This includes, for example, the Gaussian, Poisson, categorical, and Pareto type-I distribution (assuming the minimum value is known) and also the non-parametric e-detectors of \cite{shin2022detectors} (see Section \ref{sec4_2_NBA}). We extend the functional pruning approach to such $p$-variate models using a linearisation argument embedding all likelihoods into a bigger space of dimension $p+2$. This idea allows us to demonstrate a link between the functional problem and problems studied in computational geometry.

In particular, we show that the solutions of this extension correspond to identifying the supporting hyperplanes of a particular
half-space intersection problem. 
For models with one changepoint, we demonstrate that this problem further simplifies and is dual to a lower dimensional convex hull problem of dimension $p+1$.
In this latter case, we further prove, assuming independent and identically distributed (i.i.d.) data, that the number of points on this hull at time $n$ is asymptotically $\mathcal{O}(\log^p(n))$. Leveraging this leads to an empirically log-linear per iteration algorithm for calculating the likelihood ratio test statistics for a change at~$t$.

At any time-step of this algorithm, we can exactly compute the maximum likelihood ratio-statistics and maximum sparse likelihood ratio-statistics allowing only $s$ coordinates to change with $s \leq p$. 
In  Supplementary Material \ref{app:stat_control}, we derive non-asymptotic bounds of these statistics to control the false alarm rate and derive a bound on the expected detection delay. While the computational cost is quasi-linear, our exact algorithms become impractical for large $p$, due to the $\log^p(n)$ scaling. For such cases we  

provide an approximation storing only an order of $\mathcal{O}(p\log^{\tilde{p}}(n))$ candidates on the hull, where $\tilde{p} \leq p$,  chosen by the user, represents the dimension of a projective linear space.

The paper has the following structure. Section \ref{sec1_Likelihood} describes our maximum likelihood framework. Section \ref{sec2_Functional_Pruning} presents the mathematical formulation of the functional pruning problem and how it is related to a half-space intersection problem and a convex hull problem. 
In Section \ref{sec3_MdFocus} 
we develop and introduce our proposed method, called MdFOCuS, together with an approximation that scales to larger values of $p$. An important feature of both algorithms is that they calculate sets of putative changepoint locations. Given such a set it is straightforward to maximise the test statistic of interest, including under constraints on the number of components that change, or allowing different models for different components. We provide statistical guarantees of the resulting test in the Gaussian setting and 
compare MdFOCuS with competitor algorithms in terms of statistical and computational efficiency -- showing its increased accuracy when the pre-change parameters are unknown.
In Section \ref{sec4_2_NBA} we apply our methodology to track the performances of an NBA team and discuss how our theory could accelerate the recently proposed non-parametric e-detectors \cite[]{shin2022detectors}. All proofs, and links to code for this paper, are given in the Supplementary Material.

\section{Maximum Likelihood Methods for Online Single Changepoint Detection}\label{sec1_Likelihood}
\subsection{Problem set-up}\label{sec1_1_Problem}

We consider observations $\{y_t\}_{t=1,2,\dots}$ in $\RR^{p'}$ and wish to detect the possible presence of a changepoint and estimate its location. If there is a change, we denote the time of the change as $\tau$. We will consider likelihood-based approach that postulates a model for the data in the segments before and after the changepoint.  Our model will be that the data point $y_t$ is i.i.d. 
from a distribution with density $f_Y(y_t|\theta_1)$ before the change and $f_Y(y_t|\theta_2)$ after the change. 

At time $n$, twice the log-likelihood for the time series $\{y_t\}_{t=1,\dots,n}$, in terms of the pre-changepoint parameter, $\theta_1$ in $\RR^{p''}$, the post-changepoint parameter, $\theta_2$ in $\RR^{p''}$, and the location of a change, $\tau$, is
\begin{equation}
\label{eq:likelihood}
\ell_{\tau,n}(\theta_1, \theta_2)  = 2 \left (  \sum_{t=1}^\tau \log f_Y(y_t|\theta_1) +  \sum_{t=\tau+1}^n \log f_Y(y_t|\theta_2) \right).
\end{equation}

Depending on the application one may assume that the pre-changepoint parameter $\theta_1$ is either known or unknown. In either case, we can get the log-likelihood for a changepoint at $\tau$ by maximising this equation over either the unknown parameter, $\theta_2$, or parameters, $\theta_1$ and $\theta_2$. In practice, we do not know the changepoint location, so we maximise this log-likelihood over $\tau$. We can define the maximum likelihood estimator for $\tau$, for the pre-changepoint parameter being either known or unknown respectively, as 
\begin{align*}
\hat{\tau}(\theta_1) = & \arg\max_{\tau} (\max_{\theta_2}  \ell_{\tau,n}(\theta_1, \theta_2))\,, \\
    \hat{\tau}(.)  = & \arg\max_{\tau} (\max_{\theta_1,\theta_2} \ell_{\tau,n}(\theta_1, \theta_2)). 
\end{align*}
We then obtain a Generalized Likelihood Ratio (GLR) statistic \cite[]{siegmund1995using} for a change $\tau$ by comparing it with twice the log-likelihood for a model with no change. That is $2\sum_{t=1}^n \log f_Y(y_t|\theta_1)$ if $\theta_1$ is  known and $\max_{\theta_1}2\sum_{t=1}^n  \log f_Y(y_t|\theta_1)$ if it is not.
Often it is appropriate \cite[see, for example,][]{chen2020highdimensional} to consider a sparse version of the likelihood that allows only $s$ coordinates to jump, that is adding the constraint $s=||\theta_1-\theta_2||_0$, where $||\mathbf{v}||_0$ is equal to the number of non-zero entries of $\mathbf{v}$. The estimator of $\tau$ in this case, assuming the pre-change mean is known, is  
\begin{align}\label{eq:restricted_max}
\hat{\tau}^{s}(\theta_1) = & \arg\max_{\tau} \left(\max_{\substack{\theta_2 \\ ||\theta_1-\theta_2||_0=s}}  \ell_{\tau,n}(\theta_1, \theta_2)\right)\,, 
\end{align}
with $\hat{\tau}^{s}(.)$ defined analogously for the case where the pre-change mean is unknown. Clearly, $\hat{\tau}^{p}(\theta_1) = \hat{\tau}(\theta_1)$ and $\hat{\tau}^{p}(.) = \hat{\tau}(.)$. 

Offline, optimising Equation \eqref{eq:likelihood} is simple. Indeed, assuming we have access to all the data, it suffices to compute, for every possible changepoint location from $1$ to $n-1$, the likelihood (which can be done efficiently using summary statistics) and then pick the best value. Online, the problem is substantially more difficult as one needs to recompute the maximum for every new observation. Applying the previous offline strategy for every new observation is inefficient as it takes a computational cost of order $n$ to process the $n$th observation. 
It is infeasible for an online application, where $n$ grows unbounded, and for this reason, various heuristic strategies have been proposed \cite[see][for a recent review]{wang2022sequential}. From a computational perspective, we roughly identify two main types of heuristics. 
The first idea is to restrict the set of possible $(\theta_1, \theta_2)$ pairs to be discrete. The simplest case is to consider just one pair of $(\theta_1, \theta_2)$ as in the Page CUSUM strategy \cite[]{page1954continuous}, but one typically considers a grid of values \cite[]{chen2020highdimensional}. The second idea is to restrict the set of changepoint locations one considers. For example, to consider a grid of window lengths $w_1,\ldots,w_M$ and at time $n$ compute the GLR statistic only for changes at $\tau=n-w_1,\ldots,n-w_M$ \cite[]{eiauer1978use,meier2021mosum}. But this naturally introduces a trade-off between computational efficiency and statistical power: the larger the grid of windows, the larger the computational overhead~\cite[]{wang2022sequential}.

Recently, the FOCuS algorithm has been introduced that calculates the GLR test exactly but with a per-iteration cost that increases only logarithmically with the amount of data \cite[]{romano2022fast,ward2022poissonfocus} - however, this algorithm is only applicable for detecting changes in a univariate feature of a data stream. A key insight from the FOCuS algorithm is that to maximise the likelihood, one only needs to consider changes that correspond to points on a particular convex hull in the plane. This link to geometrical features of the data is one that we leverage to develop algorithms for detecting changes in multivariate features.
\subsection{Detection of a single change in the parameter of a distribution in the exponential family}\label{sec_1_2_ExpFamily}

We briefly introduce some notation. We use $\mathbf{1}_p$ to denote a vector of length $p$ in which all elements are set to one. For any vector $\mu = (\mu^1, \dots, \mu^p)$ in $\RR^p$, we define its $q$-norm as $||\mu||_q := \left( \sum_{k=1}^p |\mu^k|^q \right)^{1/q}$, where $q\geq 1$ in $\mathbb R$. For two vectors, $\mu_1 = (\mu^1_1, \dots, \mu^p_1)$ and $\mu_2 = (\mu^1_2, \dots, \mu^p_2)$ in $\RR^p$, we define their scalar product as $\langle \mu_1,\mu_2\rangle = \sum_{k=1}^p \mu_1^k\mu_2^k$. The symbol $|\cdot|$ signifies the cardinality of $\cdot$. 

We are now ready to describe our model. We assume that the density $f_Y(y_t|\theta_k)$, for observation $y_t$ in segment $k \in \{1,2\}$ which has segment-specific parameter $\theta_k$, can be written in the form
\begin{equation}
\label{eq:density}  
f_Y(y_t|\theta_k) = \exp \left \{- \frac{1}{2} \left(A'(r(\theta_{k})) - 2\langle  s(y_t), r(\theta_{k}) \rangle +B'(s(y_t))\right)\right\}\,,
\end{equation}
where $s$ and $r$ are functions mapping respectively the observations, $y_t$, and the parameters, $\theta_{k}$, to $\RR^p$, $A'$ is some convex function from $\RR^p$ to $\RR$, and $B'$ from $\RR^p$ to $\RR$. Defining the natural parameter $\eta_k$ as  $\eta_k = r(\theta_k),$ and $x_t$ as $x_t = s(y_t)$ we rewrite this density as a function of $\eta_k$ and $x_t$:
\begin{equation}
\label{eq:naturalDensity}  
f_X(x_t|\eta_k)  =  \exp\left\{ - \frac{1}{2} \left(A'(\eta_k) - 2\langle x_t, \eta_k \rangle +B'(x_t)\right)\right\}\,.
\end{equation}
This class encompasses many common exponential family models, see Table \ref{tab1_ExpFamily} in the Supplementary Material. There, we also show how a change in mean and variance of a Gaussian signal in $\RR$ can be expressed in this form. 

We can then define the log-likelihood of our (transformed) data $x_1,\ldots,x_t$, under a model with a change at $\tau$ in terms of $\eta_1$ and $\eta_2$ as
\begin{align}
\label{eq:loglikelihood_eta}
\ell_{\tau,n}(\eta_1, \eta_2) =  -\sum_{t=1}^\tau [A'(\eta_1) - 2\langle x_t, \eta_1 \rangle + B'(x_t)]  - \sum_{t=\tau+1}^n [A'(\eta_2) - 2\langle x_t, \eta_2\rangle + B'(x_t)]. 
\end{align}
The log-likelihood estimator for $\tau$, for the pre-change parameter being either known or unknown respectively, is
\begin{align*}
    \hat{\tau}(\eta_1) =  \arg\max_{\tau} (\max_{\eta_2}  \ell_{\tau,n}(\eta_1, \eta_2))\,, \mbox{ or }
    \hat{\tau}(.)  =  \arg\max_{\tau} (\max_{\eta_1,\eta_2} \ell_{\tau,n}(\eta_1, \eta_2)). 
\end{align*} 
The GLR statistic for a change $\tau$ is obtained  by comparing with $\sum_{t=1}^n [A'(\eta_1) - 2\langle x_t, \eta_1 \rangle + B'(x_t)]$ if $\eta_1$ is known and with $\max_{\eta_1}{\sum_{t=1}^n [A'(\eta_1) - 2\langle x_t, \eta_1 \rangle + B'(x_t)]}$ if it is not.

For any fixed $\eta_1$ and any $\eta_2$ the addition of $\sum_{t=1}^n \left[A'(\eta_2) - 2\langle x_t, \eta_2\rangle + B '(x_t)\right]$ to \eqref{eq:loglikelihood_eta} does not change $\arg\max_\tau \ell_{\tau,n}(\eta_1, \eta_2)$. Moreover, this addition eliminates the dependence on $n$. Hence, when searching for the maximum likelihood change, we can redefine $\ell_{\tau, n}(\eta_1, \eta_2)$ as follows:
\begin{equation}
\label{eq:likelihood_transf}
\ell_{\tau}(\eta_1, \eta_2) = \tau \left[ A'(\eta_2) - A'(\eta_1)\right] - 2\left\langle \sum_{t=1}^\tau x_t, \eta_2 - \eta_1\right\rangle,
\end{equation}
and be sure that the optimal change optimises $\ell_{\tau}(\eta_1, \eta_2)$ at least for a pair $(\eta_1, \eta_2)$. Sometimes it might make sense to regularise the parameters. This can be done by adding a function of the parameter $\Omega(\eta_1, \eta_2)$, such as  $\Omega(\eta_1, \eta_2) = \lambda (||\eta_1||_1+||\eta_2||_1)$ for some $\lambda > 0$, to \eqref{eq:loglikelihood_eta}.  We will not study this further, but by the previous addition trick, any such regularisation term could be subtracted, and the rest of our argument would hold.

Throughout, we will primarily focus on the case where $\eta_1$ is known.
Our results for the case where it is unknown 
are derived as a by-product: considering all possible $\eta_1$ simultaneously and realising they all lead to the same geometrical problem. Considering $\eta_1$ to be fixed we thus introduce the following definitions:
$\mu=\eta_2-\eta_1,\quad A(\mu) =A'(\mu +\eta_1) - A'(\eta_1)$,
and rewrite the function $\ell_{\tau}(\eta_1, \eta_2)$ as a function of $\mu$:
\begin{equation}
\label{eq:likelihood_change_variables} 
\ell_{\tau}(\mu) = \tau A(\mu) - 2\left\langle \sum_{t=1}^\tau x_t, \mu \right\rangle.
\end{equation}

\section{Linear Relaxation of the Functional Pruning Problem}\label{sec2_Functional_Pruning}

\subsection{Mathematical Formulation}\label{sec2_1_General_formulation}

Functional pruning ideas give the potential of an efficient approach to recursively calculate the  GLR statistic for a change at time $n+1$ given calculations at time~$n$ \cite[]{romano2022fast}. The idea is to restrict the set of possible changepoints prior to $n+1$ that we need to consider. Working out which changepoint locations we need to consider is obtained by viewing the GLR statistic as a function of the change parameter $\mu = \eta_2-\eta_1$, and finding the set of changepoint locations which contribute to this, i.e. for which the GLR statistic is maximum for some value of $\mu$.

In more detail, at step $n$ of a functional pruning algorithm, we consider a list of possible changepoints $\tau$ of the data from $x_1$, $x_2$ \ldots to $x_n$. Each changepoint is represented by a function $f_{\tau}$ of the last segment parameter $\mu$ and a set of $\mu$ values for which it is maximising the likelihood: $Set_{\tau}= \{\mu \; | \; \forall \ \tau'\neq \tau, \ f_\tau(\mu) > f_{\tau'}(\mu) \}$. To proceed to the next data point $n+1$ the update is informally done as follows. 
\begin{enumerate}
\item We add $n$ to the list of changepoints, with this represented by a function $f_n$.
\item We compare all other functions to this new function and accordingly restrict the set of $\mu$ for which they are optimal. 
\item We discard all changepoints whose set $Set_{\tau}$ is empty.
\end{enumerate}

The left column of Figure \ref{fig:linearisation} provides a graphical representation of these steps. The pruning of the third step is valid because, whatever the data points after $n+1$ might be, these functions (or equivalently changepoints) can never maximise the log-likelihood. Formally, define functions $f_\tau$, for $\tau$ in some set $\mathcal{T}$, such that 
\begin{equation}
\label{eq:FP_equation}
f_\tau(\mu) = a_\tau A(\mu) - 2 \langle b_\tau, \mu \rangle,
\end{equation}
where $a_\tau$ is in $\mathbb{R}$, and $b_\tau$ is in $\RR^p$ and $A$ is a convex function from $\mathcal{D}_A$, a subset of $\RR^p$. We define $Im(A)$ as the image of the function $A$. 
Function $f_\tau$ corresponds to the log-likelihood for a segmentation with a change at $\tau$ and with parameter $\mu$. This function will depend on $\mu$ through the log-likelihood contribution of the data in the segment after $\tau$, with this data defining the co-efficients $a_\tau$ and $b_\tau$. 
We then define
\begin{equation*}
F_{\mathcal T}(\mu) = \max_{\tau \in \mathcal T} \{  f_\tau(\mu)\}\,,
\end{equation*}
and our goal is to recover all the indices $\tau$ maximising $F_{\mathcal T}$  for at least one value of $\mu$. Define:
\begin{align}
\label{eq:Set_Of_FPLin_Minimizers}
\mathcal F_{\mathcal T} =  \{\tau \; | \; \exists \mu \in \mathcal{D}_A \subset \RR^p,  \forall \tau' \in \mathcal T \text{ with } \tau' \neq \tau : \ f_\tau(\mu) > f_{\tau'}(\mu) \}.
\end{align}
This is defined using a strict inequality. As function $A$ is strictly convex, (and assuming that for all $\tau \neq \tau'$ we cannot have simultaneously $a_\tau=a_{\tau'}$ and $b_\tau=b_{\tau'}$) equality can be obtained only over a set of zero measure, and can be ignored. 
%
\subsection{Linearisation of the Functional Pruning Problem}
\label{sec2_2_Linearization}
The size of the set $\mathcal{F}_{\mathcal{T}}$ depends on the nature of the function $A$. In essence, to consider any possible function $A$, we introduce a new parameter $\lambda$ to replace $A(\mu)$. This way we get a problem that does not depend on the form of $A$ and, as we will see, is easier to solve, albeit at the cost of an additional dimension. To be specific, for any function $\mu \mapsto f_\tau(\mu)$ we associate an (augmented) function $(\lambda, \mu) \mapsto g_\tau(\lambda, \mu)$ defined as:
\begin{equation*}
g_\tau(\lambda, \mu) = a_\tau \lambda - 2 \langle b_\tau, \mu \rangle.
\end{equation*}
As for the $f_\tau$, we define the maximum $G_{\mathcal{T}}(\lambda, \mu)$ as
\begin{equation*}
G_{\mathcal{T}}(\lambda, \mu) = \max_{\tau \in \mathcal{T}} \{ g_\tau(\lambda, \mu)\}.
\end{equation*}
Again our goal will be to recover the index $\tau$ strictly maximising $G_{\mathcal{T}}$:
\begin{align*}
\mathcal{G}_{\mathcal{T}} =  \{ \tau \; | \; \exists (\lambda,\mu) \in Im(A)\times\mathcal{D}_A,
      \forall \tau' \in \mathcal T \text{ with } \tau' \neq \tau :  \ g_\tau(\lambda, \mu) > g_{\tau'}(\lambda, \mu) \}. 
\end{align*}
In the previous definition, the parameter $\lambda$ is constrained to be in the image of the function $A$. 
The values attained by $f_\tau$ are also attained by $g_\tau$ on a parametric curve of the $(\lambda, \mu)$ space given by equation $\lambda = A(\mu)$, so that $f_\tau(\mu)=g_\tau(A(\mu),\mu)$. 
This trivially leads to the following result that the set of $\tau$ strictly maximising $\mathcal{G}_{\mathcal{T}}$ must contain the set strictly maximising $\mathcal{F}_{\mathcal{T}}$.
\begin{theorem}
\label{th:inclusion}
$\mathcal{F}_{\mathcal{T}} \subseteq \mathcal{G}_{\mathcal{T}}\,.$ 
\end{theorem}
Without any further assumptions on the function $A$, the set $\mathcal{G}_{\mathcal{T}}$ could be much larger than $\mathcal{F}_{\mathcal{T}}$. The following theorem demonstrates that equality holds for an important class of functions $A(\mu)$.
\begin{theorem}\label{th:equality}
If $A(\mu) = \|\mu\|^q$ with $q > 1$, then $\mathcal{F}_{\mathcal{T}} = \mathcal{G}_{\mathcal{T}}$.
\end{theorem}


In particular, with $q = 2$ we get the Gaussian loss function with $\eta_1=0$. Assuming $\eta_1=0$ is not restrictive because if $\eta_1$ is not $0$ we can consider the translated data $(x_t - \eta_1)$.
A more general result, with milder conditions for function $A$, is proposed in Supplementary Material~\ref{append1_index_Set}. One sufficient condition consists of the strong convexity of the function $A$ with $A(0) = 0$.  

%
\subsection{Half-space intersections and functional pruning}\label{sec2_3_HalfspaceIntersections}
All $g_\tau$ functions are linear in $(\lambda, \mu)$. Considering yet another variable $\kappa$ we can associate to every $g_\tau$ a function $h_\tau$ and a half-space $H_\tau$ of $\mathbb{R}^{p+2}$: 
\begin{equation*}
H_\tau = \{ \kappa, \lambda, \mu \ | \ h_\tau(\kappa, \lambda, \mu) = a_\tau\lambda - 2 \langle b_\tau, \mu\rangle -\kappa \leq 0 \}.
\end{equation*}
From here we straightforwardly get that our linearised functional problem is included in a half-space intersection problem. This result is formalised in the following theorem.
\begin{theorem}\label{th:halfspace_intersection}
If $\tau$ is in $\mathcal{G}_{\mathcal{T}}$
then 
\begin{enumerate}
\item there exists $(\kappa, \lambda, \mu)$ such that $h_\tau(\kappa, \lambda, \mu) =0$ and for all $\tau' \neq \tau$, $h_{\tau'}(\kappa, \lambda, \mu) < 0$;
\item $h_\tau(\kappa, \lambda, \mu)=0$ defines a supporting hyperplane of the set $\cap_\tau H_\tau$.
\end{enumerate}
\end{theorem}
The set $\cap_\tau H_\tau$ defines an unbounded polyhedron as taking any $\kappa \geq \max_{\tau} \{c_\tau\}$, $\lambda=0$, and $\mu=0$ we solve all inequalities. 
In Figure \ref{fig:linearisation} we illustrate our linearised functional problem. 

Computing the intersection of $|\mathcal{T}|$ half-spaces is a well-studied problem in computational geometry. It is closely related to the convex hull problem by projective duality \cite[]{avis1995good}, and several algorithms have been proposed  \cite[see for example][]{preparata1979finding,anderson_1984,Barber1996,ORourke_1998}. See \cite{seidel2017convex} for a recent short introduction to this topic.
\begin{figure}[!ht]
    \centering

    \begin{subfigure} 
        \centering
        \includegraphics[width=\linewidth]{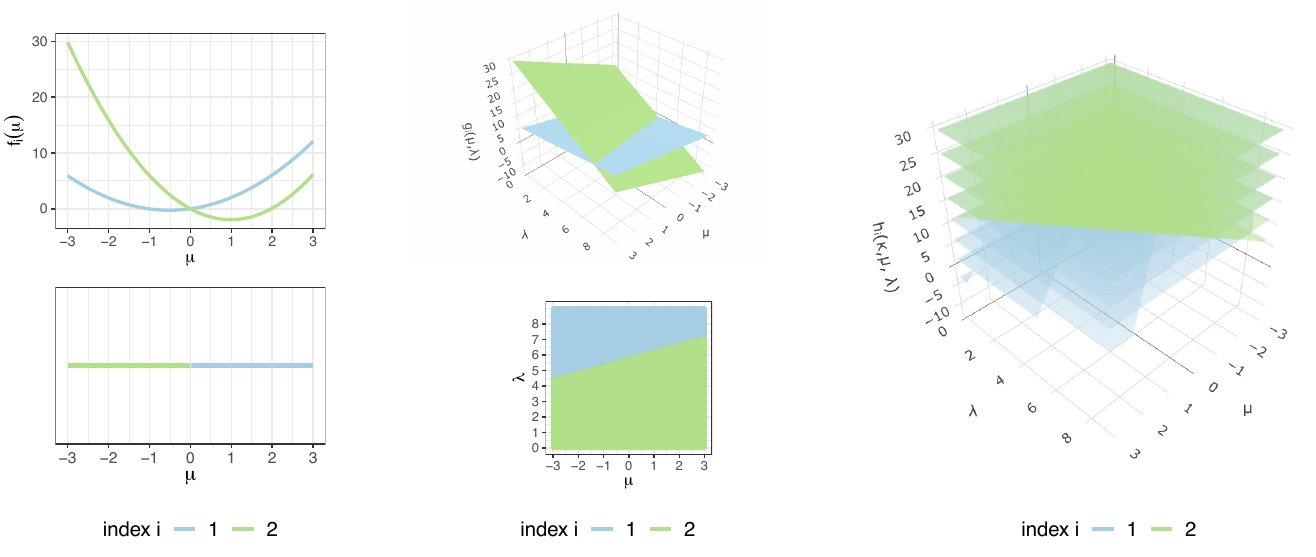}
    \end{subfigure}

    \begin{subfigure} 
        \centering
        \includegraphics[width=\linewidth]{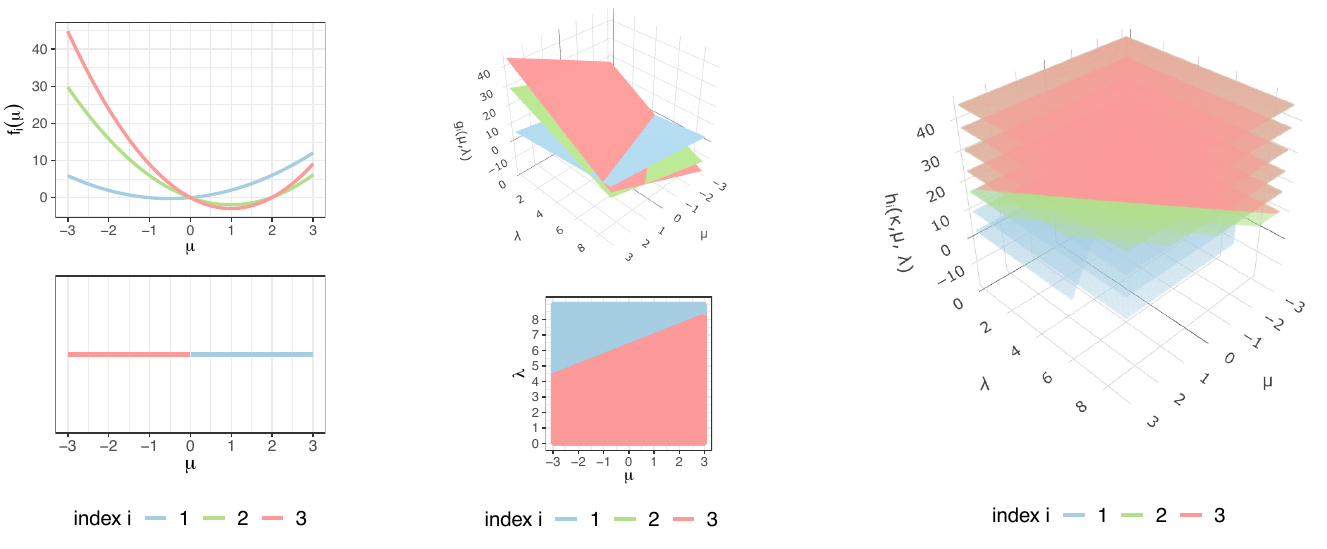}
    \end{subfigure}

    \caption{Example of the functional pruning problem and linearised functional problem over times: $f_i$ and $\proj \{\max_{i} \{f_i(\mu)\}\}$ functions (left); $g_i$ and $ \proj \{\max_{i} \{g_i(\mu,\lambda)\}\}$ functions (middle); the level surfaces of $h_i$ functions, where each surface is for a fixed $\kappa$ value (right) for changepoint candidates $i \in \mathcal{T}$. We simulate time series $\{x_t\}_{t=1,2,..} \sim \mathcal{N}(0,1)$ and consider two time-steps, $n = 3$ (top two rows) and $n = 4$ (bottom two rows). We define the coefficients $a_i$ as $i$, $b_i$ as $\sum_{t=1}^i x_t$. The change $i=2$ associated with quadratics $f_2$ and surface $g_2$ is pruned at $n = 4$.}\label{fig:linearisation}
\end{figure}
\subsection{Projective-duality} 
\label{sec:Projective_duality}

The following theorem states that in our case the half-space intersection problem in dimension $p+2$ is dual to a convex hull problem in dimension $p+1$. 

\begin{theorem}
\label{th:convexhull}
An index $\tau$ being in $\mathcal{G}_{\mathcal{T}}$ is equivalent to the point $(a_\tau, b_\tau)$ in $\mathbb{R}^{p+1}$ being a vertex of the convex hull of $\{ (a_\tau, b_\tau) \}_{\tau \in \mathcal{T}}$.
\end{theorem}

We illustrate the points on the convex hull of Theorem \ref{th:convexhull} for $p=1$ and $p=2$ in Figure \ref{Figure_CH_and_RW}. As will be seen in Section \ref{sec3_MdFocus}, this is relevant for the online detection of a single changepoint, see \eqref{eq:likelihood_change_variables}. 
Importantly, computing the convex hull of $|\mathcal T|$ points is a well-studied problem in computational geometry.
Over the years many algorithms have been developed \cite[]{Chand1970AnAF,Preparata1977,Miller1988,Barber1996,kenwright2023convex}.
In the rest of this paper, we will use the \texttt{QuickHull} algorithm of \cite{Barber1996} and its implementation in the \texttt{qhull} library \url{http://www.qhull.org/}.
The \texttt{QuickHull} algorithm works in any dimension and is reasonably fast for our application up to $p = 5$.  

\begin{figure}[!t]%
    \centering
    \begin{subfigure}
        \centering
        \includegraphics[scale=0.65]{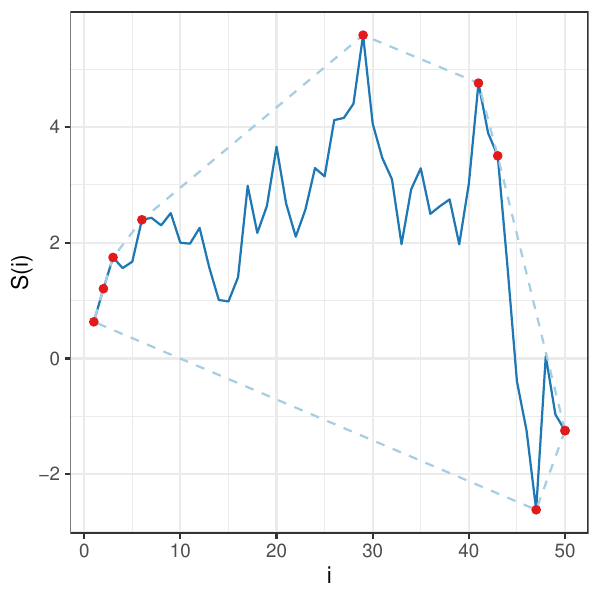}
    \end{subfigure}%
    \hfill
    \begin{subfigure}
    \centering
        \includegraphics[scale=0.25]{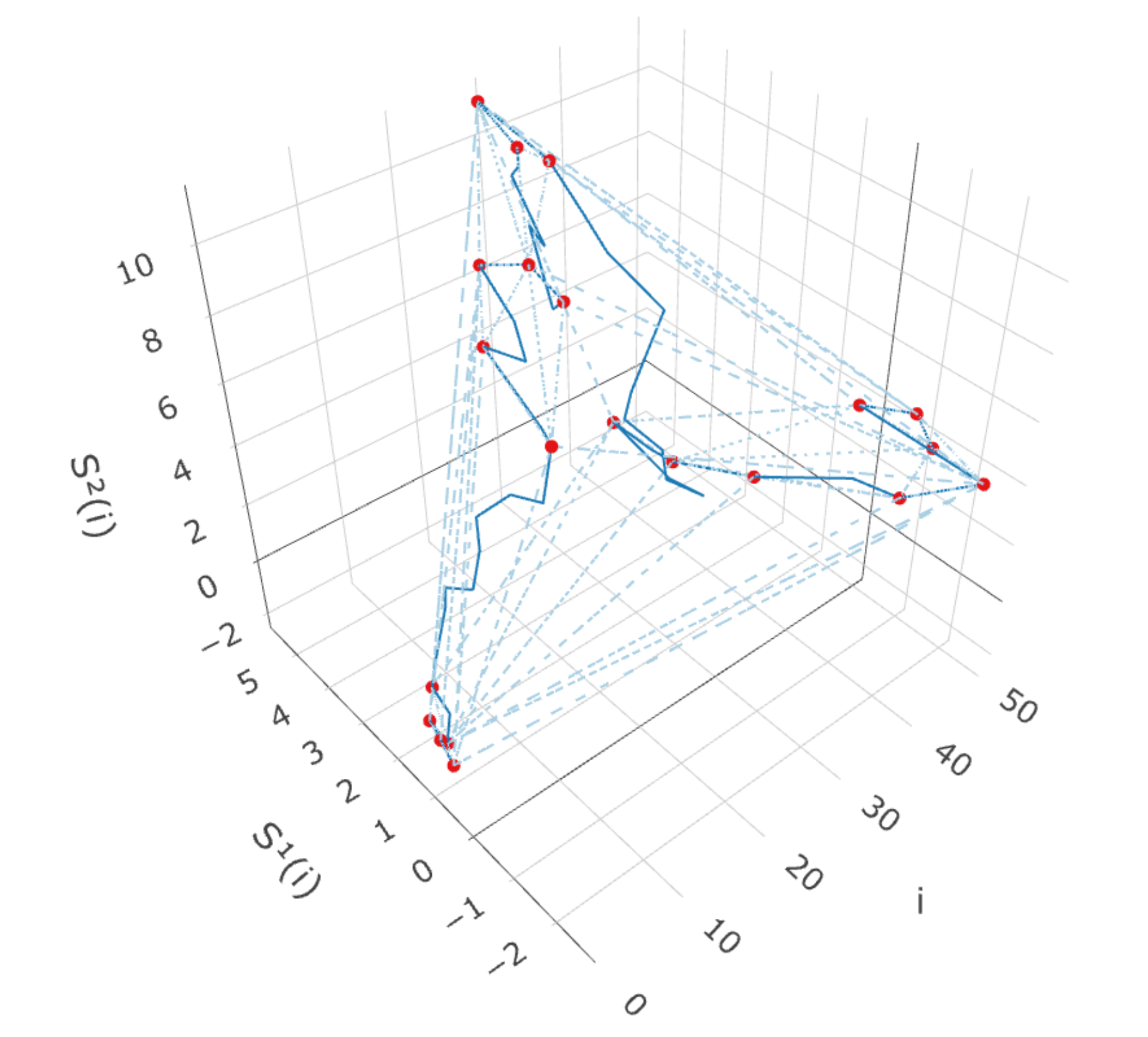}
    \end{subfigure}

    \caption{Plot of $S(i) =\sum_{t=1}^i x_t$  as a function of $i$ (dark blue) and the convex hull of  $\left\{\left( a_i=i,b_i=S(i)\right)\right\}_{i \in \{1,\dots, n-1\}}$ (dashed blue line) for dimensions $p=1$ (left) and $2$ (right). The vertices of the convex hull $\tau \in \mathcal{T}$ are marked in red. We simulated time series $\{x_t\}_{t \in \{1,\dots,n\}}$ with $n = 51$ data points with $x_t  \sim \mathcal N_p(0, I_p)$.}\label{Figure_CH_and_RW}
\end{figure}
\section{On The Detection of a Single Changepoint Online} \label{sec3_MdFocus}
\subsection{From the Functional Description to  a Convex Hull} \label{sec_3_1_ConvexHull}
We now apply the results of Section \ref{sec2_Functional_Pruning} to the problem of detecting a single changepoint described in Section \ref{sec_1_2_ExpFamily}, to develop an efficient online algorithm. To do this, it is helpful to define a one-to-one mapping $P$ between any changepoint $\tau$ and a point in dimension $p+1$: 
\begin{equation}
\label{eq:point_definition}
    P(\tau) = \left(\tau, \sum_{t=1}^{\tau} x_t\right).
\end{equation}
The following theorem, based on Theorem \ref{th:convexhull} essentially states that the changepoint maximising the likelihood \eqref{eq:loglikelihood_eta}
for some value $\eta_1$ and $\eta_2$ corresponds to a point on a particular convex hull.
\begin{theorem}\label{th:changepoint_andhull}
Assume the log-likelihood, $\ell_{\tau,n}(\eta_1, \eta_2)$ of a change at $\tau$ in $\{1, \ldots, n-1\}$ can be written as in  \eqref{eq:loglikelihood_eta}. Then, for any $\eta_1$ and $s \leq p$ taking $\hat{\tau} = \hat{\tau}^{s}(\eta_1)$ or $\hat{\tau}^{s}(.)$ we have that $P(\hat{\tau})$  
is on the convex hull of  $\{P(\tau)\}_{\tau \in \{1, \ldots, n-1\}}$.
\end{theorem}
\begin{remark}
This result also extends to other test statistics that maximise the sparse log-likelihood for some $s$. This includes statistics that threshold the contribution from each component of the time series (see Section \ref{app:stat_control_definition} in Supplementary Material).
\end{remark}

Theorem \ref{th:changepoint_andhull} is coherent with the notion of pruning as in the FOCuS algorithm \cite[]{romano2022fast}. 
If $P(\tau_0)$ is strictly inside the hull, i.e. not on the boundary of the hull $\{ P(\tau)\}_{\tau \in \{1, \ldots, n-1\}}$ at time $n$, then it is inside the hull at any future time $n' \geq n$. Thus, based on Theorem~\ref{th:changepoint_andhull}, the likelihood associated with the change at time $\tau_0$ will never be optimal after time $n$.

\subsection{Bound on The Number of Changepoint Candidates}\label{sec3_2_Bound}
From a computational perspective, Theorem \ref{th:changepoint_andhull} is useful only if the set of points on the convex hull is significantly smaller than $n$. We now use recent results \cite[]{kabluchko2017convex} to bound the expected number of faces and vertices of the convex hull $\{P(\tau) \}_{\tau \in \{1, \ldots, n-1\}}$, assuming that the components of $x_\tau$, $x_\tau^d$ (with $d= 1,\dots, p$), are i.i.d. 
with a continuous distribution. In turn, using Theorem \ref{th:convexhull}, this gives us an upper bound on the expected number of changepoints that can maximise the likelihood function \eqref{eq:loglikelihood_eta} for some value of $\eta_1$ and $\eta_2$. This bound will be useful to derive guarantees on the computational complexity of the procedure (see Algorithms \ref{MdFOCuS_algo} and \ref{MdFOCuS_algo_dyadic}).

\begin{theorem}
\label{th:boundnumbervertexes}
Assuming all $x_\tau^d$ (with $d= 1, \dots, p$) are i.i.d. 
with a continuous distribution, then the expected number of faces (denoted as $U_{n}^p$) and vertices (denoted as $V_{n}^p$) of the convex hull of $\{P(\tau) \}_{\tau \in \{1, \ldots, n-1\}}$ are the following:
\begin{equation*}
\mathbb{E}(U_n^p) = \frac{2p!}{(n-1)!} \begin{bmatrix}
n \\
p+1\\
\end{bmatrix}, 
\quad \mathbb{E}(V_n^p)= \frac{2}{(n-1)!} \sum_{l=0}^{\infty} \begin{bmatrix}
n \\
p+1- 2l\\
\end{bmatrix},    
\end{equation*}
where $\begin{bmatrix}
n \\
m\\ 
\end{bmatrix}$
are the Stirling numbers of the first kind \cite[]{Adamchik1997OnSN}.
\end{theorem}

\begin{corollary}
\label{col:boundnumbervertexes}
We have the following equivalents for the expected numbers of faces $U_n^p$ and vertices $V_n^p$:
\begin{equation*}
  \mathbb{E}(U_n^p) \sim 2 \left(\log(n)\right)^p\,, \quad \mathbb{E}(V_n^p)\sim \frac{2}{p!}  \left(\log(n)\right)^p\,, \quad n \rightarrow \infty.  
\end{equation*}
\end{corollary}
The proof of Corollary \ref{col:boundnumbervertexes} follows from Remark 1.4 of \cite{kabluchko2017convex}. 
The theorem and corollary provide an expected upper limit on the number of points, or candidate locations for changepoints, that need to be evaluated to compute the CUSUM likelihood ratio test \eqref{eq:loglikelihood_eta} exactly without resorting to approximations. This limit increases as $\mathcal{O}(\log(n)^p)$ and so will give small bounds on the number of candidate changepoints for small $p$. We empirically illustrate the validity of Theorem~\ref{th:boundnumbervertexes} for i.i.d. Gaussian and Poisson data in supplementary material \ref{append4a_GaussianMd} and \ref{append5_PoissonMd}.

\subsection{MdFOCuS: a Simple Online Algorithm}\label{sec3_3_Algo}

In this section, we propose an algorithm to compute the GLR statistics at every time step $n$ for a known and unknown pre-change parameter $\eta_1$. In both cases, based on Theorem \ref{th:changepoint_andhull} we only need to consider the changepoints $\tau$ such that $P(\tau)$ is on the hull of $\left\{P(\tau')\right \}_{\tau' \in \{1,\dots, n-1\}}. $ 

Let us call $\mathcal{T}_n$ the set of index $\tau$ such that $P(\tau) = (\tau, \sum_{t=1}^\tau x_t)$ is on the convex hull of all $\{P(\tau)\}_{\tau \in\{1, \ldots, n-1\}}$. Restating Theorem~\ref{th:convexhull}, the changepoint $\hat{\tau}$ maximising the likelihood \eqref{eq:loglikelihood_eta} 
is in $\mathcal{T}_n$. Hence, at step $n$, we only need to compute the likelihood (and likelihood ratio statistics) of changepoints $\tau$ in $\mathcal{T}_n$. Furthermore, as already explained after Theorem \ref{th:changepoint_andhull} we can prune all $\tau$ that are not in $\mathcal{T}_n$ as they will never be in $\mathcal{T}_{n'}$ for any $n'\geq n$.

Computing the convex hull from scratch at every time step $n$ is inefficient: it requires at least $\mathcal{O}(n)$ operations and leads to an overall complexity at least quadratic in $n$. Ideally, we would like to update the hull using an efficient online or dynamic convex hull algorithm as the one in \cite{Clarkson1989}. There exist implementations of such an algorithm in dimensions $p=2$ \cite[]{Melkman1987} and $p=3$ \cite[]{cgal:hs-ch3-23b}. However, for simplicity, here we rely on an iterative use of the offline {\texttt{QuickHull}} algorithm \cite[]{Barber1996} that also works for larger dimensions and is available in the {\texttt{geometry R}} package \cite[]{geometry_R_package}.

In detail, for any $n' > n$ we have $\mathcal{T}_{n'} \subset \mathcal{T}_{n} \cup \{n, \ldots, n'-1\}$. Therefore if we have already computed $\mathcal{T}_{n}$, we only apply {\texttt{QuickHull}} to the points $P(\tau)$ with $\tau$ in $\mathcal{T}_{n} \cup \{n, \ldots, n'-1\}$. A first idea would be to apply the {\texttt{QuickHull}} algorithm at every time step $n, n+1, \ldots$, that is compute
$\mathcal{T}_{n}$ as the hull of all $P(\tau)$ with $\tau$ in $\mathcal{T}_{n-1} \cup \{n\}$. 
In practice, we discovered that running the {\texttt{QuickHull}} algorithm intermittently, rather than at every iteration, resulted in faster run times. Specifically, we maintain a set of linearly increasing indices, denoted as $\mathcal{T}'_n$, and only periodically execute {\texttt{QuickHull}} to reconstruct the true $\mathcal{T}_n$.
More precisely, to decide when to run {\texttt{QuickHull}}, we update a maximum size variable for $\mathcal{T}'_n$ (denoted \texttt{maxSize}). Then, at step $n+1$: \begin{itemize}
    \item if $|\mathcal{T}'_{n}|$, is smaller than \texttt{maxSize}  we set $\mathcal{T}'_{n+1}$ to $\mathcal{T}'_{n} \cup \{n\}$,
    \item otherwise we compute $\mathcal{T}'_{n+1}$ as the convex hull of all $P(\tau)$ with $\tau$ in $\mathcal{T}'_{n} \cup \{n\}$. 
\end{itemize} 
In the latter case we also update \texttt{maxSize} to $\texttt{maxSize}=\lfloor \alpha|\mathcal{T}'_{n+1}| + \beta \rfloor$, with $\alpha \geq 1$ and $\beta \geq 0$. It is important to note that, in this way, we always retain a few more possible changepoints than necessary:  $\mathcal{T}_{n+1} \subseteq \mathcal{T}'_{n+1}$.

If $\beta=1$ and $\alpha=1$ we run the {\texttt{QuickHull}} algorithm at every time step. Larger values of $\alpha$ mean we run it less often, but at the cost of having a longer list of indices in $\mathcal{T}'_n$. We empirically evaluated different values of $\alpha$ with $\beta=1$, and found $\alpha=2$ gave good practical performance; see Supplementary Material \ref{append4_Runtime_alpha}.

A formal description of our procedure, called MdFOCuS is presented in Algorithm~\ref{MdFOCuS_algo}. The algorithm covers both known and unknown pre-change parameters $\eta_1$, with the only difference being the calculation of the likelihood (lines 7 and 8). Therefore, the two variants store the same number of candidate changepoints and have the same time and memory complexities. In detail, computing the full likelihood ratio statistics for any particular changepoint has a time complexity of $\mathcal{O}(p)$, and therefore lines 7 and 8 in Algorithm~\ref{MdFOCuS_algo} have a complexity of  $\mathcal{O}(|\mathcal{T}|p)$. Recovering the maximum has complexity $\mathcal{O}(|\mathcal{T}|)$. Overall, this is $\mathcal{O}(|\mathcal{T}|p)$. 

\begin{algorithm}[ht]
\caption{MdFOCuS algorithm}
\label{MdFOCuS_algo}
\begin{algorithmic}[1]
\State {\bf Inputs:} $\{x_t\}_{t=1,2,\dots}$ $p$-variate independent time series; $thrs$ threshold; $\eta_1$ pre-change parameter (if known); $\texttt{maxSize} > p+1$, $\alpha \geq 1$, $\beta\geq 0$, pruning parameters.
\State {\bf Outputs:} stopping time $n$ and maximum likelihood change : $\hat{\tau}(.)$ or $\hat{\tau}(\eta_1)$
    \State $\LLR \gets -\infty,\quad \mathcal T \gets  \emptyset, \quad n \gets  1$ 
\While{($\LLR < thrs$)}  
     \State $n \gets n+1$
     \State $\mathcal T \gets \mathcal T \cup \{n-1\}$
     \State $\hat{m} \gets \max_{\tau \in \mathcal T} \{\max_{\eta_1, \eta_2} \ell_{\tau,n} (\eta_1,\eta_2)\}$ \Comment{ $\eta_1$ is fixed if pre-change known}
    \State $\LLR \gets \hat{m} - \max_{\eta_1, \eta_2} \ell_{n,n} (\eta_1,\eta_2)$  \Comment{ $\eta_1$ is fixed if pre-change known}
     \If{($|\mathcal{T}| > \texttt{maxSize}$)}
        \State $ \mathcal T \gets \Call{QuickHull}{\{P(\tau)\}_{\tau \in \mathcal T}}$ \Comment{Find the indices on the hull of all $P(\tau)$}
        \State $\texttt{maxSize} \gets   \lfloor \alpha|\mathcal{T}| + \beta \rfloor $
      \EndIf
\EndWhile
\State \Return $n$ and  $\LLR$
\end{algorithmic}
\end{algorithm}

\label{sec-EmpiricalRunTimeGauss} 
A link to a simple \texttt{R} implementation of Algorithm \ref{MdFOCuS_algo} for a change in the mean of $p$-dimensional Gaussian signal can be found in the Supplementary Materials. Our code is based on the {\texttt{convhulln}} function implementing the {\texttt{QuickHull}} algorithm in the {\texttt{geometry}} {\texttt{R}} package, and for $p=3$ it takes 2 minutes to process $n=5\times 10^5$ data points on an
Intel(R) Xeon(R) Gold 6252 CPU @ 2.10GHz processor choosing $\alpha=2$ and $\beta=1$. As can be seen in Figure \ref{fig:runtime_heuristic} (left), such run time is comparable to that of the {\texttt{ocd}} {\texttt{R}} package \cite[]{chen2020highdimensional}.
 
We also implemented Algorithm \ref{MdFOCuS_algo} for a change in the mean of a multi-dimensional Gaussian signal in R/C++ using the {\texttt{qhull}} library.  We use the Gaussian distribution as an example of a distribution with continuous density and independent dimensions.
To be specific, for $p$ from $1$ to $6$, we simulated time series $x_{1:n}$, with $x_t \sim \mathcal N_p(0, I_p)$ i.i.d. We varied the time series length from $2^{10}+1 (\approx 10^3)$ to $2^{23}+1(\approx8\times10^6)$ and simulated $100$ data sets for each length.
The average run times as a function of $n$ are presented on the log-scale in Figure \ref{GM_Runtime}. As expected, it is faster than our {\texttt{R}} implementation: for example, for $p=3$ it processes in 2 minutes for $n=10^6$ rather than $5\times 10^5$. The run times for a known or unknown $\eta_1$ are almost identical. Finally, we made a simple linear regression to model the logarithm of the run time as a function of $\log(n)$. The estimated slope coefficients are reported in Figure \ref{GM_Runtime} and are all smaller than $1.2$. 
\begin{figure*}[!t]%
     \centering
     \includegraphics[width=0.6\linewidth]{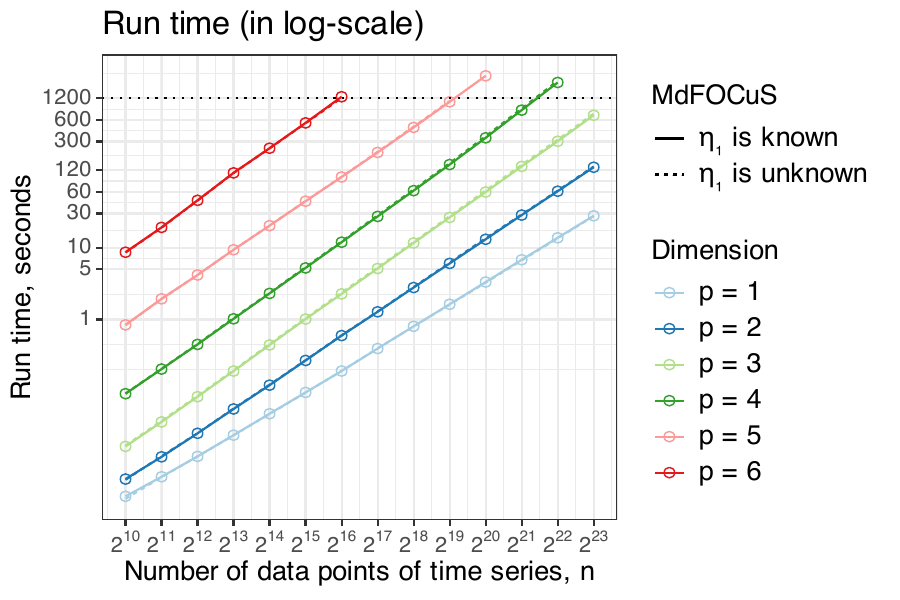}
     \caption{Run times in seconds of MdFOCuS with known (dashed line) and unknown (full line) pre-change parameter $\eta_1$, both with $\alpha = 2$ and $\beta = 1$, in dimension $p=1,\dots, 6$ for i.i.d Gaussian data. Run times are averaged over $100$ data sets. We considered a maximum running time of $20$ minutes (horizontal dotted black line) for the algorithms, and thus do not report run times for $p\geq4$ for large $n$. The slope of the curves estimated are 
     for known $\eta_1$: $\approx 1.005$ ($p=1$), $\approx 1.116$ ($p=2$), $\approx 1.180$ ($p=3$), $\approx 1.196$ ($p=4$), $\approx 1.143$ ($p=5$) and $\approx 1.202$ ($p=6$); 
     and for unknown $\eta_1$: $\approx 1.006$  ($p=1$), $\approx 1.118$ ($p=2$), $\approx 1.183$ ($p=3$), $\approx 1.202$ ($p=4$), $\approx 1.144$ ($p=5$) and $\approx 1.197$ ($p=6$).}
\label{GM_Runtime}
\end{figure*}



\subsection{A provably quasi-linear online implementation}\label{sec3_3_Dyadic}

We showed how Algorithm \ref{MdFOCuS_algo} is exact and empirically fast, but it is non-trivial to derive its worst or expected time complexity as the time at which we run {\texttt{QuickHull}} is stochastic. In Supplementary Material \ref{dyalic_section} we propose Algorithm~\ref{MdFOCuS_algo_dyadic} that is closely related but easier to analyse. This algorithm updates the set $\mathcal{T}$ 
in a deterministic manner on predetermined chunks of the observations. This deterministic choice of chunks allows for a precise quantification of the expected complexity of the algorithm 
using Corollary \ref{col:boundnumbervertexes}
and the assumption that the convex hull algorithm is at worst quadratic.
To be specific, here are our assumptions.
\begin{assumption}
\label{assump:Vn_complexity}
In any dimension $p\ge1$, there are positive constants $c_1$ and $c_2$ for which the expected number of vertices, $\mathbb{E}[V_n^p]$, can be bounded as follows:
\begin{equation*}
\mathbb{E}[V_n^p] \leq c_1 \log(n)^p + c_2.   
\end{equation*}
\end{assumption}
This is true under the conditions of 
Theorem \ref{th:boundnumbervertexes} and Corollary \ref{col:boundnumbervertexes}. Next, define $Time(n)$ to be the worst-case computational cost of finding the convex hull of $n$ points. 
\begin{assumption}
\label{assump:Time_complexity}
For our convex hull algorithm 
there exist two positive constants $c_3$ and $c_4$ such that  $Time(n) \leq c_3n^2 + c_4$. 
\end{assumption}
The assumption on the quadratic complexity is true for {\texttt{QuickHull}} for $p\le2$ \cite[]{Barber1996} and using the algorithm of \cite{Chazelle1993} it is true for $p\le3$. 

In essence, Algorithm~\ref{MdFOCuS_algo_dyadic} computes the convex hull using a divide-and-conquer strategy. Out of simplicity, let us consider the offline setting and assume $n-1$ is a power of $2$, that is $n-1=2^{\log_2(n-1)}$. Rather than running the hull algorithm on all points $\{P(\tau)\}_{\tau \in\{1, \ldots, n-1\}}$ we first run the algorithm on smaller chunks of size $2^{q_{min}}$ : $\tau \in \{2^{q_{min}}j+1, \ldots, 2^{q_{min}}j+ 2^{q_{min}} \},$ with $q_{min}$ a constant chosen by the user and $j$ a number smaller than $(n-1)/2^{q_{min}}$, and then dyadically merge the results. To be specific for each scale $q$ with $q_{min} < q < \log_2(n-1)$ we get the points on chunk  $\{2^{q+1}j+1, \ldots, 2^{q+1}j+ 2^{q+1} \}$, by merging the results of the two smaller chunks $\{2^{q+1}j+1, \ldots, 2^{q+1}j+ 2^{q} \},$ and  $\{2^{q+1}j+2^{q}+1, \ldots, 2^{q+1}j+ 2^{q+1} \}$. This merging is computationally efficient because, under Assumption \ref{assump:Vn_complexity}, the number of points on the hull of a chunk of size $2^q$ is expected to be small (of order $(\log(2^q))^p= (q\log(2))^p$).

One difficulty in Algorithm~\ref{MdFOCuS_algo_dyadic} is to make this dyadic or divide-and-conquer merging in an online fashion. To the best of our efforts, this makes the formal description a bit tedious, and we refer the interested reader to the Supplementary Material~\ref{dyalic_section} and Algorithm \ref{dyadic_online} and \ref{MdFOCuS_algo_dyadic} for a precise description. We now state our bound on the expected complexity of Algorithms \ref{dyadic_online} and \ref{MdFOCuS_algo_dyadic}.

\begin{lemma}\label{lemma:bound_update_dyadic}
Under 
Assumptions \ref{assump:Vn_complexity} and 
\ref{assump:Time_complexity}, the expected time complexity of Algorithm \ref{dyadic_online} and \ref{MdFOCuS_algo_dyadic} 
is $\mathcal{O}\left(n \log(n)^{p+1}\right)$.
\end{lemma}

We empirically study the complexity of Algorithm~\ref{MdFOCuS_algo_dyadic} in Figure~\ref{GM_DyadicRuntime} of Supplementary Material~\ref{dyalic_section}, where we see that it is slower than Algorithm~\ref{MdFOCuS_algo}. 
It is interesting to compare the bound on the computational cost with that of an oracle algorithm that knows the set of points on the convex hull at each iteration and whose computational cost is only evaluating the test statistic for each point on the hull. This would have a complexity of order $n\log^p(n)$, as $\sum_{t=1}^n \log^p(t) \geq n/2 \log^p(n/2)$. This is just a factor of $\log n$ faster than the bound in Lemma \ref{lemma:bound_update_dyadic}. 

\subsection{Two variant of MdFOCuS}
\label{sec3_3_AlgoMR}

In this subsection, we describe two variants of the algorithm, arising from variations in (i) the test statistic construction and (ii) the convex hull update. These variants are available regardless of whether or not the pre‐change parameter is known.

The first variant will allow MdFOCuS to improve on scenarios where only a subset of the coordinates change: out of simplicity in lines 7 and 8 of Algorithm~\ref{MdFOCuS_algo} we only consider the likelihood ratio statistics \eqref{eq:loglikelihood_eta}  that assumes all coordinates are changing. However, based on Theorem \ref{th:changepoint_andhull}, the changepoints maximising the sparse likelihood ratio statistics \eqref{eq:restricted_max}, where only $s$ coordinates are changing, also correspond to points on the hull, and we can compute the test statistics in this case by replacing lines 7 and 8 with the corresponding sparse likelihoods. 

To compute all sparse max likelihood ratio statistics 
we sort the per-dimension ratio statistics (this can be done in $\mathcal{O}(p\log(p))$) and sum these from the largest to the smallest in $\mathcal{O}(p)$ to iteratively recover the best for all sparsity levels $s$. We then recover the maximum for each sparsity level in $\mathcal{O}(\mathcal{T})$. If we consider all sparsity levels the total complexity is $\mathcal{O}(|\mathcal{T}|p\log(p))$.

In practice, it makes sense to consider a subset of all sparsity levels. In our simulations we consider: $s \in \{2^0, 2^1, 2^2, \ldots, 2^{\lfloor \log_2(p)\rfloor}, p \}$.
In a similar fashion to the FOCuS approximation, we also consider the sum of the max likelihood ratio statistics for each coordinate. 
As can be seen in Figure \ref{fig:runtime_heuristic} , the empirical run time of MdFOCuS with sparse statistics is 
comparable to that of  {\texttt{ocd}} \cite[]{chen2020highdimensional}.

The second variant extends the algorithm to high-dimensional settings, where a full convex hull update is computationally prohibitive. In fact, under the condition of Theorem \ref{th:boundnumbervertexes}, Algorithm \ref{MdFOCuS_algo} is expected to store $\mathcal{O}(\log^p(n))$ candidates and the $\log^p(n)$ term means that the algorithm is impractical for $p$ larger than $5,$ as exemplified in Figure \ref{GM_Runtime}. In Supplementary Material~\ref{heuristic_section}, we propose an approximate algorithm by recovering only a subset of the points on the hull. We obtain this subset by considering projections on at most $\tilde{p}$ variates, running the hull algorithm on each projection, and taking the union of the time-points (of the vertices) retained for each projection.
For each projection or group of variates, under the condition of Theorem \ref{th:boundnumbervertexes}, we store $\mathcal{O}(\log^{\tilde{p}}(n))$ candidates. Considering $O(p)$ such groups we store $\mathcal{O}(p\log^{\tilde{p}}(n))$ candidates, thus 
reducing the computational complexity.

This heuristic for $\tilde{p}=2$, implemented in \texttt{R}, runs on average in about 100 seconds for $p=100$ and $n=10000$. We further explore the run times of this heuristic in Figure~\ref{fig:runtime_heuristic}  
and show that its run time is comparable to {\texttt{ocd}}.
As shown in Section \ref{sec4_1_Simulation} it provides good statistical performance. 

\begin{figure*} [!t]
\includegraphics[width=\linewidth]{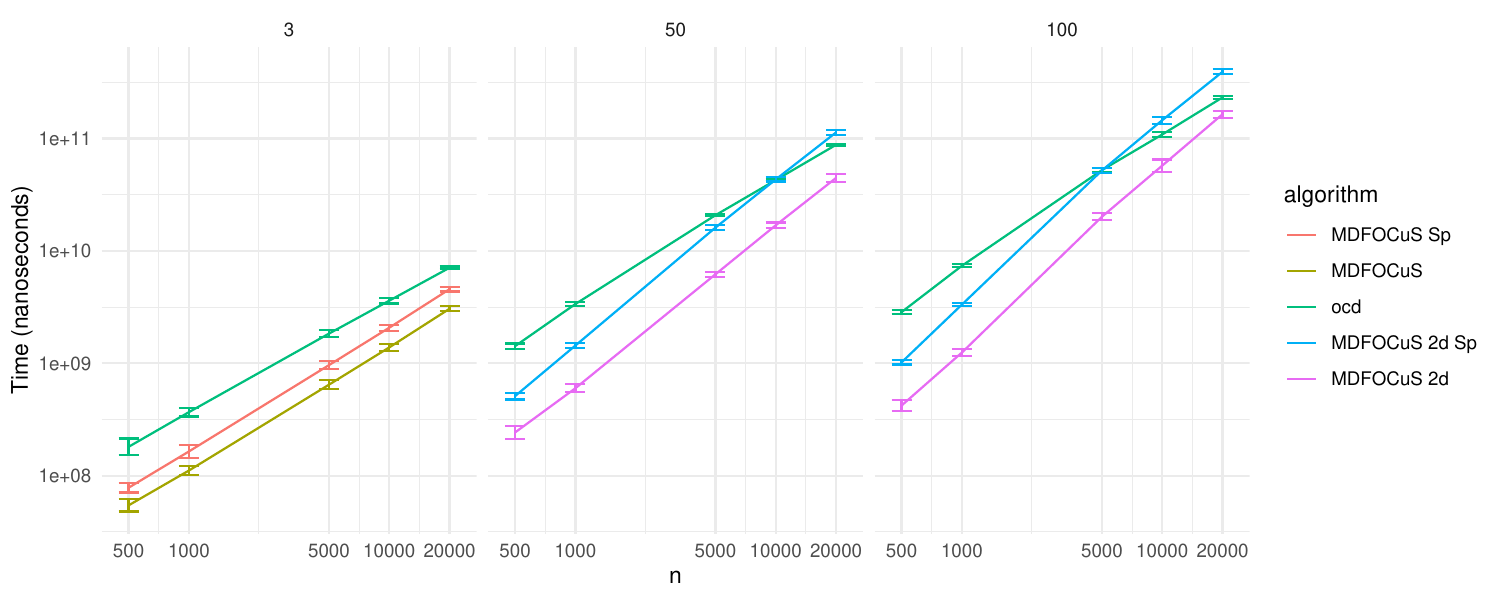}
\caption{Average run times of MdFOCuS, MdFOCuS with sparse statistics (Sp), and the 2$d$ hull approximation and \texttt{ocd} over 50 replicates for each scenario, for $p=3, 50$ and $100$ and $n$ varying from $500$ to $10000$ on an Intel Xeon Gold 6342 (32) @ 2.793GHz. }\label{fig:runtime_heuristic}
\end{figure*}

\subsection{Statistical Guarantees}\label{subsec:statguarantees}

In Supplementary Material~\ref{app:stat_control}, we provide several statistical guarantees for our procedure considering sparse or dense statistics in the multivariate Gaussian case. Our results also cover test statistics that threshold the contribution from each component.
We consider properties of our algorithm if there is no change and then if there is one. In the first scenario, we cover both the pre-change parameter being known or unknown,  while in the second, we only consider the pre-change parameter being known. We summarise the main results and their implications here.

We first consider the properties of our algorithm if there is no change. In this case, one may want to control either the false-positive rate, that is, the probability of ever detecting a change, or the average run length, the expected time until we detect a change. Our results give thresholds for both cases, but we will summarise them here just for an average run length of $\gamma$. At time $n$, we define the test statistic assuming a sparsity level $s$ and a candidate change at $\tau$ as $S^s_{\tau,n}$. 
Testing for $m$ different sparsity levels, and with the following threshold for sparsity $s$
\[
c_s(\gamma) = s + 2\zeta + 2\sqrt{s\zeta}, \mbox{ where } \zeta=4\log \gamma + s\log p + \log m + 5\log2
\]
our MdFOCus algorithm will have an average run length of at least $\gamma$ (this follows by applying the argument in Remark \ref{app-rem-ARL} to the result of Proposition \ref{coro:H0_all} with $N=2\gamma$ and $\alpha=1/2$). In the case where $s=p$ the $s\log(p)$ term in $\zeta$ can be omitted.

Next, we consider the power of detecting a change. This can be quantified in terms of the average detection delay, informally the expected time after a change that you would detect it, or bounds on the detection delay, the probability of detecting a change within some time-interval that it occurs. We provide results on both, but will summarise here only the latter. Let the change in mean be specified by the vector $\delta$. For a test at sparsity level $s$, denote by $\delta^{[s]}$ the vector of length $s$ that consists of the $s$ entries of $\delta$ with the largest absolute value. Finally, let $\tau^*$ denote the true changepoint and $T^{s}$ the time at which we detect a change using the $s$-sparse statistic. Then from Proposition \ref{coro:H1_ranked} and Remark \ref{rem:DDtos}, and using the threshold $c_s(\gamma)$, we have that with probability at least $(1-\alpha)$,
\begin{equation} \label{eq:DD}
T^{s}-\tau^* < \frac{2(c_s(\gamma)-s) + 2 \sqrt{s \log(3/\alpha)} - 8 \log(\alpha/3) }{\|\delta^{[s]}\|^2}.
\end{equation}

We can draw some insights from these results. Excluding $s=p$, first observe that $c_s(\gamma)$ is increasing with $s$. This means, unsurprisingly, that to minimise the bound on the detection delay $T^{s}-\tau^*$, we never want to use $s$ greater than the true sparsity of $\delta$, with the notable exception of $s=p$. More generally, it is helpful to consider two different regimes depending on the relative sizes of $\gamma$ and $p$. The first is a low-dimensional setting where $\log \gamma$ is much larger than $p$. In this case, the numerator of the right-hand side of (\ref{eq:DD}) is dominated by the $\log \gamma$ term in $c_s(\gamma)$, and will be approximately $8\log \gamma/\|\delta^{[s]}\|^2$. In this case, maximising $\|\delta^{[s]}\|^2$ is key to minimising the detection delay, and this can be guaranteed by using the dense test statistic, $s=p$. Any advantage of using smaller $s$, in the case that some components do not change, will be small.

The other regime is where $p$ is much larger than $\log \gamma$. In this case, using the dense statistics will give a detection delay with leading order (from the right-hand side of Equation \ref{eq:DD}) that is proportional to $\sqrt{p}/\|\delta\|^2$. By comparison, the $s$-sparse statistic will have a detection delay with leading order proportional to $(s\log p)/ \|\delta^{[s]}\|^2$. If only $s$ components change, and all change by the same amount, then this latter quantity will be lower order in $p$ when $s=o(\sqrt{p})$, as is consistent with other results on high-dimensional change detection \cite[e.g.][]{chen2020highdimensional}. Finally, we can see that the benefit of robustness of performance to the true level of sparsity that comes using $m=O(\log p)$ tests for differing sparsity levels is only a $\log\log p$ increase in the thresholds and hence in the detection delays.

\subsection{Empirical statistical performance}\label{sec4_1_Simulation}

In this section, we evaluate the empirical performance of the {\texttt{R}} implementation of MdFOCuS with the sparse statistics (as in Section \ref{sec3_3_AlgoMR}) on synthetic 3, 5, and 100-dimensional time series under the Gaussian change-in-mean setting, with an additional simulation study for a 50-dimensional time series presented in Supplementary Material~\ref{sec:approx_comparison}. 
We consider two scenarios: pre-change mean known and unknown. For each experiment, we report results aggregated from 300 data sets.

For MdFOCuS for dimension $p=3$ and $5$, we consider sparsity level $s=1,$ the full max-likelihood ratio (either $s=3$ or $s=5$ depending on the dimension), and the sum of the max likelihood ratio statistics per dimension as proposed in \cite{romano2022fast}. For higher dimensions ($p = 100$ and $50$), we use the approximate algorithm with $\tilde{p}=2$ that analyses 2$d$-projections of the data and set the sparsity levels as in Section \ref{sec3_3_AlgoMR}. We compare the performance of MdFOCuS with the high-dimensional online changepoint detection method {\texttt{ocd}} \cite[]{chen2020highdimensional} and a naive multivariate implementation of FOCuS, obtained by running FOCuS independently for each dimension and summing or taking the maximum of the resulting traces \citep[see Section~6 of][]{romano2022fast} at each iteration. Additionally, as a comparison, for the pre-change mean unknown case, we report results of {\texttt{ocd}} and the pre-change mean known  FOCuS with a plug-in estimate of the pre-change parameter, obtained over a probation (training) period of 250 observations.

Thresholds were selected under a Monte-Carlo simulation to achieve an average run length of $5000$ observations. 
An additional study controlling the false positive rate can be found in Supplementary Material~\ref{sec:approx_comparison}. For a change at $t = 1000$, we evaluate performances for a range of different change magnitudes and sparsity levels. 
To compare performances across different sparsity regimes, changes in the affected time series were normalised such that the change direction lies in the $p$-dimensional sphere, i.e., ($|| \eta_1 - \eta_2||^2_2 = \delta$), where $\delta$ is the overall change magnitude and $\eta_1, \eta_2$ are the pre- and post-change means, respectively.
Tables \ref{tab:dd3d} to Table \ref{tab:dd100} show results of the Average Detection Delays (ADDs) for 3, 5, and 100-dimensional time series, respectively. 

For all dimensions of time series, in both pre-change parameter known and unknown scenarios, we observe that MdFOCuS generally achieves smaller detection delays in cases where the change affects all the time series and over changes of larger magnitudes. In scenarios with sparsity set to 1, the simple FOCuS aggregation tends to be slightly better. This is expected as if we only have a change in one dimension, both the sparse statistics of FOCuS and MdFOCuS  will be the same, but MdFOCuS  will have a slightly larger threshold, resulting in a slower delay.

For the pre-change mean known case, over changes of smaller magnitudes, we find MdFOCuS  to be relatively close to the best competitor in terms of ADD, while it generally outperforms the other methods for larger changes. Furthermore, MdFOCuS is seen to have consistent performances across different sparsity levels. As {\texttt{ocd}} uses a grid for the possible size of change for each variate, its performance is affected by how close the true change size is to the nearest point on the grid. In particular, its performance deteriorates when the true change is outside the set of grid values. We observe this in its deterioration of performance for low magnitudes (0.125) and large magnitudes (1 or 2), in particular over dimensions $p=3, 5$. 

For the pre-change mean unknown scenario, we observe that MdFOCuS generally outperforms competitor methods for magnitudes larger than 0.250, with the exception of scenarios with a change in a single dimension, as explained above. Moreover, in dimensions $p=5$ and $100$, we can clearly see the advantages of performing the likelihood ratio test for a change when the pre-change mean is unknown over using a plug-in estimate for the pre-change parameter.
The only exception occurs in the case of FOCuS, when estimating the pre-change parameter from the training data, for $p=3$, which out-performs other methods for the smallest size of change. 
Whilst this might seem counter-intuitive, as the error from the estimation should propagate, in dimension $p=3$ for a small signal-to-noise ratio, the estimation is good enough to not impair the average run length and increase the thresholds, while benefiting from an easier optimisation (as we see in the pre-change mean known case). However, as we move to higher dimensions $p=5$ and $100$, it is unlikely to obtain a better estimate of the true pre-change parameter simultaneously in all dimensions. This results in larger thresholds to control the desired run length, and therefore in a global loss of power. Of course, having more training observations improves the results, in particular for $p=3$, however, this still has little to no effect in higher dimensions (see Supplementary Material~\ref{app:extra_training_data}).


\begin{table}[!h]
\centering
\caption{\label{tab:dd3d}Average detection delay with 3-dimensional time series for a change at time $t = 1000$. Pre-change parameter is known on the left, and unknown on the right. Best results per row are highlighted in bold. For the pre-change unknown case FOCuS (est 250) and ocd estimate the pre-change parameter from 250 training samples.}
\centering
\resizebox{\ifdim\width>\linewidth\linewidth\else\width\fi}{!}{
\fontsize{12}{14}\selectfont
\begin{tabular}[t]{rrrrrrrrr}
\toprule
\multicolumn{2}{c}{ } & \multicolumn{3}{c}{Pre-change Known} & \multicolumn{4}{c}{Pre-change Unknown} \\
\cmidrule(l{3pt}r{3pt}){3-5} \cmidrule(l{3pt}r{3pt}){6-9}
magnitude & sparsity & FOCuS  & MdFOCuS  & ocd & FOCuS (est 250) & FOCuS & MdFOCuS & ocd (est 250)\\
\midrule
0.125 & 1 & \textbf{786.67} & 806.54 & 1227.57 & \textbf{1470.76} & 1644.05 & 1726.24 & 1783.41\\
0.125 & 2 & \textbf{791.36} & 817.12 & 1421.54 & \textbf{1693.06} & 1770.33 & 1815.88 & 1987.04\\
0.125 & 3 & \textbf{859.03} & 886.83 & 1496.18 & \textbf{1723.76} & 1806.46 & 1853.35 & 2097.23\\
\addlinespace
0.250 & 1 & 250.58 & 255.78 & \textbf{227.18} & 470.69 & \textbf{340.46} & 364.50 & 430.33\\
0.250 & 2 & \textbf{253.56} & 256.76 & 268.86 & 556.86 & 369.36 & \textbf{363.55} & 516.28\\
0.250 & 3 & 273.92 & \textbf{266.27} & 300.69 & 630.94 & 404.19 & \textbf{395.40} & 559.66\\
\addlinespace
0.500 & 1 & 65.85 & 66.46 & \textbf{56.46} & 113.75 & 69.73 & \textbf{69.36} & 83.90\\
0.500 & 2 & 71.85 & 69.93 & \textbf{65.18} & 132.99 & 78.42 & \textbf{75.28} & 96.50\\
0.500 & 3 & 79.21 & 76.52 & \textbf{74.50} & 155.61 & 87.95 & \textbf{84.11 }& 103.97\\
\addlinespace
0.750 & 1 & 31.80 & 32.16 & \textbf{28.95 }& 50.73 & {32.62} & \textbf{33.22} & 39.46\\
0.750 & 2 & 34.66 & 33.31 & \textbf{32.52 }& 61.46 & 36.15 & \textbf{35.35} & 43.47\\
0.750 & 3 & 37.19 & \textbf{34.84} & 35.43 & 71.75 & 39.40 & \textbf{36.47} & 46.40\\
\addlinespace
1.000 & 1 & \textbf{18.51} & 19.04 & 18.87 & 29.18 & \textbf{19.06} & 19.53 & 24.16\\
1.000 & 2 & 20.30 & \textbf{19.53} & 20.45 & 34.55 & 20.94 & \textbf{20.42} & 25.45\\
1.000 & 3 & 21.68 & \textbf{20.27} & 21.52 & 40.03 & 22.49 & \textbf{20.72} & 26.23\\
\addlinespace
2.000 & 1 &\textbf{ 5.35} & 5.41 & 6.40 & 8.06 & \textbf{5.43 }& 5.53 & 7.12\\
2.000 & 2 & 5.97 & \textbf{5.72 }& 6.64 & 9.46 & 6.13 & \textbf{5.78 }& 7.25\\
2.000 & 3 & 6.51 & \textbf{6.09} & 6.80 & 10.62 & 6.71 & \textbf{6.16} & 7.41\\
\bottomrule
\end{tabular}}
\end{table}

\begin{table}[!h]
\centering
\caption{\label{tab:dd5}Average detection delay with 5-dimensional time series for a change at time $t = 1000$. Pre-change parameter is known on the left, and unknown on the right. Best results per row are highlighted in bold. For the pre-change unknown case FOCuS (est 250) and ocd estimate the pre-change parameter from 250 training samples.}
\centering
\resizebox{\ifdim\width>\linewidth\linewidth\else\width\fi}{!}{
\fontsize{12}{14}\selectfont
\begin{tabular}[t]{rrrrrrrrr}
\toprule
\multicolumn{2}{c}{ } & \multicolumn{3}{c}{Pre-change Known} & \multicolumn{4}{c}{Pre-change Unknown} \\
\cmidrule(l{3pt}r{3pt}){3-5} \cmidrule(l{3pt}r{3pt}){6-9}
magnitude & sparsity & FOCuS  & MdFOCuS & ocd  & FOCuS (est 250) & FOCuS & MdFOCuS & ocd (est 250)\\
\midrule
0.125 & 1 & \textbf{833.11} & 866.74 & 911.17 & 2219.97 & \textbf{1718.76} & 1802.03 & 2458.66\\
0.125 & 2 & \textbf{892.82} & 918.78 & 1106.68 & 2221.24 & \textbf{1874.47} & 2049.54 & 2441.70\\
0.125 & 3 & \textbf{925.08} & 940.54 & 1101.38 & 2255.65 & \textbf{1924.89} & 2036.49 & 2493.92\\
0.125 & 4 & \textbf{835.44} & 836.87 & 1059.06 & 2102.15 & \textbf{1825.65} & 1881.13 & 2296.26\\
0.125 & 5 & 953.85 &\textbf{ 948.59} & 1297.37 & 2228.96 &\textbf{ 2005.88} & 2062.98 & 2448.01\\
\addlinespace
0.250 & 1 & 259.75 & 267.12 & \textbf{229.64} & 871.35 & \textbf{360.65} & 368.47 & 1064.11\\
0.250 & 2 & 281.77 & 279.88 & \textbf{253.79} & 963.88 & \textbf{433.69} & 450.35 & 1084.21\\
0.250 & 3 & 297.41 & 289.83 & \textbf{287.10} & 1004.43 & \textbf{456.72} & 467.69 & 1077.31\\
0.250 & 4 & 260.23 & \textbf{242.46} & 265.40 & 918.63 & 413.22 & \textbf{382.42} & 852.33\\
0.250 & 5 & 304.13 & \textbf{288.07} & 309.16 & 1023.24 & 504.75 & \textbf{495.38} & 1000.05\\
\addlinespace
0.500 & 1 & 72.98 & 73.81 & \textbf{64.74} & 206.83 & \textbf{80.54} & 80.69 & 247.08\\
0.500 & 2 & 84.75 & 80.83 & \textbf{73.93 }& 243.50 & 95.86 & \textbf{91.35} & 253.56\\
0.500 & 3 & 91.31 & 84.41 & \textbf{83.54} & 267.82 & 103.57 & \textbf{96.62} & 252.69\\
0.500 & 4 & 80.89 & \textbf{71.59} & 75.94 & 250.83 & 94.23 & \textbf{81.13} & 199.82\\
0.500 & 5 & 95.41 &\textbf{ 86.77} & 89.48 & 293.72 & 110.49 & \textbf{100.25} & 235.87\\
\addlinespace
0.750 & 1 & 33.44 & 34.02 & \textbf{32.98} & 91.82 & \textbf{34.46 }& 35.00 & 108.72\\
0.750 & 2 & 39.21 & 37.79 & \textbf{36.22} & 113.10 & 40.34 & \textbf{39.06 }& 113.74\\
0.750 & 3 & 41.29 &\textbf{ 38.16} & 38.61 & 124.88 & 44.27 & \textbf{41.49} & 110.91\\
0.750 & 4 & 37.03 &\textbf{ 33.10} & 34.71 & 116.62 & 41.61 & \textbf{35.94 }& 87.09\\
0.750 & 5 & 45.65 & \textbf{40.26} & 42.24 & 137.55 & 50.88 & \textbf{44.03} & 104.84\\
\addlinespace
1.000 & 1 & \textbf{19.63} & 19.84 & 21.37 & 51.75 & \textbf{20.02} & 20.30 & 62.33\\
1.000 & 2 & 22.88 & \textbf{21.86 }& 22.88 & 63.86 & 23.55 & \textbf{22.49} & 62.99\\
1.000 & 3 & 24.67 & \textbf{22.83} & 23.65 & 71.35 & 25.38 & \textbf{23.47} & 62.24\\
1.000 & 4 & 22.45 & \textbf{19.29} & 20.32 & 67.41 & 23.76 & \textbf{20.20} & 49.88\\
1.000 & 5 & 27.18 & \textbf{23.63} & 25.07 & 81.02 & 28.62 & \textbf{24.79} & 61.01\\
\addlinespace
2.000 & 1 & \textbf{5.79} & 5.84 & 7.25 & 13.40 & \textbf{5.73} & 5.82 & 15.78\\
2.000 & 2 & 6.88 & \textbf{6.58 }& 7.47 & 16.57 & 6.86 & \textbf{6.53} & 15.66\\
2.000 & 3 & 7.21 & \textbf{6.55 }& 7.33 & 18.62 & 7.29 &\textbf{ 6.72 }& 15.74\\
2.000 & 4 & 6.64 & \textbf{5.72} & 6.10 & 17.92 & 6.92 & \textbf{5.90 }& 12.99\\
2.000 & 5 & 7.91 & \textbf{6.87} & 7.40 & 21.30 & 8.24 &\textbf{ 7.08} & 15.55\\
\bottomrule
\end{tabular}}
\end{table}

\begin{table}[!h]
\centering
\caption{\label{tab:dd100}Average detection delay with 100-dimensional time series for a change at time $t = 1000$. Pre-change parameter is known on the left, and unknown on the right. Best results per row are highlighted in bold. For the pre-change unknown case FOCuS (est 250) and ocd estimate the pre-change parameter from 250 training samples.}
\centering
\resizebox{\ifdim\width>\linewidth\linewidth\else\width\fi}{!}{
\fontsize{12}{14}\selectfont
\begin{tabular}[t]{rrrrrrrrr}
\toprule
\multicolumn{2}{c}{ } & \multicolumn{3}{c}{Pre-change Known} & \multicolumn{4}{c}{Pre-change Unknown} \\
\cmidrule(l{3pt}r{3pt}){3-5} \cmidrule(l{3pt}r{3pt}){6-9}
magnitude & sparsity & FOCuS  & MdFOCuS & ocd  & FOCuS (est 250) & FOCuS & MdFOCuS & ocd (est 250)\\
\midrule
0.125 & 1 & 1417.15 & 1391.54 & \textbf{1171.63} & 3459.08 & 2708.45 & \textbf{2704.65} & 3546.43\\
0.125 & 5 & 1881.10 & 1810.57 & \textbf{1627.00} & 3497.94 & \textbf{2859.51} & 2872.20 & 3595.24\\
0.125 & 10 & 2084.84 & 1978.85 & \textbf{1940.41} & 3501.48 & \textbf{2846.45} & 2874.17 & 3595.46\\
0.125 & 50 & \textbf{2252.48} & 2289.56 & 2386.98 & 3494.67 & 2888.46 & \textbf{2881.59} & 3591.96\\
0.125 & 100 & \textbf{2323.59} & 2338.00 & 2431.51 & 3507.66 & \textbf{2859.11} & 2887.60 & 3599.24\\
\addlinespace
0.250 & 1 & 389.11 & 380.09 & \textbf{335.95} & 2659.50 & 730.54 & \textbf{720.45} & 2841.37\\
0.250 & 5 & 641.18 & 555.57 & \textbf{505.39} & 3054.36 & 1520.80 & \textbf{1451.53} & 3232.02\\
0.250 & 10 & 789.27 & 673.15 & \textbf{595.10} & 3178.29 & 1879.64 & \textbf{1692.20} & 3334.21\\
0.250 & 50 & 1057.30 & 895.53 & \textbf{827.30} & 3323.17 & 2353.01 & \textbf{2309.73} & 3453.27\\
0.250 & 100 & 1080.26 & 916.64 & \textbf{891.09} & 3346.50 & \textbf{2341.24} & 2414.38 & 3476.58\\
\addlinespace
0.500 & 1 & 105.79 & 107.89 & \textbf{104.79} & 721.53 & \textbf{118.84} & 121.50 & 859.71\\
0.500 & 5 & 169.93 & 147.36 & \textbf{139.23} & 1297.63 & 214.14 & \textbf{185.07} & 1449.56\\
0.500 & 10 & 232.04 & 176.44 & \textbf{170.03} & 1653.97 & 296.51 & \textbf{217.84} & 1850.22\\
0.500 & 50 & 357.90 & 248.41 & \textbf{239.08} & 2480.88 & 578.17 & \textbf{348.83} & 2719.69\\
0.500 & 100 & 372.34 & 254.51 & \textbf{244.07} & 2611.86 & 679.19 & \textbf{365.79} & 2841.33\\
\addlinespace
0.750 & 1 & \textbf{47.78} & 48.67 & 57.56 & 310.89 & \textbf{50.08} & 50.68 & 461.51\\
0.750 & 5 & 75.34 & \textbf{66.80} & 73.47 & 580.05 & 85.34 & \textbf{73.90} & 715.73\\
0.750 & 10 & 104.45 & \textbf{80.81} & 84.16 & 798.60 & 117.25 & \textbf{88.30} & 925.56\\
0.750 & 50 & 184.09 & 116.47 & \textbf{114.29} & 1626.21 & 239.73 & \textbf{133.60} & 1836.28\\
0.750 & 100 & 193.77 & 120.24 & \textbf{115.21} & 1823.30 & 279.42 & \textbf{137.98} & 2074.96\\
\addlinespace
1.000 & 1 & \textbf{27.12} & 27.55 & 38.67 & 173.04 & \textbf{27.76} & 28.41 & 317.49\\
1.000 & 5 & 43.72 & \textbf{39.02} & 46.07 & 325.28 & 46.44 & \textbf{41.45} & 476.20\\
1.000 & 10 & 60.36 & \textbf{47.40} & 50.77 & 458.62 & 63.85 & \textbf{49.35} & 599.74\\
1.000 & 50 & 109.11 & 67.33 & \textbf{66.02} & 1074.41 & 135.79 & \textbf{73.06} & 1218.93\\
1.000 & 100 & 116.60 & 69.23 &\textbf{ 67.59} & 1288.45 & 157.47 & \textbf{75.55} & 1466.28\\
\addlinespace
2.000 & 1 & \textbf{7.58} & 7.65 & 12.93 & 43.03 & \textbf{7.53} & 7.73 & 142.00\\
2.000 & 5 & 12.92 & \textbf{11.12} & 13.80 & 82.41 & 12.95 & \textbf{11.33} & 203.48\\
2.000 & 10 & 16.46 & \textbf{13.16} & 14.80 & 117.67 & 16.30 & \textbf{13.48} & 250.44\\
2.000 & 50 & 31.85 & 18.37 & \textbf{18.08} & 315.56 & 34.85 & \textbf{18.70} & 451.59\\
2.000 & 100 & 34.28 & 18.97 & \textbf{18.34} & 422.46 & 41.93 & \textbf{19.37} & 493.39\\
\bottomrule
\end{tabular}}
\end{table}

\section{Application to NBA Plus-Minus Data}\label{sec4_2_NBA}
In this section, we illustrate our methodology to detect a change in the Cleveland Cavaliers (an NBA team) Plus-Minus score from season 2010-11 to 2017-18. This example was proposed in \cite{shin2022detectors}.
We use this example for two reasons.
    First, it highlights that our computational framework is also applicable to CUSUM Exponential baseline e-detectors proposed in \cite{shin2022detectors} which in essence are a non parametric generalization of generalized likelihood ratio rules.
    Second, it illustrates how modelling not only the signal's mean but also its variance can simplify the calibration of the detection threshold, without necessarily resorting to a non-parametric approach.

We recovered the Plus-Minus score using the \texttt{nbastatR} package \cite[]{nbastatr} available at \url{https://github.com/abresler/nbastatR}. The Cavaliers Plus-Minus is represented in Figure \ref{fig:NBA_2010_2018}, and is a time series of the difference in score of the Cavaliers and their opponents over successive games.
We consider detecting a changepoint in this data online.
\begin{figure*}[!ht]%
\begin{center}
\includegraphics[width= 0.75\linewidth]{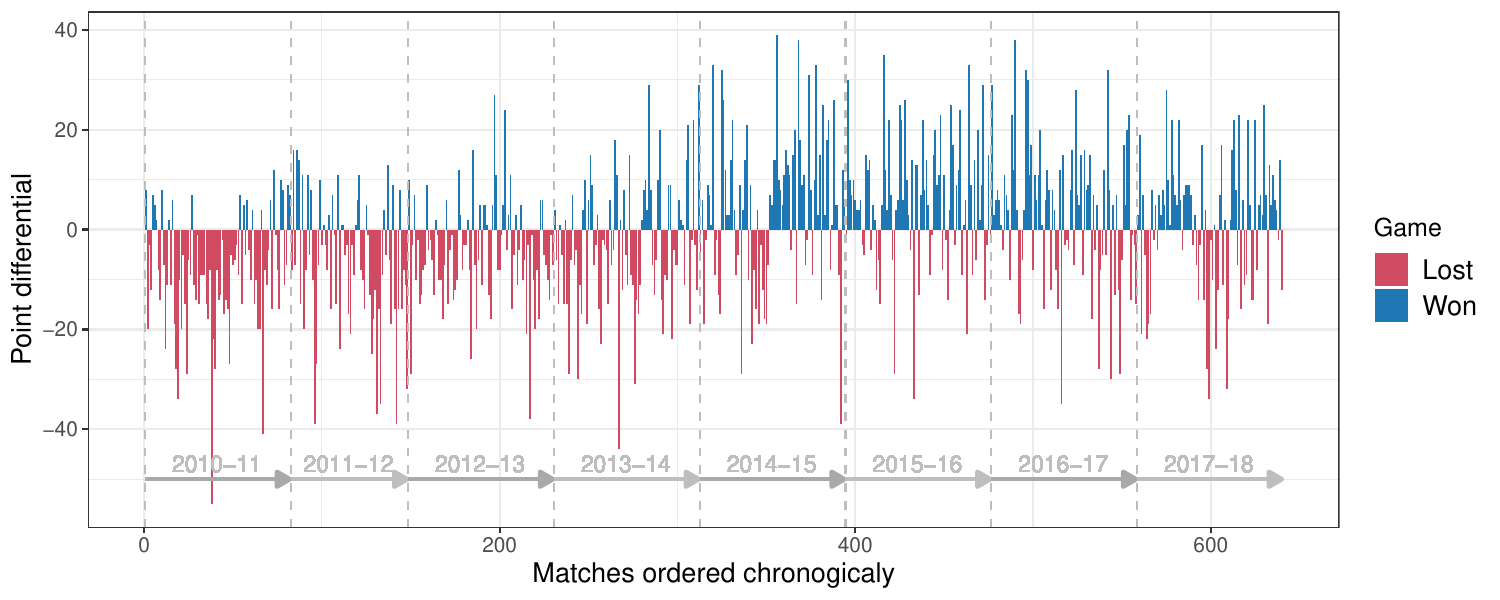}
\caption{
Plus-Minus score of the Cavaliers from season 2010-11 to season 2017-18 as a bar plot. Losses are in dark red and wins are in dark blue. Ends of seasons are represented with vertical dashed lines and seasons are shown on the bottom with grey arrows. 
}
\label{fig:NBA_2010_2018}
\end{center}
\end{figure*}

As detailed in \cite{shin2022detectors}, it is nontrivial to fit this problem into commonly used parametric sequential changepoint detection procedures for two main reasons. 
    First, it is not easy to choose a parametric model to fit the score. In particular, we would like to take into account that we have integer-valued data, and we expect some form of dependency.
    Second, assuming we have chosen a model or likelihood function, calibrating the threshold to detect a changepoint is not trivial. A small change in the distribution we use as a reference can substantially change the threshold and as a consequence increase the detection time or lead to false detections.
E-detectors \cite[]{shin2022detectors}, are a non-parametric solution to these two problems.

\subsection{Convex hull for CUSUM e-detectors} \label{sec4_connect_CH_edetectors}
Following Equations (27) and (41) of \cite{shin2022detectors} 
we can write their update of the e-detector CUSUM statistics for a given $\lambda$ as
\begin{equation*}
M_n^{CU}(\lambda) = \exp{(\lambda s(x_n) - \psi(\lambda)v(x_n))} \max\{ M_{n-1}^{CU}(\lambda), 1 \},
\end{equation*}
where $\lambda$ is in a subset of $\mathbb{R}$, $s$ and $v$ are real-valued functions, and $\psi$ is a finite and strictly convex function. Taking logs and considering all possible changepoints prior to $n$ we recover with induction that
\begin{equation*}
\log(M_n^{CU}(\lambda)) = \max_{\tau < n} \left\{\sum_{t=\tau+1}^n(\lambda s(x_t) - \psi(\lambda)v(x_t)) \right\}.
\end{equation*}

By subtracting $\sum_{t=1}^n (s(x_t) \lambda - \psi(\lambda)v(x_t))$ in the previous bracket, we do not change the maximiser, and we recover our functional pruning framework of Equation~\eqref{eq:FP_equation} with $a_\tau = \sum_{t=1}^{\tau} v(x_t)$, and $2 b_\tau=\sum_{t=1}^{\tau} s(x_t)$.
Applying Theorem~\ref{th:convexhull} we get that the changepoint optimising the e-detector CUSUM score for any $\lambda$ value
corresponds to a vertex on the convex hull of 
\begin{equation}
\label{eq:edetect-hull}
\left(\sum_{t=1}^{\tau} v(x_t),\frac{1}{2}\sum_{t=1}^{\tau} s(x_t)\right)_{\tau <n}.  
\end{equation}
As argued in \cite{shin2022detectors}  it makes sense to consider a mixture of e-detectors, that is different $\lambda$ values. For computational reasons \cite[]{shin2022detectors} considered a finite number of values, chosen 
using their \texttt{computeBaseLine} function. 

We will now compare the baseline numbers of two e-detectors for Plus-Minus score to the number of points on the hull of $\eqref{eq:edetect-hull}$.
\cite{shin2022detectors} considered two e-detector models. In the first, which we call Winning rate, they model the winning rate and consider $v(x_t) = 1$ and $s(x_t) = \mathbbm{1}_{x_t > 0} - p_0$, where $\mathbbm{1}$ is the indicator function and $p_0$ is taken to be $0.49$ in the numerical application. With \texttt{computeBaseLine} they selected 69 $\lambda$ values.
In their second e-detector model, which we call Plus-Minus, they analyse the Plus-Minus score more directly. The score is first normalized between 0 and 1 as $x'_t = (x_t + 80)/160$. Then they define $s(x') = (x'/m - 1)$ and $v(x) = (x'/m - 1)^2$, where $m$ is chosen to be $0.494$ in the numerical application. With \texttt{computeBaseLine} they selected 190 $\lambda$ values. 
For both 
e-detectors we computed the number of points in the convex hull for $n$ between $10$
and $n=640$ (the last match of season 2017-18) and obtained for Winning rate at most $16$ and for Plus-Minus at most $18$ points. 
Therefore out of the, respectively, $69$ and $190$ candidates suggested by \texttt{computeBaseLine} many are redundant as they correspond to the same changepoint.

\subsection{A parametric approach to the Plus-Minus score}\label{sec:data_definition}

Considering a change in the mean of a Gaussian model to fit the Plus-Minus score of the Cavaliers does not take into account the discrete nature of the data and possible dependencies. It is also not clear how one could get a valid estimation of the pre-changepoint parameter. Finally, as argued in \cite{shin2022detectors} calibrating the threshold to detect a changepoint is not trivial. 
As a robust parametric solution to these problems, we consider a change in the mean and variance of a univariate Gaussian signal with unknown pre-changepoint parameters.

For a change in mean and variance using our convex hull framework \eqref{eq:FP_equation} we need to keep track of the vertices on the hull of a 3-dimensional set of points
\begin{equation}\label{eq:hull_mean_and_variance}
\left(\tau, \sum_{t=1}^{\tau} x_t, \sum_{t=1}^{\tau} x_t^2\right)_{\tau < n}.
\end{equation}

\noindent Interestingly, this mean and variance model is invariant to shift and scaling. Based on that, we see that all the points on the convex hull of the Plus-Minus e-detector are on the hull of the change in mean and variance model. 

As the data is discrete, two (or more) consecutive scores may be exactly equal. For segments containing equal scores, the estimated variance would be $0$, and the likelihood would blow up to infinity.
To get around this problem we considered a minimum value for the variance, which we set to $1$ (chosen as substantially below the variance seen for plus-minus statistics for NBA teams).

In Figure \ref{fig:NBA_cumsum_fig} we represent the CUSUM statistics for the change in mean and change in mean and variance model. Although the scales are quite different, the two statistics have very similar trends and importantly both seem to increase around season 2013-14 (starting at match number 312 represented by a vertical dotted line). Although both models reconstruct the same changepoint, we will see in the next subsection that the mean-variance model shows less variability in terms of detection delay, achieving faster detections once a threshold is set to achieve a fixed run length.

\begin{figure*}[!t]%
    \centering

        \centering
        \includegraphics[width=0.35\linewidth]{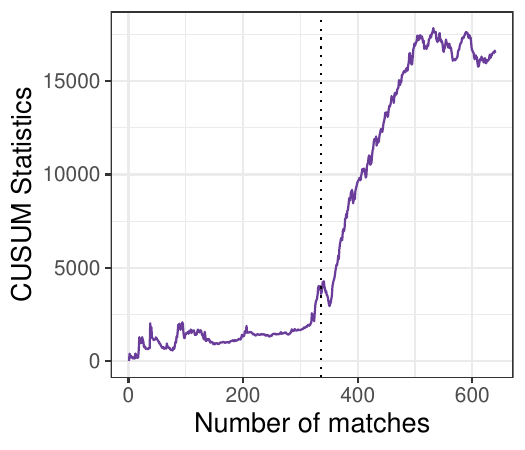}
        \centering
        \includegraphics[width=0.35\linewidth]{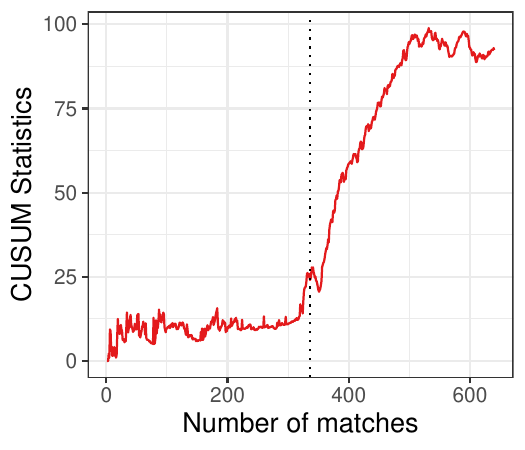}

    \caption{CUSUM statistics for the mean model (left) for the mean and variance model (right). The dotted line corresponds to the beginning of season 2014-15 and match number 312.}\label{fig:NBA_cumsum_fig}
\end{figure*}

Interestingly, at any time-point after 320 (8th match of season 2014-15) the maximum likelihood changepoint is always between matches 279 and 280 (5th and 7th of February 2014)  for both models. 
This is well before the beginning of season 2014-15 and the comeback of Lebron James. However, this matches the arrival of David Griffin as general manager on the 6th of February \cite[]{wiki_Griffin} .



\begin{figure*}[t]
\begin{center}
\includegraphics[width= 0.75\linewidth]{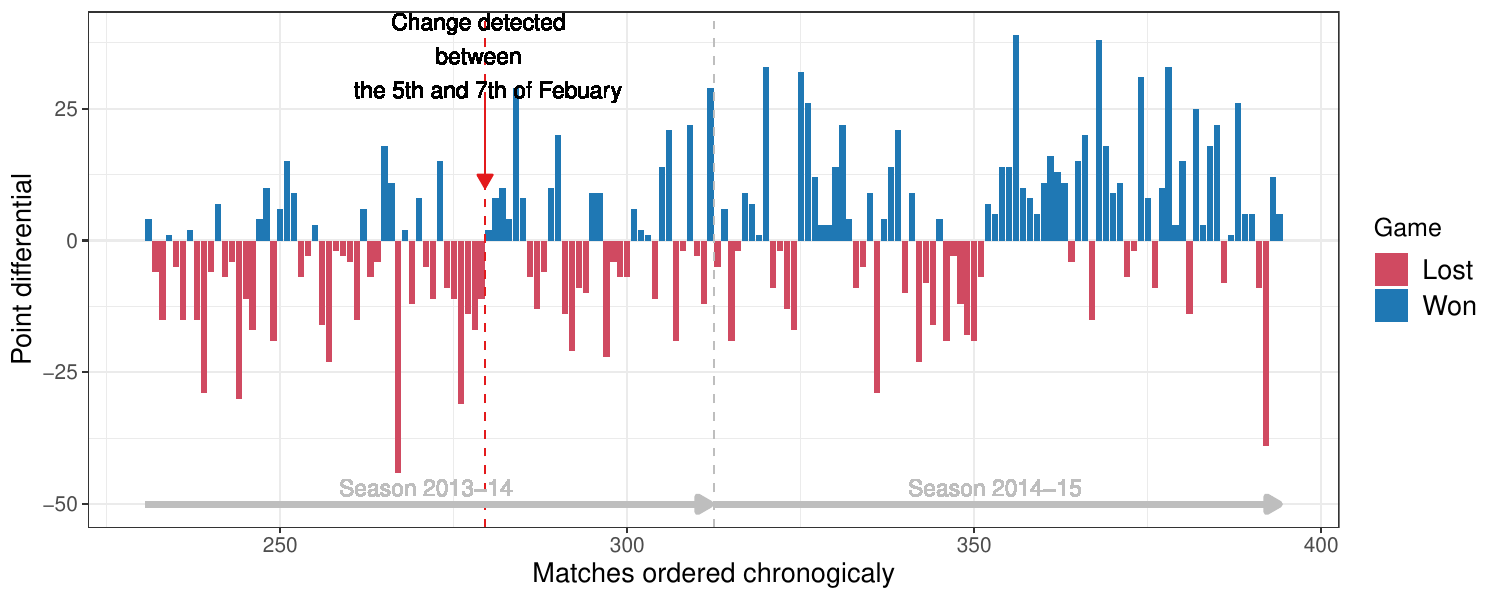}
\caption{
Plus-Minus score of the Cavaliers from seasons 2013-14 to 2014-15 as a bar plot. Losses are in dark red and wins are in dark blue. The ends of seasons are represented with vertical dashed lines (grey) and seasons are shown on the bottom with grey arrows.
}\label{fig:NBA_2013_2014}
\end{center}
\end{figure*}

\subsection{Simulating NBA Plus-Minus scores without changepoint}

Based on Figure \ref{fig:NBA_cumsum_fig} it is clear that for a particular range of the threshold both the mean and mean and variance model with unknown pre-changepoint parameter can detect a changepoint after the beginning of season 2014-15. To calibrate this threshold one typically needs a simulation model. Below we consider different choices of simulation model to see how sensitive the threshold is to this choice for both models. In all cases we will set the CUSUM threshold such that under our null simulation over 1000 matches we would detect a change only in 5\% of the cases.


The simulation models we considered are: (i) the difference of two independent Poisson random variables; (ii) the difference of two negative binomial random variables; (iii) a uniform distribution on $\{-80,\ldots,80\}$ \cite[as suggested in][]{shin2022detectors}; (iv) an independent resampling of past games. For (i) and (ii) we consider a range of parameters, with the mean of the random variables set to be 95, 105 or 115; and with the negative-binomial scale parameter set to 1 or 2. For (iv) we considered either sampling past matches of the Cavaliers or of NBA teams in general, from seasons 1999/2000 to 2009/10. Finally we also considered version of the scenarios in (i) and (ii) where we introduced temporal dependencies by adding to the mean a time-varying function either simulated from an AR(1) model with auto-correlation of 0.6 or using a sin function.

\begin{figure*}[t]%
        \centering
        \includegraphics[width=0.35\linewidth]{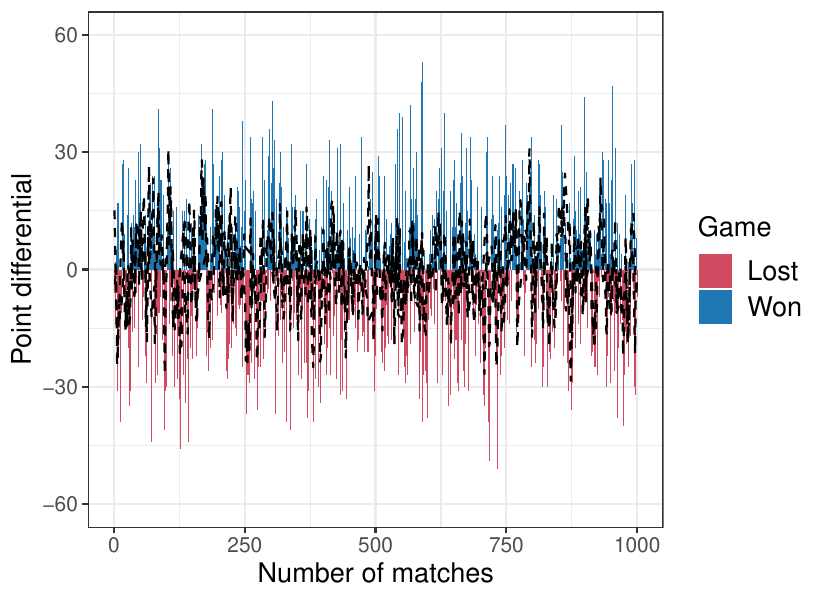}
        \includegraphics[width=0.35\linewidth]{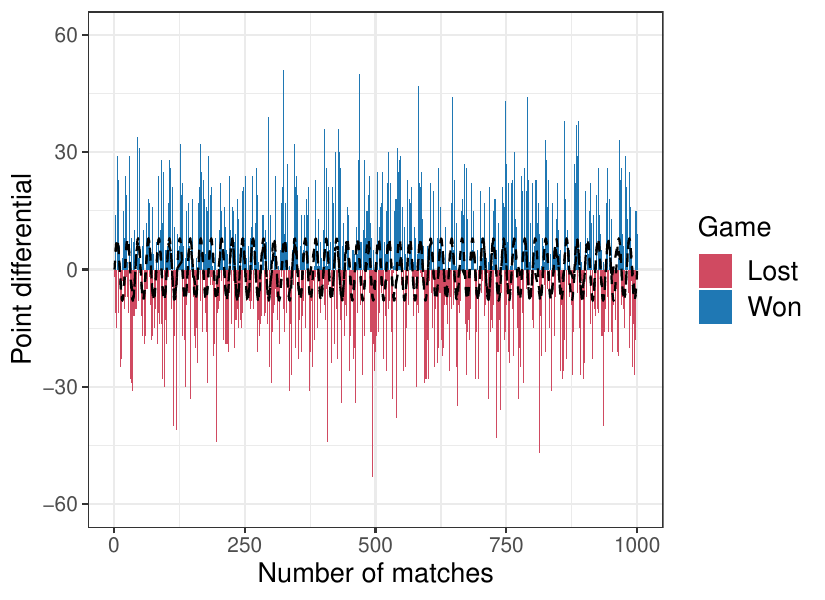}

    \caption{Two Plus-Minus score simulation examples based on the Negative Binomial model with a mean of $115$ and scale parameter of $50$ using two dependency scenarios: an AR(1) model with AR parameter $0.6$ (left) and using sine waves with a period of 20 matches (right).}\label{fig:example_sim}
\end{figure*}

Figure \ref{fig:threshold_and_detectime} shows the various thresholds we get for our two models, and the corresponding detection times. 
The threshold is much more variable for the change in mean model and this translates to a much more variable detection time. To be specific for the mean model, the change is detected at the earliest at match 330 (the 5th of December 2014), but in $21.7\%$ of the cases the threshold is so high that the changepoint is never detected (reported in the histogram as a detection at $1000$). For the mean and variance model, the change is always detected at the earliest at match 339 (the 23rd of December 2014) and at the latest at match 368 (the 20th of February 2015). In fact in $94.2\%$ of the simulation models the change is detected before with the mean and variance model.

\begin{figure*}[t]%
    \centering
        \includegraphics[width=0.4\linewidth]{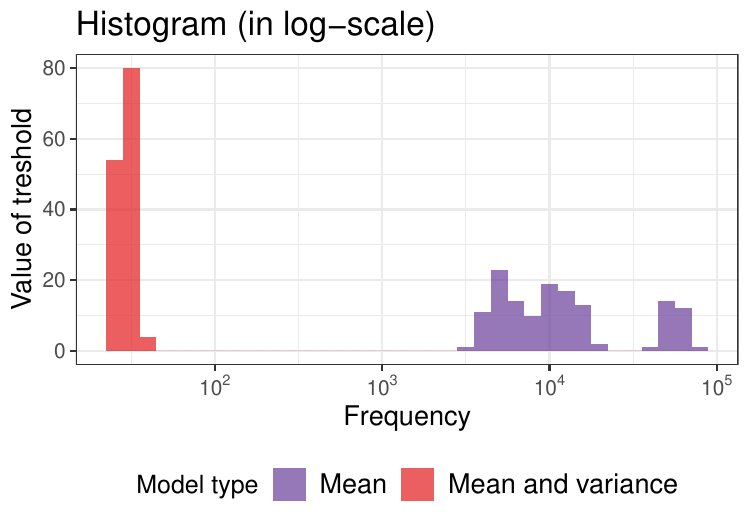}
        \includegraphics[width=0.4\linewidth]{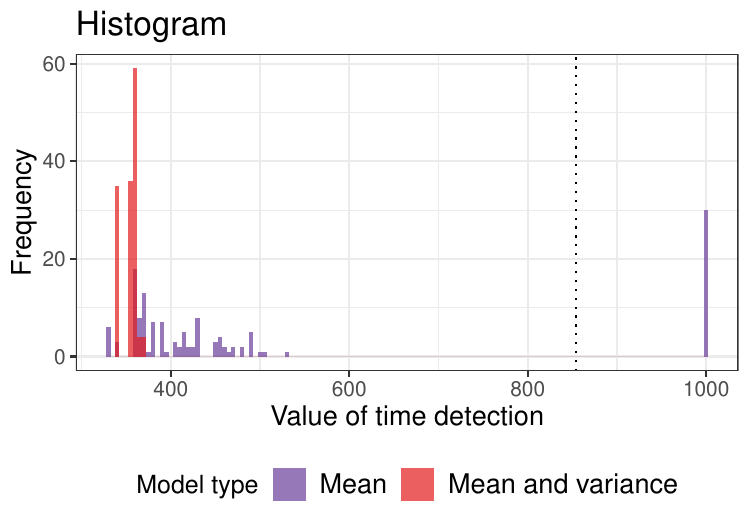}

    \caption{Threshold (left) and detection time (right) obtained for the mean model and the mean and variance model.}\label{fig:threshold_and_detectime}
\end{figure*}

\section{Discussion}
\label{sec_Discussion}

In this paper, we establish a connection between the detection of a single changepoint in $p$-dimensional data streams and a convex hull problem in dimension $p+1$. 
Based on this, we derive a simple yet efficient algorithm, MdFOCuS. This algorithm finds a set of time-points that contain the points associated with all vertices of a convex hull, and that we show contains all possible times that can maximise the likelihood-ratio statistic for a change. The method is general in that in can be used for a wide-range of models from the exponential family, including multivariate data where different streams are modelled by different distributions. It can also deal with the pre-change parameter being known or unknown, and covers tests where we allow only a subset of components of a multivariate data stream change. For all these varieties, the only difference is the part of the algorithm which maximises the function associated with each time-point we have stored.  We also showed that it can be used to speed up the calculation of e-detectors.
MdFOCuS is efficient for $p\leq 5$. For larger $p$ we provide an approximation storing only  $\mathcal{O}(p\log^{\tilde{p}}(n))$ vertices of the convex hull where  $\tilde{p} \leq p.$ 

Improvements to the computational complexity of the procedure could be made to allow the procedure to scale to higher dimensions and longer sequences. For example, a sequential update of the boundary of the convex hull of our points, to improve on the complexity of \texttt{QuickHull}. 
Recently, \cite{ward2023constantperiteration} offered an approach to speed-up the maximisation step. This aims to recycle calculations to give a bound on the test statistic that is much more efficient to calculate. Whilst this idea can directly be applied in the multivariate settings we consider, the bounds are too loose to give any noticeable speed-up. A less conservative bound could be the object of future study. 

More generally, the ideas in this paper which show a connection between changepoint detection and convex hull (or half-space intersection) problems sheds light on the geometric nature of the changepoint detection problem and is a promising line for future research.

\bibliographystyle{abbrvnat}
\bibliography{reference}

\newpage
\begin{appendices}

\section*{Supplementary Material for `Multivariate Maximum Likelihood Changepoint Detection via Half-Space Intersection and Convex Hull'}

\section{Overview of Code}

The implementation of the MdFOCuS algorithm, the Plus-Minus score data, the computer code for the simulation study, and the applications discussed in this article can be found in the following GitHub repositories:
\begin{itemize}
    \item \url{https://github.com/guillemr/Focused/tree/main} is a link to the \texttt{R} implementation of the MdFOCuS 
    and code for simulation studies.
    \item \url{https://github.com/lpishchagina/focus} is a link to the \texttt{R}/\texttt{C++} implementation of the MdFOCuS algorithm.
    \item \url{https://github.com/grosed/changepoint_online} is a link to the Python implementation of the MdFOCuS algorithm.
    \item \url{https://github.com/lpishchagina/dyadicMdFOCuS} is a link to the \texttt{R}/\texttt{C++} implementation of the MdFOCuS algorithm with dyadic update. It is a different repository from the MdFOCuS algorithm because the dyadic version is interesting mostly from a theoretical perspective to bound the expected complexity but empirically much slower than MdFOCuS;
    \item \url{https://github.com/abresler/nbastatR} is the \texttt{R} package \cite{nbastatr} we used to recover the Plus-Minus score data.
\end{itemize}

\section{Proofs}

\subsection{Proof of Theorem \ref{th:equality}}\label{app:proof_Th2}
Consider that $\tau$ is in $\mathcal{G}_{\mathcal{T}}$. By definition, there exists $(\lambda_0,\mu_0) \in Im(A)\times\mathcal{D}_A =  \mathbb{R}^+ \times \RR^p$ such that for all $\tau' \neq \tau$ in $\mathcal{T}$  we have: $g_\tau(\lambda_0, \mu_0) > g_{\tau'}(\lambda_0, \mu_0)$. Multiplying by $\alpha >0$ we get by linearity
\begin{equation*}
g_\tau(\alpha \lambda_0, \alpha \mu_0) = \alpha  g_\tau(\lambda_0,  \mu_0) > \alpha g_{\tau'}( \lambda_0,  \mu_0) = g_{\tau'}(\alpha \lambda_0, \alpha \mu_0).
\end{equation*}
Therefore, the optimally of index $\tau$ is true on the half-line $\{\alpha(\lambda_0,\mu_0), \alpha > 0\}$. Note that we can exclude the cases $\|\mu_0\| = 0$ or $\lambda_0 = 0$ as, by a continuity argument, the optimal set for index $\tau$ in  $\mathcal{G}_{\mathcal{T}}$ (if non empty) is an open set. 

We now search for $\alpha_0 \in (0, +\infty)$ such that for any index $\tau'' \ \in \ \mathcal{T}$, $g_{\tau''}(\alpha_0 \lambda_0, \alpha_0 \mu_0) = f_{\tau''}(\alpha_0 \mu_0)$. If such $\alpha_0$ exists, we have
\begin{equation*}
a_{\tau''} \alpha_0\lambda_0 - 2 \langle \alpha_0\mu_0, b_{\tau''} \rangle = a_{\tau''} \|\alpha_0\mu_0\|^q - 2 \langle \alpha_0\mu_0, b_{\tau''} \rangle\,,
\end{equation*}
leading to solution $\alpha_0 = \frac{\lambda_0^{\frac{1}{q-1}}}{\|\mu_0\|\frac{q}{q-1}}$. $\alpha_0$ is well defined as we excluded $\|\mu_0\| = 0$ and $\lambda_0 = 0$. Further, note that the result is still true for $a_\tau = 0$ or $a_{\tau'} = 0$. 

\subsection{Proof of Theorem \ref{th:convexhull}}\label{app:proof_Th4}
Consider any $\tau$ in $\mathcal{G}_{\mathcal{T}}$. Using Theorem \ref{th:halfspace_intersection} we have that $h_\tau(\kappa, \lambda, \mu) =0$ and for all $\tau' \neq \tau$ $h_{\tau'}(\kappa, \lambda, \mu) < 0$. It can not be that $\lambda=0$ and $\mu=0$ as in that case $h_\tau(\kappa, 0, 0)=0$ and then $\kappa = 0$, and $h_{\tau'}(\kappa, 0, 0)=0$ contradicting $h_{\tau'}(\kappa, \lambda, \mu) < 0$.

Now we have the following equivalence:
\begin{enumerate}
\item $\exists \kappa, \lambda, \mu, \ \text{with} \ (\lambda, \mu) \neq 0$ such that 
$h_\tau(\kappa, \lambda, \mu) =0$ and $\forall \tau' \neq \tau\quad
h_{\tau'}(\kappa, \lambda, \mu) <0$.
\item $\exists \kappa, \lambda, \mu, \ \text{with} \ (\lambda, \mu) \neq 0$ such that
$\left\langle (a_\tau, b_\tau), \left(\lambda, -2\mu\right) \right\rangle = \kappa$ and $\forall \tau' \neq \tau \quad  \left\langle (a_{\tau'}, b_{\tau'}), \left(\lambda, -2\mu\right) \right\rangle < \kappa$.
\item $(a_\tau, b_\tau)$ is on the convex hull $\{ (a_\tau, b_\tau) \}_{\tau \in \mathcal{T}}$. 
\end{enumerate}

To go from point 2 to point 3 we note that a point is a part of the convex hull if there is a hyperplane (defined by a non-zero vector and a constant) going through this point and such that all other points are strictly on one side of the hyperplane.

\subsection{Proof of Theorem \ref{th:changepoint_andhull}}\label{app:proof_Th5}
First consider the case $s=p$.
Starting from  \eqref{eq:loglikelihood_eta}, we get to \eqref{eq:likelihood_change_variables} as described in Section \ref{sec_1_2_ExpFamily} and we recover our functional pruning problem of equation \eqref{eq:Set_Of_FPLin_Minimizers}  as described in Section \ref{sec2_1_General_formulation}, taking $\mathcal{T} = \{1, \ldots n-1\}$, $a_\tau = \tau$ and $b_\tau=\sum_{t=1}^\tau x_t$. Thus applying Theorem \ref{th:convexhull}, and then Theorem \ref{th:inclusion} we get the desired result for $\hat{\tau}(\eta_1)$.

The index maximising the likelihood and its value do depend on $\eta_1$, however, the set of points we need to build the convex hull does not depend on the value of $\eta_1$. 
Considering $\eta_1 = \arg\max_{\eta_1} \left(\max_{\eta_2}\ell_{\hat{\tau}(.),n}(\eta_1, \eta_2) \right)$ we have
$\hat{\tau}(.) = \hat{\tau}(\eta_1)$ and we recover the desired result for $\hat{\tau}(.)$.

Let us now consider the $s$ first coordinates of $x_t$. We call $x'_t$ the vector of $\mathbb{R}^{s}$ containing the first $s$ coordinates of $x_t$ and define $P'(\tau) = \left(\tau, \sum_{t=1}^{\tau} x'_t\right)$. Now, using the previous result, the change maximizing the likelihood restricted to the first $s$ coordinates is on the hull of $P'(\tau)$ and therefore also on the hull of $P(\tau)$. The same argument applies to any other choice of $s$ coordinates, and thus will apply to the set that maximises the sparse log-likelihood.

\subsection{Proof of Theorem \ref{th:boundnumbervertexes}}\label{app:proof_thm6}

The proof comes from applying results of \cite{kabluchko2017convex}. To do so we need to prove that the random walk $\{P(\tau) \}_{\tau \in \{1, \ldots, n-1\}} = \{ (\tau, \sum_{t=1}^\tau x_t) \}_{\tau \in \{1, \ldots, n-1\}}$ satisfies two properties, $(Ex)$ and $(GP)$ defined below.
\begin{itemize}
\item ($Ex$) For any permutation $\sigma$ of the set $\{1,\dots, n-1\}$, we have the distributional equality 
\begin{equation*}
\{ (1, x_\tau) \}_{\tau \in \{1, \ldots, n-1\}} \stackrel{D}{=} \{ (1, x_{\tau'}) \}_{\tau' \in \sigma}    
\end{equation*}
\item($GP$) Any $p+1$ vectors of the form $(\tau, \sum_{t=1}^\tau x_t)$ are almost surely linearly independent.
\end{itemize}
We show this in the following lemma.
\begin{lemma}
\label{lem:RW_conditions}
Assuming that all $x_\tau^d$ are independent and identically distributed with a continuous distribution, the properties $(Ex)$ and $(GP)$ of \cite{kabluchko2017convex} are satisfied. 
\end{lemma}

\begin{proof}
In the case where all $x_\tau^d$ are independent and identically distributed, the property $(Ex)$ holds. To show that property $(GP)$ holds, consider time $n \geq p+2$ and choose any indices $\tau_1, \dots, \tau_{p+1}$ such that $0 < \tau_1 < \cdots < \tau_{p+1} < n$. We define the following variables for $d=1,\ldots,p+1$:
\begin{equation*}
 S_{d} = \left(\tau_d, \sum_{t=1}^{\tau_d} x_t\right)\in \RR^{p+1}\,,\quad M_d = S_{d} - S_{d-1} \,, \quad \text{ with }S_0 = 0\,.   
\end{equation*}
It can be observed that $S_{d} = \sum_{k \leq d} M_{k}$ for any $d = 1,\dots, p+1$. Therefore, the linear independence of $S_1, \dots, S_{p+1}$ is equivalent to the linear independence of $M_1, \dots, M_{p+1}$. We denote the matrix of size $(p+1) \times (p+1)$ where its $d$-th column is  $M_{d}$ as $M$. Consider its transpose matrix $M^T$. Its first column-vector is filled with integers between $1$ and $n-1$. Linear algebra shows that proving linear independence for column vectors in $M$ or in $M^T$ is equivalent. The last $p$ column vectors in $M^T$ are drawn independently with a continuous density and thus are almost surely linearly independent, spanning a hyperplane of dimension $p$. The probability that a point in $\RR^{p+1}$, the first column vector of $M^T$, is in this span is $0$. In conclusion, for any given $n$, there are finitely many matrices $M$ to consider, and their probabilities satisfy $(GP)$, ensuring that property $(GP)$ holds.
\end{proof}

Based on Lemma \ref{lem:RW_conditions} properties $(Ex)$ and $(GP)$ hold and we can apply formula (2) from \cite{kabluchko2017convex} on the expected values of $U_n^ p$ and $V_n^p$. This completes the proof of Theorem \ref{th:boundnumbervertexes}.

\begin{remark}
Theorem \ref{th:boundnumbervertexes} gives the expectation for i.i.d. observations. One might wonder what happens if there is a changepoint in the distribution at $i < n$. Note that the points on the hull $\{P(\tau)\}_{\tau \in \{1, \ldots, n-1\}}$ are necessarily on the hull of $\{P(\tau)\}_{\tau \in \{1, \ldots, i\}}$ or on the hull of $\{P(\tau)\}_{\tau \in \{i+1, \ldots, n-1\}}$. Thus applying the theorem on all points from $1$ to $i$ and from $i+1$ to $n-1$ we expect at most $\mathbb{E}(V_{i+1}^p) + \mathbb{E}(V_{n-i}^p) \leq 2 \mathbb{E}(V_{n}^p)$ points on the hull (as the expectation of $V_i^p$ is increasing with $i$). Thus if there is a changepoint (and assuming we have not stopped because the GLR statistics is lower than the specified threshold) the expected number of candidate changepoints is still in $\mathcal{O}(\log(n)^p)$.
\end{remark}

\subsection{Empirical validation of the bound of Theorem \ref{th:boundnumbervertexes}}\label{append4a_GaussianMd}
We empirically tested the validity of Theorem \ref{th:boundnumbervertexes} for $p\leq 5$. We use the Gaussian distribution as an example of a distribution with continuous density and independent dimensions.
To be specific, for $p$ from $1$ to $5$, we simulated time series $x_{1:n}$, with $x_t \sim \mathcal N_p(0, I_p)$ i.i.d. We varied the time series length from $2^{10}+1 (\approx 10^3)$ to $2^{23}+1(\approx8\times10^6)$ and simulated $100$ data sets for each length. We used the \texttt{QuickHull} algorithm \cite{Barber1996} implemented in the R package \texttt{geometry} \cite{geometry_R_package} to evaluate the number of faces and vertices of the convex hull of $\{P((\tau)\}_{\tau \in \{1, \ldots, n-1\}}$. It can be seen in Figure \ref{Figure_GM_CH_estimations} that the observed number of faces and vertices are close to their theoretical expectations (from Theorem \ref{th:boundnumbervertexes}). Additional details on the calculations of the Stirling numbers are provided in Supplementary Material~\ref{append3_Stirling}.

\begin{figure}[!t]
     \centering
     \includegraphics[width=0.95\linewidth]{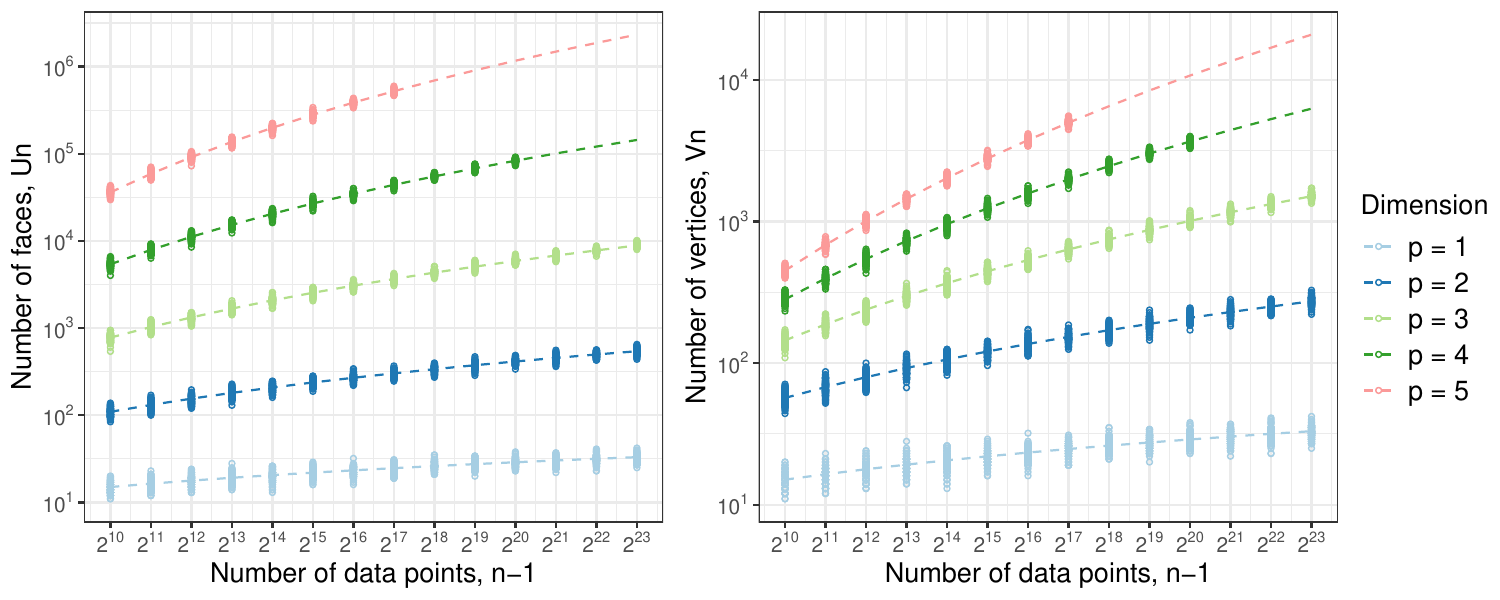}
     \caption{The number of faces (left) and vertices (right) of the convex hull of $\{P(\tau)\}_{\tau \in \{1\dots,n-1\}}$ for dimension $1\le p \le 5$. We simulated $100$ i.i.d. Gaussian data $\mathcal N_p(0, I_p)$ for $n$ from $(2^{10}+1)$ to $(2^{23}+1)$. Dashed lines correspond to the expected number of faces (left) and vertices (right) presented in Theorem \ref{th:boundnumbervertexes}. We consider a maximum running time of $40$ minutes for the \texttt{QuickHull} algorithm. This is why some results for $p=4$ and $5$ for large number of observations are missing.}
\label{Figure_GM_CH_estimations}
   \end{figure}

The proof of Theorem~\ref{th:boundnumbervertexes} depends on the fact that the distribution is continuous. We thus empirically tested the robustness of the expectations to this assumption by considering time series drawn with a discrete i.i.d. Poisson distribution. The results closely resemble the Gaussian case and details can be found in Section~\ref{append5_PoissonMd}.

\section{Exponential Family Models}

\begin{table}[t]
  \centering
  \caption{\label{tab1_ExpFamily}Modelling a change in the mean of a $p$-variate data with independent Gaussian (known variance), Poisson, Binomial, Exponential, and Pareto type-I errors. Examples of  distributions from the natural exponential family with a parameter $\theta$ in $\mathbb R^{p}$ and the corresponding forms of the functions $s$, $r$, $A'$, and $B'$. Without loss of generality, the variance of the Gaussian model is assumed to be $1$ for all dimensions. The number of trials is assumed to be $m$ for all dimensions in the Binomial model. For the Pareto type-I model, the minimum value for each dimension, denoted $y_m$, is assumed to be known.}
  \label{tab:expfam}
  \resizebox{\linewidth}{!}{%
    \begin{tabular}{lcccc}
      \toprule
      Distribution & $\eta := r(\theta)$ & $x := s(y)$ & $A'(\eta)$ & $B'(x)$ \\
      \midrule
      Gaussian (changes in mean) & $\theta$ & $y$ & $\|\eta\|^2$ & $\|x\|^2 + \log (2\pi)^p$ \\
      Poisson                    & $\log\theta$ & $y$ & $2\left\langle e^{\eta}, \mathbf{1}_p\right\rangle$ & $2\left\langle \log x!, \mathbf{1}_p\right\rangle$ \\
      Binomial                   & $\log\frac{\theta}{1-\theta}$ & $y$ & $2m\langle \log(1+e^\eta), \mathbf{1}_p\rangle$ & $-2\left\langle \log \begin{pmatrix} m \\ x \end{pmatrix}, \mathbf{1}_p\right\rangle$ \\
      Exponential                & $-\theta$ & $y$ & $-2\langle \log (-\eta), \mathbf{1}_p\rangle$ & 0 \\
      Pareto type-I              & $-\theta-1$ & $\log(y)$ & $-2\langle \log (-1-\eta), \mathbf{1}_p\rangle + 2\langle 1+\eta, \log(y_m)\rangle$ & 0 \\
      \bottomrule
    \end{tabular}%
  }
\end{table}

Examples of models and change-types covered by our setting are given in Table \ref{tab1_ExpFamily}. Our algorithms can also deal with data where different components of a multivariate stream are from different models.

Our setting also covers the change in  mean ($\theta_1$) and variance ($\theta_2$) (so $p''=2$) of a Gaussian signal in $\RR$ (so $p'=1$). In this slightly more complex case we have $p' \neq p=2$  and we can define $s$, $r$, $A'$ and $B'$ as follows 
\begin{align*}
\theta  =  (\theta_1, \theta_2),~
    \eta   =  \left(\frac{\theta_1}{\theta_2}, -\frac{1}{2\theta_2}\right), ~
    s(y)  =  (y, y^2),~
     A'(\eta)   =  \frac{\theta_1^2}{\theta_2} - \log(\theta_2), ~
    B'(x)   =  \log(2\pi). 
\end{align*}

\section{Index Set Equality}\label{append1_index_Set}
The equality $\mathcal{F}_{\mathcal{T}} = \mathcal{G}_{\mathcal{T}}$ can be proved with weaker conditions on $A$ than those of Theorem~\ref{th:equality}. This equality has a simple geometric interpretation: to be true, any half-line in $Im(A)\times\mathcal{D}_A$ has to intersect the manifold defined by $\{(A(\mu),\mu), \mu \in \mathcal{D}_A\}$.
\begin{theorem}\label{th:equality_generalCase}
If $A$ is continuous and defined on a cone with $Im(A) = \mathbb{R}^+$ such that $A(0) = 0$, $\nabla A(0) = 0$  and $\lim_{\|\mu\| \to \infty}\frac{A(\mu)}{\|\mu\|} = +\infty$  we have $\mathcal{F}_{\mathcal{T}} = \mathcal{G}_{\mathcal{T}}$.
\end{theorem}
\begin{proof}
Consider $\tau$ in $\mathcal{G}_{\mathcal{T}}$. For all $\tau' \neq \tau$, the inverse image of the open set $(0,+\infty)$ by linear functions $g_\tau - g_{\tau'}$ is a non-empty intersection of half-spaces. Moreover, there exists $(\lambda_0,\mu_0) \in Im(A)\times\mathcal{D}_A$ such that for all $\tau' \neq \tau$ in $\mathcal{T}$  we have: $g_\tau(\lambda_0, \mu_0) > g_{\tau'}(\lambda_0, \mu_0)$.  For $\alpha >0$ we get by linearity
\begin{equation*}
g_\tau(\alpha \lambda_0, \alpha \mu_0) = \alpha  g_\tau(\lambda_0,  \mu_0) > \alpha g_{\tau'}( \lambda_0,  \mu_0) = g_{\tau'}(\alpha \lambda_0, \alpha \mu_0).
\end{equation*}
Therefore, as all points of the half-line $\{\alpha(\lambda_0,\mu_0), \alpha > 0\}$ belong to this set, it shows that index $\tau$ is optimal on an open polyhedral cone.

Suppose that for the considered $(\lambda_0,\mu_0)$,  there exists $\alpha_0 >0$ such that $\lambda_0 = \frac{A(\alpha_0 \mu_0)}{\alpha_0}$ then $\alpha_0\mu_0$ is a value in the cone $\mathcal{D}_A$ verifying all the constraints $f_\tau(\alpha_0\mu_0) > f_{\tau'}(\alpha_0\mu_0)$, so that $\tau$ also belongs to $\mathcal{F}_{\mathcal{T}}$. In geometric terms, this means that any half-line in the polyhedral cone associated with index $\tau$ intersects the manifold defined by $\{(A(\mu),\mu), \mu \in \mathcal{D}_A\}$.

It remains to prove that the proposed assumptions over $A$ are sufficient to find such an alpha proving the set equality. We consider the same half-line  $\{\alpha(\lambda_0,\mu_0), \alpha > 0\}$. By a linear approximation about $0$, we have $A(\epsilon) = A(0) + \langle \epsilon, \nabla A(0) \rangle + o(\|\epsilon\|) = o(\|\epsilon\|)$ and locally, we get $ A(\alpha\mu_0) =  o(\|\alpha\mu_0\|) < \alpha\lambda_0$ for small $\alpha$. Notice that as $Im(A) = \mathbb{R}^+$ and the cone is open, we have $\lambda_0 > 0$ and can choose $\|\mu_0\| > 0$. The half-line is therefore above the manifold ("$(\alpha\lambda_0,\alpha\mu_0) > ( A(\alpha\mu_0),\alpha\mu_0)$").

Using condition $\lim_{\alpha \to \infty}\frac{A(\alpha\mu_0)}{\|\alpha\mu_0\|} = +\infty$, continuity of $A$ and previous local result $\frac{A(\alpha\mu_0)}{\|\alpha\mu_0\|} = o(1)$ for small $\alpha$, it shows that there exists $\alpha_0$ such that $\frac{A(\alpha_0\mu_0)}{\|\alpha_0\mu_0\|} = \frac{\lambda_0}{\|\mu_0\|}$ (as $\frac{\lambda_0}{\|\mu_0\|} > 0$). Thus $(\alpha_0\lambda_0, \alpha_0\mu_0) = (A(\alpha_0\mu_0), \alpha_0\mu_0)$ which leads to $f_\tau(\alpha_0\mu_0) > f_{\tau'}(\alpha_0\mu_0)$ and index $\tau$ is in $\mathcal{F}_{\mathcal{T}}$.
\end{proof}
\section{Stirling numbers of the first kind in terms of the harmonic series}
\label{append3_Stirling}
To obtain numerical values of $U^p_n$ and $V_n^p$ of the convex hull of $\{P (\tau)\}_{\tau\in \{1,\dots,n-1\}}$  for $1\le p\le 5$ we need to know the expressions of $\begin{bmatrix}
n \\
m\\
\end{bmatrix}$ for $m=0,\dots,6$. As shown in \cite{Adamchik1997OnSN}, the general formula for the Stirling numbers of the first kind in terms of the harmonic series is
\begin{equation}
    \label{eq-Stirling}
    \begin{bmatrix}
n\\
m
\end{bmatrix} = \frac{(n-1)!}{(m-1)!} \omega(n, m-1)\,.
\end{equation}
The $\omega$-sequence is defined recursively by 
\begin{equation*}
\omega(n,m) = \mathbbm{1}_{m=0} +\sum_{k=0}^{m-1} (1-m)_k \sigma_{k+1} \omega (n,m-1-k)\,,    
\end{equation*}
where $\sigma_{k+1} =\sum_{i=1}^{n-1} \frac{1}{i^{k+1}}$ is the partial sums of the harmonic series $ \left \{ \frac{1}{i^{k+1}}\right\}_{i=1,2,\dots}$ and $(1-m)_k$ are the Pochhammer symbols. 
Applying (\ref{eq-Stirling}) for $m=0,\dots, 6$, we get 
$$
\begin{bmatrix}
n \\
0\\
\end{bmatrix}  = 1\,, \quad
\begin{bmatrix}
n \\
1\\
\end{bmatrix}  = (n-1)!\,, 
\\
\begin{bmatrix}
n \\
2\\
\end{bmatrix}  = (n-1)!\sigma_1\,,
$$

$$
\begin{bmatrix}
n \\
3\\
\end{bmatrix}  = \frac{(n-1)!}{2} (\sigma_1^2 - \sigma_2)\,,
\quad
\begin{bmatrix}
n \\
4\\
\end{bmatrix}  = (n-1)! \left ( \frac{\sigma_1^3}{6} -\frac{\sigma_1 \sigma_2}{2}+ \frac{\sigma_3}{3}\right )\,,
$$
$$
\begin{bmatrix}
n \\
5\\
\end{bmatrix}  = (n-1)! \left (  \frac{\sigma_1^4}{24} -\frac{\sigma_1^2\sigma_2}{4} + \frac{\sigma_1\sigma_3}{3}+\frac{\sigma_2^2}{8} - \frac{\sigma_4}{4}\right )\,,
$$
$$
\begin{bmatrix}
n \\
6\\
\end{bmatrix}  = (n-1)!\left (  \frac{\sigma_1^5}{120}-\frac{\sigma_1^3\sigma_2}{12} + \frac{\sigma_1^2\sigma_3}{6}+\frac{\sigma_1\sigma_2^2}{8} - \frac{\sigma_1\sigma_4}{4}- \frac{\sigma_2\sigma_3}{6} + \frac{\sigma_5}{5}\right ).
$$

\section{Additional Empirical evaluations}\label{App:additional_simulation_studies} 

\subsection{Multidimensional Poisson Model}
\label{append5_PoissonMd}
We simulated time series with each data-point simulated a $p$ independent Poisson random variables with mean 1. We varied $p$ and $n$ as in 
Section~\ref{append4a_GaussianMd}.
Using the same procedure as in Section \ref{append4a_GaussianMd}, we evaluate the number of faces and vertices of the convex hull $\{P(\tau)\}_{\tau \in \{1\dots,n-1\}}$. The results are presented in Figure~\ref{Figure_PM_CH_estimations}. We note that the results are similar to the Gaussian case (see Figure \ref{Figure_GM_CH_estimations}).

Next, using the same procedure as in Section \ref{sec-EmpiricalRunTimeGauss}, we evaluated the run times of MdFOCuS for a known or unknown pre-change parameter for Poisson distribution. The results are presented in Figure~\ref{PM_Runtime}. We note that the results are similar to the Gaussian case (see Figure~\ref{GM_Runtime}).

\begin{figure}[ht]%
     \centering
     \includegraphics[width=0.95\linewidth]{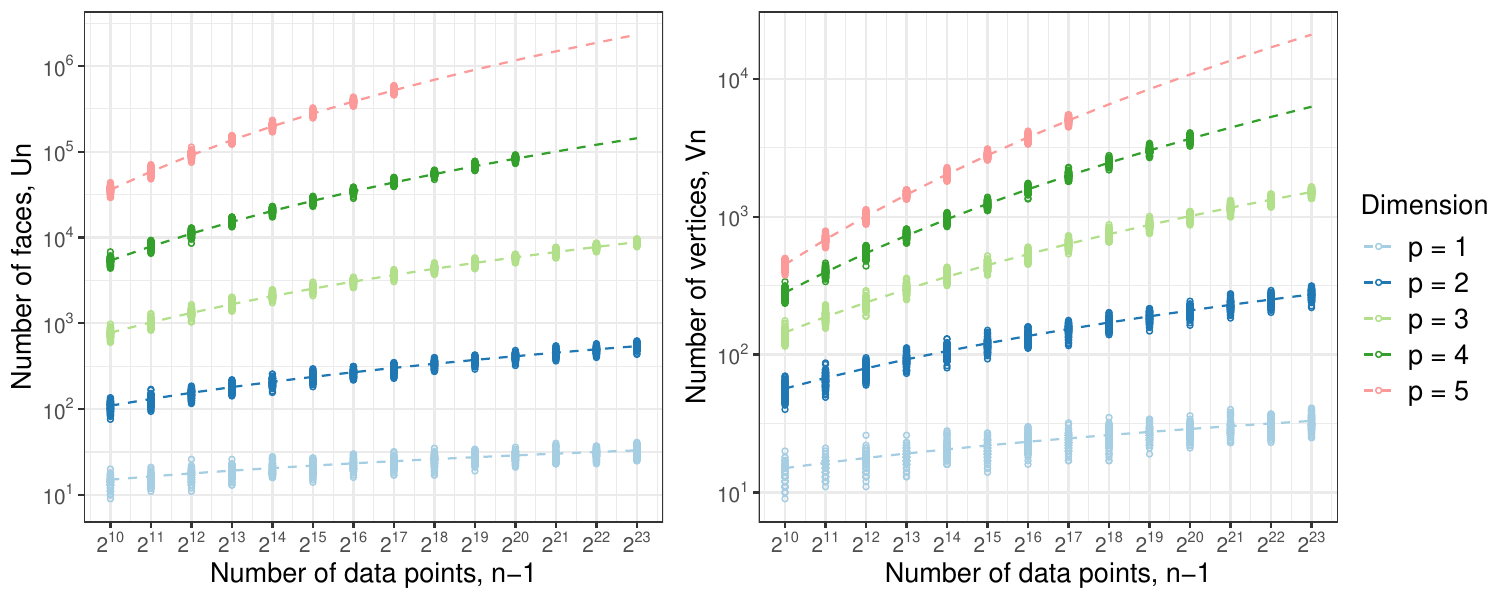}
     \caption{The number of faces (left) and vertices (right) of the convex hull built from the set of points $\{P(i)\}_{i \in \{1\dots,n-1\}}$ for dimension $1\le p \le 5$. Results averaged over 100 realisations of independent $p$-dimensional Poisson random variables with mean 1 data. Dashed lines correspond to the expected number of faces (left) and vertices (right) presented in Theorem \ref{th:boundnumbervertexes}. We consider a maximum running time of 40 minutes for the \texttt{Quickhull} algorithm, hence no data points for $p=4$ and $5$ for large numbers of observations.}
\label{Figure_PM_CH_estimations}
\end{figure}
\begin{figure}[ht]%
     \centering
     \includegraphics[width=0.8\linewidth]{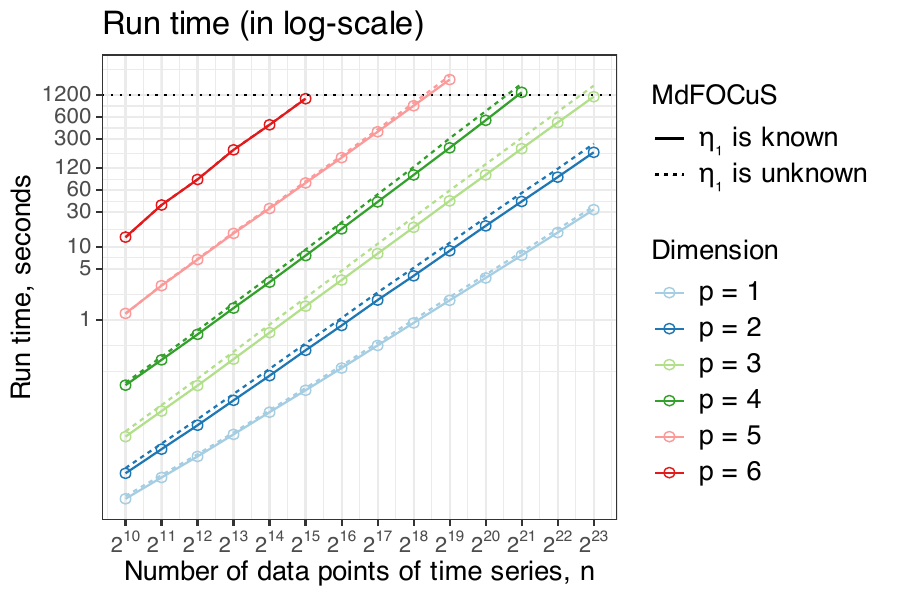}\caption{
     Run times in seconds of MdFOCuS with known (dashed line) and unknown (full line) pre-change parameter $\eta_1$, both with $\alpha = 2$, $\beta = 1$ and $maxSize =14$, in dimension $p=1,\dots, 6$ using time series $x_{1:n}$ simulated from the independent Poisson model. Run times are averaged over $100$ data sets. We considered a maximum running time of $20$ minutes (horizontal dotted black line) for the algorithms, hence no run times for $p=4,5$ and $6$ for large numbers of observations. The slope of the curves was obtained using a simple linear regression: 
      for known $\eta_1$: $\approx 1.013$ ($p=1$), $\approx 1.126$ ($p=2$), $\approx 1.192$ ($p=3$), $\approx 1.210$ ($p=4$), $\approx 1.172$ ($p=5$) and $\approx 1.250$ ($p=6$); 
     for unknown $\eta_1$: $\approx 1.018$  ($p=1$), $\approx 1.138$ ($p=2$), $\approx 1.213$ ($p=3$), $\approx 1.240$ ($p=4$), $\approx 1.190$ ($p=5$) and $\approx 1.254$ ($p=6$).
}\label{PM_Runtime}
\end{figure}

\subsection{Run time as a function of a slope parameter 
}
\label{append4_Runtime_alpha}
We consider the empirical run times of the MdFOCuS algorithm with a known pre-change parameter $\eta_1$ as a function of parameter $\alpha$. Using the implementation of  Algorithm \ref{MdFOCuS_algo} for a change in the mean of multi-dimensional Gaussian signal (see Table \ref{tab1_ExpFamily}) in R/C++ using the \texttt{qhull} library we studied how the $\alpha$  parameter affects its run time. To do this, we generate $100$ time series with $n=10^5$ data points, with each data point an independent realisation of a standard $p$-dimensional Gaussian random variable, in dimension $p = 1,2$ and $3$. In Figure \ref{GM_alpha_runtime} we report the average run time as a function of $\alpha$ for MdFOCuS for the pre-change parameter known case. The value of \texttt{maxSize} was always initialised to $p+2$ and $\beta$ was always set to $1$. We observe a minimum run time for $\alpha$ in $[2, 4]$. As a consequence, in the rest of the simulations we always set $\alpha$ to $2$.
\begin{figure*}[ht]%
\centering
\makebox{\includegraphics[width=.65\linewidth]{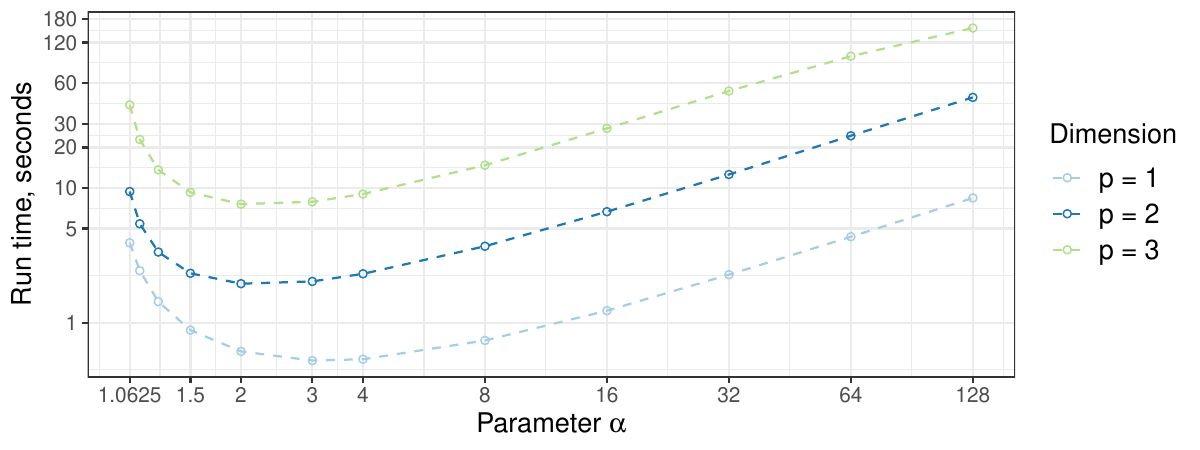}}
\caption{Run time of MdFOCuS as a function of $\alpha$ for $p$-variate time series with $n=10^5$ data points in dimension $p = 1,2$ and $3$. Results averaged over 100 simulations.}
\label{GM_alpha_runtime}
\end{figure*}

\subsection{50-dimensional time series, controlled by false positive rate}\label{sec:approx_comparison}

We present, in Table \ref{tab:tab:dd}, an additional simulation study to compare the simple one-dimentional multivariate implementation proposed by \cite{romano2022fast} (here called FOCuS), with the high-dimensional approximation for MdFOCuS using the two-dimensional approximation of the convex hull alongside the partial likelihood ratio statistic (see  Section \ref{sec3_3_AlgoMR} and subsection \ref{heuristic_section}). To facilitate a fair comparison, we tune the thresholds of all methods to achieve the same pre-change false positive rate of 0.01. We applied a Bonferroni correction to cope with multiple thresholds. The simulation scenarios mirror those of the main study but with the changepoint occurring at time 10000. This is to emphasize the impact of the approximations. We expect the two-dimensional approximation to outperform the one-dimensional approximation proposed by \cite{romano2022fast}, particularly for small-magnitude changes.  However, as observed in the main study, the original FOCuS method is still superior when observing a change in one dimension, as, despite in that case the statistics being equal, its threshold is slightly lower (because it considers fewer statistics in total).


\begin{table}[!h]
\centering
\caption{\label{tab:tab:dd}Average detection delay with 50-dimensional time series for a change at time $t = 10000$. Pre-change parameter is known on the left, and unknown on the right. Best results per row are highlighted in bold. Results are by controlling the false positive rate. FOCuS (est 500) estimates the pre-change parameter from training data of length 500.}
\centering
\resizebox{\ifdim\width>\linewidth\linewidth\else\width\fi}{!}{
\fontsize{12}{14}\selectfont
\begin{tabular}[t]{rrrrrrr}
\toprule
\multicolumn{2}{c}{ } & \multicolumn{2}{c}{Pre-change Known} & \multicolumn{3}{c}{Pre-change Unknown} \\
\cmidrule(l{3pt}r{3pt}){3-4} \cmidrule(l{3pt}r{3pt}){5-7}
magnitude & sparsity & FOCuS  & MdFOCuS & FOCuS (est 500) & FOCuS & MdFOCuS\\
\midrule
0.125 & 1 & \textbf{2347.52} &  2426.00 & 9980.48 & \textbf{3161.26} &  3295.96\\
0.125 & 5 & 3524.53 & \textbf{ 3380.42} & 9981.21 & 6088.01 &  \textbf{5375.96}\\
0.125 & 10 & 3997.59 & \textbf{3814.82} & 9980.60 & 7299.97 &  \textbf{6362.57}\\
0.125 & 25 & 4288.27 & \textbf{4224.50} & 9979.48 & 8474.92 &  \textbf{7479.11}\\
0.125 & 50 & 4509.18 & \textbf{4362.75} & 9982.43 & 8874.72 &  \textbf{7853.60}\\
\addlinespace
0.250 & 1 & \textbf{584.56} & 599.49 & 8932.53 & \textbf{627.85} &  644.16\\
0.250 & 5 & 977.95 & \textbf{871.53} & 9779.57 & 1167.59 &  \textbf{970.06}\\
0.250 & 10 & 1159.92 & \textbf{996.14} & 9845.31 & 1504.27  & \textbf{1117.59}\\
0.250 & 25 & 1299.23 & \textbf{1096.57} & 9893.83 & 1925.51  & \textbf{1322.48}\\
0.250 & 50 & 1350.10 & \textbf{1163.00} & 9893.97 & 2216.25  & \textbf{1400.04}\\
\addlinespace\textbf{}
0.500 & 1 &\textbf{ 147.73} & 151.25 & 2711.61 & \textbf{150.70} &  153.69\\
0.500 & 5 & 245.45 & \textbf{212.26} & 4818.17 & 259.12 &  \textbf{215.56}\\
0.500 & 10 & 310.56 & \textbf{247.15} & 6139.80 & 357.20 &  \textbf{257.50}\\
0.500 & 25 & 358.34 & \textbf{274.32} & 7015.92 & 458.89 &  \textbf{290.81}\\
0.500 & 50 & 385.03 & \textbf{290.86} & 7316.34 & 534.77 &  \textbf{309.12}\\
\addlinespace
0.750 & 1 &\textbf{ 66.55 } & 67.78 & 1194.15 & \textbf{66.86 }&  68.33\\
0.750 & 5 & 109.84 & \textbf{96.61 }& 2208.06 & 112.88 &  \textbf{97.46}\\
0.750 & 10 & 146.07  & \textbf{111.96} & 3022.77 & 159.63  & \textbf{113.25}\\
0.750 & 25 & 173.78  &\textbf{ 123.47 }& 3713.07 & 208.40  & \textbf{128.58}\\
0.750 & 50 & 185.83  &\textbf{ 129.10 }& 3977.72 & 239.54  & \textbf{134.76}\\
\addlinespace
1.000 & 1 & \textbf{37.90 } & 39.03 & 668.06 & \textbf{38.16}  & 39.06\\
1.000 & 5 & 63.03  & \textbf{55.14 }& 1254.92 & 64.23  & \textbf{55.10}\\
1.000 & 10 & 84.73  &\textbf{ 63.52 }& 1759.00 & 89.73  & \textbf{63.96}\\
1.000 & 25 & 102.40  & \textbf{73.13} & 2226.95 & 120.41  & \textbf{75.12}\\
1.000 & 50 & 109.77  & \textbf{75.73} & 2427.92 & 139.52  & \textbf{78.38}\\
\addlinespace
2.000 & 1 & \textbf{10.22} &  10.52 & 166.65 & \textbf{10.20}  & 10.52\\
2.000 & 5 & 17.22 & \textbf{14.74 }& 314.93 & 17.38  & \textbf{14.75}\\
2.000 & 10 & 23.24  &\textbf{ 17.06} & 456.38 & 24.11  & \textbf{17.10}\\
2.000 & 25 & 28.69 & \textbf{18.96} & 596.50 & 32.33  & \textbf{19.40}\\
2.000 & 50 & 31.41 &\textbf{ 19.90} & 662.81 & 38.29  & \textbf{20.43}\\
\bottomrule
\end{tabular}}
\end{table}

\subsection{Additional Plug-in estimated methods}\label{app:extra_training_data}

The following Tables \ref{tab:dd3_extra}, \ref{tab:dd5_extra} and \ref{tab:dd100_extra} show detection delays for various amount of training data prior to the monitoring period, and can be used to aid the comparison of the methods in the main simulation study.
\begin{table}[!h]
\centering
\caption{\label{tab:dd3_extra}Average detection delay with 3-dimensional time series for a change at time $t = 1000$. Comparison of pre-change mean known methods with plug-in oracle value (ora) against plug-in mean estimate for differing training data sizes (500, 250 and 100).}
\centering
\resizebox{\ifdim\width>\linewidth\linewidth\else\width\fi}{!}{
\fontsize{12}{14}\selectfont
\begin{tabular}[t]{rrrrrrrrrr}
\toprule
magnitude & sparsity & FOCuS (ora)& FOCuS (est 500) & FOCuS (est 250) & FOCuS (est 100) & ocd (ora) & ocd (est 500) & ocd (est 250) & ocd (est 100)\\
\midrule
0.125 & 1 & 786.67 & 1162.32 & 1470.76 & 2253.81 & 1227.57 & 1721.56 & 1783.41 & 2490.39\\
0.125 & 2 & 791.36 & 1327.92 & 1693.06 & 2302.31 & 1421.54 & 2043.79 & 1987.04 & 2415.50\\
0.125 & 3 & 859.03 & 1352.66 & 1723.76 & 2329.12 & 1496.18 & 2046.77 & 2097.23 & 2337.57\\
\addlinespace
0.250 & 1 & 250.58 & 325.38 & 470.69 & 1251.62 & 227.18 & 330.05 & 430.33 & 1102.87\\
0.250 & 2 & 253.56 & 364.98 & 556.86 & 1375.63 & 268.86 & 389.39 & 516.28 & 1173.38\\
0.250 & 3 & 273.92 & 401.56 & 630.94 & 1433.94 & 300.69 & 468.61 & 559.66 & 1204.47\\
\addlinespace
0.500 & 1 & 65.85 & 80.59 & 113.75 & 312.64 & 56.46 & 71.71 & 83.90 & 177.84\\
0.500 & 2 & 71.85 & 92.22 & 132.99 & 354.26 & 65.18 & 81.19 & 96.50 & 181.77\\
0.500 & 3 & 79.21 & 106.81 & 155.61 & 396.17 & 74.50 & 92.76 & 103.97 & 169.63\\
\addlinespace
0.750 & 1 & 31.80 & 37.30 & 50.73 & 133.10 & 28.95 & 34.93 & 39.46 & 67.11\\
0.750 & 2 & 34.66 & 42.69 & 61.46 & 154.01 & 32.52 & 38.90 & 43.47 & 70.93\\
0.750 & 3 & 37.19 & 49.00 & 71.75 & 172.62 & 35.43 & 42.38 & 46.40 & 68.02\\
\addlinespace
1.000 & 1 & 18.51 & 21.83 & 29.18 & 73.23 & 18.87 & 21.98 & 24.16 & 35.99\\
1.000 & 2 & 20.30 & 25.32 & 34.55 & 86.22 & 20.45 & 23.68 & 25.45 & 37.82\\
1.000 & 3 & 21.68 & 27.79 & 40.03 & 98.67 & 21.52 & 24.70 & 26.23 & 37.56\\
\addlinespace
2.000 & 1 & 5.35 & 6.18 & 8.06 & 19.03 & 6.40 & 7.06 & 7.12 & 9.14\\
2.000 & 2 & 5.97 & 7.18 & 9.46 & 22.10 & 6.64 & 7.10 & 7.25 & 9.64\\
2.000 & 3 & 6.51 & 7.89 & 10.62 & 24.63 & 6.80 & 7.32 & 7.41 & 9.48\\
\bottomrule
\end{tabular}}
\end{table}

\begin{table}[!h]
\centering
\caption{\label{tab:dd5_extra}Average detection delay with 5-dimensional time series for a change at time $t = 1000$. Comparison of  pre-change mean known methods with plug-in oracle value (ora) against plug-in mean estimate for differing training data sizes (500, 250 and 100).}
\centering
\resizebox{\ifdim\width>\linewidth\linewidth\else\width\fi}{!}{
\fontsize{12}{14}\selectfont
\begin{tabular}[t]{rrrrrrrrrr}
\toprule
magnitude & sparsity & FOCuS (ora) & FOCuS (est 500) & FOCuS (est 250) & FOCuS (est 100) & ocd (ora) & ocd (est 500) & ocd (est 250) & ocd (est 100)\\
\midrule
0.125 & 1 & 833.11 & 1488.18 & 2219.97 & 2706.40 & 911.17 & 1583.01 & 2458.66 & 3081.32\\
0.125 & 2 & 892.82 & 1534.50 & 2221.24 & 2685.51 & 1106.68 & 1599.44 & 2441.70 & 3074.78\\
0.125 & 3 & 925.08 & 1615.20 & 2255.65 & 2691.14 & 1101.38 & 1706.65 & 2493.92 & 3053.87\\
0.125 & 4 & 835.44 & 1472.64 & 2102.15 & 2629.56 & 1059.06 & 1472.04 & 2296.26 & 3005.94\\
0.125 & 5 & 953.85 & 1584.83 & 2228.96 & 2701.14 & 1297.37 & 1621.73 & 2448.01 & 3046.09\\
\addlinespace
0.250 & 1 & 259.75 & 435.84 & 871.35 & 1775.82 & 229.64 & 387.14 & 1064.11 & 2268.01\\
0.250 & 2 & 281.77 & 493.33 & 963.88 & 1839.07 & 253.79 & 424.68 & 1084.21 & 2281.61\\
0.250 & 3 & 297.41 & 517.02 & 1004.43 & 1924.93 & 287.10 & 446.80 & 1077.31 & 2336.42\\
0.250 & 4 & 260.23 & 471.25 & 918.63 & 1819.93 & 265.40 & 365.10 & 852.33 & 2140.46\\
0.250 & 5 & 304.13 & 541.09 & 1023.24 & 1988.96 & 309.16 & 439.40 & 1000.05 & 2350.21\\
\addlinespace
0.500 & 1 & 72.98 & 107.78 & 206.83 & 537.16 & 64.74 & 96.04 & 247.08 & 763.99\\
0.500 & 2 & 84.75 & 128.09 & 243.50 & 610.94 & 73.93 & 108.04 & 253.56 & 828.08\\
0.500 & 3 & 91.31 & 139.43 & 267.82 & 655.26 & 83.54 & 110.27 & 252.69 & 846.27\\
0.500 & 4 & 80.89 & 131.85 & 250.83 & 608.11 & 75.94 & 93.22 & 199.82 & 686.05\\
0.500 & 5 & 95.41 & 154.27 & 293.72 & 703.96 & 89.48 & 111.43 & 235.87 & 841.18\\
\addlinespace
0.750 & 1 & 33.44 & 47.69 & 91.82 & 228.48 & 32.98 & 44.57 & 108.72 & 341.20\\
0.750 & 2 & 39.21 & 58.27 & 113.10 & 267.93 & 36.22 & 48.39 & 113.74 & 380.14\\
0.750 & 3 & 41.29 & 64.95 & 124.88 & 294.45 & 38.61 & 49.04 & 110.91 & 381.23\\
0.750 & 4 & 37.03 & 61.29 & 116.62 & 277.83 & 34.71 & 41.39 & 87.09 & 305.25\\
0.750 & 5 & 45.65 & 74.44 & 137.55 & 326.77 & 42.24 & 50.95 & 104.84 & 366.27\\
\addlinespace
1.000 & 1 & 19.63 & 27.77 & 51.75 & 128.48 & 21.37 & 27.25 & 62.33 & 209.80\\
1.000 & 2 & 22.88 & 33.47 & 63.86 & 151.84 & 22.88 & 28.27 & 62.99 & 220.30\\
1.000 & 3 & 24.67 & 36.48 & 71.35 & 167.23 & 23.65 & 28.63 & 62.24 & 216.61\\
1.000 & 4 & 22.45 & 35.42 & 67.41 & 156.69 & 20.32 & 23.52 & 49.88 & 167.50\\
1.000 & 5 & 27.18 & 43.05 & 81.02 & 186.18 & 25.07 & 29.31 & 61.01 & 204.52\\
\addlinespace
2.000 & 1 & 5.79 & 7.72 & 13.40 & 31.86 & 7.25 & 7.97 & 15.78 & 56.68\\
2.000 & 2 & 6.88 & 9.16 & 16.57 & 38.91 & 7.47 & 8.06 & 15.66 & 55.90\\
2.000 & 3 & 7.21 & 10.14 & 18.62 & 42.92 & 7.33 & 8.15 & 15.74 & 53.75\\
2.000 & 4 & 6.64 & 10.04 & 17.92 & 40.14 & 6.10 & 6.66 & 12.99 & 42.46\\
2.000 & 5 & 7.91 & 11.92 & 21.30 & 48.72 & 7.40 & 8.10 & 15.55 & 52.15\\
\bottomrule
\end{tabular}}
\end{table}

\begin{table}[!h]
\centering
\caption{\label{tab:dd100_extra}Average detection delay with 100-dimensional time series for a change at time $t = 1000$. Comparison of pre-change mean known methods with plug-in oracle value (ora) against plug-in mean estimate for differing training data sizes (500, 250 and 100).}
\centering
\resizebox{\ifdim\width>\linewidth\linewidth\else\width\fi}{!}{
\fontsize{12}{14}\selectfont
\begin{tabular}[t]{rrrrrrrrrr}
\toprule
magnitude & sparsity & FOCuS (ora) & FOCuS (est 500) & FOCuS (est 250) & FOCuS (est 100) & ocd (ora) & ocd (est 500) & ocd (est 250) & ocd (est 100)\\
\midrule
0.125 & 1 & 1417.15 & 3269.06 & 3459.08 & 3615.97 & 1171.63 & 3361.86 & 3546.43 & 3696.77\\
0.125 & 5 & 1881.10 & 3429.84 & 3497.94 & 3604.75 & 1627.00 & 3490.18 & 3595.24 & 3691.58\\
0.125 & 10 & 2084.84 & 3450.20 & 3501.48 & 3607.64 & 1940.41 & 3510.97 & 3595.46 & 3695.33\\
0.125 & 50 & 2252.48 & 3473.50 & 3494.67 & 3611.70 & 2386.98 & 3547.64 & 3591.96 & 3704.31\\
0.125 & 100 & 2323.59 & 3505.76 & 3507.66 & 3610.74 & 2431.51 & 3564.11 & 3599.24 & 3705.24\\
0.250 & 1 & 389.11 & 1552.28 & 2659.50 & 3436.24 & 335.95 & 1643.28 & 2841.37 & 3542.47\\
0.250 & 5 & 641.18 & 2363.35 & 3054.36 & 3510.27 & 505.39 & 2477.35 & 3232.02 & 3607.01\\
0.250 & 10 & 789.27 & 2590.60 & 3178.29 & 3518.29 & 595.10 & 2706.71 & 3334.21 & 3622.86\\
0.250 & 50 & 1057.30 & 3026.93 & 3323.17 & 3549.36 & 827.30 & 3142.73 & 3453.27 & 3662.49\\
0.250 & 100 & 1080.26 & 3101.34 & 3346.50 & 3544.06 & 891.09 & 3200.56 & 3476.58 & 3650.38\\
0.500 & 1 & 105.79 & 359.11 & 721.53 & 1798.04 & 104.79 & 411.98 & 859.71 & 2132.81\\
0.500 & 5 & 169.93 & 661.68 & 1297.63 & 2580.08 & 139.23 & 708.99 & 1449.56 & 2809.88\\
0.500 & 10 & 232.04 & 900.72 & 1653.97 & 2868.20 & 170.03 & 958.71 & 1850.22 & 3077.24\\
0.500 & 50 & 357.90 & 1730.82 & 2480.88 & 3231.38 & 239.08 & 1868.74 & 2719.69 & 3408.29\\
0.500 & 100 & 372.34 & 1916.38 & 2611.86 & 3237.55 & 244.07 & 2093.59 & 2841.33 & 3407.16\\
0.750 & 1 & 47.78 & 158.37 & 310.89 & 761.25 & 57.56 & 228.43 & 461.51 & 1138.50\\
0.750 & 5 & 75.34 & 297.90 & 580.05 & 1389.39 & 73.47 & 354.74 & 715.73 & 1704.61\\
0.750 & 10 & 104.45 & 408.77 & 798.60 & 1831.94 & 84.16 & 458.03 & 925.56 & 2114.73\\
0.750 & 50 & 184.09 & 958.18 & 1626.21 & 2651.16 & 114.29 & 1036.67 & 1836.28 & 2889.70\\
0.750 & 100 & 193.77 & 1156.00 & 1823.30 & 2730.10 & 115.21 & 1283.48 & 2074.96 & 2956.24\\
1.000 & 1 & 27.12 & 89.54 & 173.04 & 419.58 & 38.67 & 158.25 & 317.49 & 772.93\\
1.000 & 5 & 43.72 & 168.56 & 325.28 & 791.73 & 46.07 & 237.96 & 476.20 & 1142.06\\
1.000 & 10 & 60.36 & 235.29 & 458.62 & 1094.25 & 50.77 & 298.13 & 599.74 & 1422.86\\
1.000 & 50 & 109.11 & 580.49 & 1074.41 & 2036.18 & 66.02 & 630.09 & 1218.93 & 2284.94\\
1.000 & 100 & 116.60 & 747.72 & 1288.45 & 2190.83 & 67.59 & 823.90 & 1466.28 & 2453.87\\
2.000 & 1 & 7.58 & 22.69 & 43.03 & 103.23 & 12.93 & 71.01 & 142.00 & 341.47\\
2.000 & 5 & 12.92 & 42.04 & 82.41 & 199.48 & 13.80 & 101.81 & 203.48 & 485.40\\
2.000 & 10 & 16.46 & 60.00 & 117.67 & 283.82 & 14.80 & 124.40 & 250.44 & 594.86\\
2.000 & 50 & 31.85 & 162.50 & 315.56 & 734.61 & 18.08 & 226.70 & 451.59 & 1021.58\\
2.000 & 100 & 34.28 & 222.17 & 422.46 & 920.47 & 18.34 & 255.33 & 493.39 & 1104.39\\
\bottomrule
\end{tabular}}
\end{table}

\newpage
\section{Dyadic MdFOCuS : an algorithm with a quasi-linear expected complexity}
\label{dyalic_section}

In this appendix we present an algorithm for online changepoint detection. In essence, as in the MdFOCuS algorithm of the main text, it maintains a set of points that is slightly larger than the set of points on the hull. It does so applying the \texttt{QuickHull} (or any convex hull) algorithm on successive and non-overlapping deterministic chunks of the points.
This deterministic choice of chunks is likely sub-optimal but it allows for a precise quantification of the expected complexity of the algorithm under the assumption of Theorem \ref{th:boundnumbervertexes}
and under the assumption that the convex hull algorithm is at worst quadratic. To be specific we obtain a $\mathcal{O}(n\log(n)^{p+1})$ bound on the expected time complexity.

We begin with two preliminary results. The first is a crude bound on the expectation of the square of the number of points on the convex hull of $n$ points  (that is $(V_n^p)^2$ following the notation of Section \ref{sec3_2_Bound}) under the assumptions of Theorem \ref{th:boundnumbervertexes}.
The second result bounds the expected complexity of computing the hull of $2n$ points if we already know the points on the hull of the first $n$ points and the last $n$ points. 
In all that follows we will assume that Assumptions \ref{assump:Vn_complexity} and \ref{assump:Time_complexity} described in the main text are true. For ease of reading, we repeat them here.

\setcounter{assumption}{0}
\begin{assumption}
In any dimension $p\ge1$, there are positive constants $c_1$ and $c_2$ for which the expected number of vertices, $\mathbb{E}[V_n^p]$, can be bounded as follows:
\begin{equation*}
\mathbb{E}[V_n^p] \leq c_1 \log(n)^p + c_2.   
\end{equation*}
\end{assumption}

This is true under the conditions of 
Theorem \ref{th:boundnumbervertexes} and Corollary \ref{col:boundnumbervertexes}.


Remember we define $Time(n)$ to be the worst-case time complexity of finding the convex hull of $n$ points.
\begin{assumption}
For our convex hull algorithm we are using there exist two positive constants $c_3$ and $c_4$ such that $Time(n) \leq c_3n^2 + c_4$. 
\end{assumption}

The assumption on the quadratic complexity is true for {\texttt{QuickHull}} for $p\le2$ \cite[]{Barber1996} and using the algorithm of \cite{Chazelle1993} it is true for $p\le3$.
We now state our bound on the expectation of the square of the number of points on the hull $V_n^p$.
\begin{lemma}\label{lemmma:bound_expectedsquare}
\begin{equation}\label{eq:bound_expectedsquare}
\mathbb{E}[(V_n^p)^2] \le c_1 n\log(n)^p + c_2n.
\end{equation}
\end{lemma}

\begin{proof}
    $V_n^p$ is bounded by $0$ and $n-1$. It follows that 
\begin{equation}
\mathbb{E}[(V_n^p)^2] = \sum_{i=0}^{n} \mathbb{P}(V_n^p=i) i^2  \le n \sum_{i=0}^{n}\mathbb{P}(V_n^p=i) i \le n \mathbb{E}[V_n^p] \le c_1 n\log(n)^p + c_2n.
\end{equation}
We get the last inequality with Assumption \ref{assump:Vn_complexity}.
\end{proof}

We now state our bound on the complexity of merging the hull of two successive chunks.

\begin{lemma}\label{lemma:bound_fusion}
    Consider we have $2n$ data points. Assume we know $\mathcal{W}_1$ the index of the points on the hull of $\{P(\tau)\}_{1 \leq \tau \leq n}$ and $\mathcal{W}_2$ the index of the points on the hull of $\{P(\tau)\}_{n+1 \leq \tau \leq 2n}$. Under the condition of Corollary \ref{col:boundnumbervertexes} and Assumption \ref{assump:Time_complexity} computing the index of the points on the hull of $\{P(\tau)\}_{1 \leq \tau \leq 2n}$ can be done in less than $4c_1c_3 n (\log n)^p + 4c_2c_3 n+c_4$ time.
\end{lemma}

\begin{proof}
We will compute those points as the points on the hull of $\{P(\tau)\}_{\tau \in \mathcal{W}_1 \cup \mathcal{W}_2}$.
By Assumption \ref{assump:Time_complexity} the time complexity of processing these $|\mathcal{W}_1|+|\mathcal{W}_2|$ points is 
\begin{equation*}
Time \left(|\mathcal{W}_1|+|\mathcal{W}_2|\right) \le c_3\left(|\mathcal{W}_1|+|\mathcal{W}_2|\right)^2 + c_4 \leq 2 c_3 \left(|\mathcal{W}_1|^2 + |\mathcal{W}_2|^2\right) + c_4.    
\end{equation*}
Now taking the expectation with respect to the number of points in $\mathcal{W}_1$ and $\mathcal{W}_2$, and noting that under the assumption of Corollary \ref{col:boundnumbervertexes} the distribution of $|\mathcal{W}_1|$ is the same as the distribution of $|\mathcal{W}_2|$ we get 
\begin{equation*}
\mathbb{E}\left(Time \left(|\mathcal{W}_1|+|\mathcal{W}_2|\right)\right) \le 4c_3\mathbb{E}\left(|\mathcal{W}_1|^2\right) + c_4.  
\end{equation*}
Using \eqref{eq:bound_expectedsquare} 
we get that the expected complexity is
\begin{equation}
\label{eq:bound_fusion}
\mathbb{E}\left(Time \left(|\mathcal{W}_1|+|\mathcal{W}_2|\right)\right) \le 4c_1c_3 n (\log n)^p + 4c_2c_3 n+c_4.
\end{equation}
\end{proof}

Now, before presenting the dyadic algorithm, we explain the set of candidates it considers and how this set is updated.

\subsection{Dyadic algorithm: Set of candidates}

Recall that at time $n$ we would ideally only store the vertices on the hull of $\{P(\tau)\}_{\tau < n}$. In the following algorithm we store the vertices on the hull of chunks of the form $\left \{k 2^q +1, \dots, (k+1) 2^q\right\}$ for all $q \geq q_{min}$. 
To be specific, we consider the base 2 description of $n$: $n = \sum_{q=0}^{\lfloor \log_2(n) \rfloor} \overline{n}_q 2^q,$ with $\overline{n}_q=0$ or $1$, and define for $q$ in $\{0, \ldots, \lfloor \log_2(n-1) \rfloor\}$
the subset $\mathcal{U}_n^{q}$ of $\{1, \ldots, n-1\}$ as
\begin{equation*}
   \mathcal{U}_n^{q} = 
\{ \tau | \sum_{q'=q+1}^{\lfloor \log_2(n-1) \rfloor} \overline{(n-1)}_{q'} 2^{q'} 
<\tau \leq 
\sum_{q'=q}^{\lfloor \log_2(n-1) \rfloor} \overline{(n-1)}_{q'} 2^{q'} \},    
\end{equation*}
with the convention that for $q  = \lfloor \log_2(n-1) \rfloor$  the sum $\sum_{q'=q+1}^{\lfloor \log_2(n-1) \rfloor} \overline{(n-1)}_{q'} 2^{q'}$ is equal to $0$.

We provide a few examples of these sets for $n=20$ to $n=31$ in Table \ref{tab:exampleUset}. It can be observed that often $\mathcal{U}_n^q$ differs from $\mathcal{U}_{n+1}^q$ only for small values of $q$.

\begin{table} [!ht]
\caption{\label{tab:exampleUset}A few $\mathcal{U}_n^q$ sets.}
\centering
\begin{tabular}{cccccc} 
    \hline 
           $n$   & $q=4$ & $q=3$ & $q=2$ & $q=1$ & $q=0$ \\ \hline
    31 & $\{1, \ldots 16\}$ & $\{17, \ldots 24 \}$ & $\{25, \ldots 28 \}$ & $\{ 29, 30\}$ & $\emptyset$ \\ 
    30 & $\{1, \ldots 16\}$ & $\{17, \ldots 24 \}$ & $\{25, \ldots 28 \}$ & $\emptyset$ & $\{ 29\}$ \\ 
    29 & $\{1, \ldots 16\}$ & $\{17, \ldots 24 \}$ & $\{25, \ldots 28 \}$ & $\emptyset$ & $\emptyset$ \\ 
    28 & $\{1, \ldots 16\}$ & $\{17, \ldots 24 \}$ & $\emptyset$ & $\{25, 26 \}$ & $\{27\}$ \\ 
    27 & $\{1, \ldots 16\}$ & $\{17, \ldots 24 \}$ & $\emptyset$ & $\{25, 26 \}$ & $\emptyset$ \\ 
    26 & $\{1, \ldots 16\}$ & $\{17, \ldots 24 \}$ & $\emptyset$ & $\emptyset$ & $\{25\}$ \\ 
    25 & $\{1, \ldots 16\}$ & $\{17, \ldots 24 \}$ & $\emptyset$ & $\emptyset$ & $\emptyset$ \\ 
    24 & $\{1, \ldots 16\}$ & $\emptyset$ & $\{17,\ldots 20\}$ & $\{21, 22\}$ & $\{23\}$ \\ 
    23 & $\{1, \ldots 16\}$ & $\emptyset$ & $\{17,\ldots 20\}$ & $\{21, 22\}$ & $\emptyset$ \\ 
    22 & $\{1, \ldots 16\}$ & $\emptyset$ & $\{17,\ldots 20\}$ & $\emptyset$ & $\{21\}$ \\ 
    21 & $\{1, \ldots 16\}$ & $\emptyset$ & $\{17,\ldots 20\}$ & $\emptyset$ & $\emptyset$ \\ 
    20 & $\{1, \ldots 16\}$ & $\emptyset$ & $\emptyset$ & $\{17, 18\}$ & $\{19\}$ \\ \hline
\end{tabular}
\end{table}

For any $q_{min} \leq \lfloor \log_2(n-1) \rfloor$  \begin{equation*}
\bigcup_{q=q_{min}}^{\lfloor \log_2(n-1) \rfloor} \mathcal{U}_n^{q} \quad \cup \quad \{\tau | \sum_{q'=q_{min}}^{\lfloor \log_2(n-1) \rfloor} \overline{(n-1)}_{q'} 2^{q'}
<\tau \leq n-1 \}, 
\end{equation*}
is a partition of $\{1, \ldots, n-1\}$. The right-side set of this partition has at most $2^{q_{min}}$ elements.
We use this partition to update the set
of candidates. To be specific our algorithm maintains

\begin{equation*}
\mathcal{T}_n  =  \bigcup_{q=q_{min}}^{\lfloor \log_2(n-1) \rfloor} \mathcal{W}_n^q  \quad \cup 
\quad  \left(
\{ P(\tau) | \sum_{q'=q_{min}}^{\lfloor \log_2(n-1) \rfloor} \overline{(n-1)}_{q'} 2^{q'}
< \tau \leq n-1 \} 
\right)\,,
\end{equation*}
with $\mathcal{W}_n^q = \text{\textsc{Hull}} \left( \{P(\tau)\}_{\tau \in \mathcal{U}_n^{q} }\right)$.
Note that $\mathcal{T}_n$ contains all the points (or index of the points) on the hull of $\{P(\tau)\}_{\tau \in \{1,\dots, n-1\}}$. 

In the next paragraphs, we bound the expected number of points in $\mathcal{T}_n$ and then explain how we compute $\mathcal{T}_{n+1}$ from $\mathcal{T}_{n}$.


\begin{lemma}\label{lemma:bound_numberofcandidate}
    For a fixed $q_{min}$, under Assumption \ref{assump:Vn_complexity} we have
    \begin{equation}
\label{eq:bound_numberofcandidate}
\mathbb{E}(|\mathcal{T}_n|) \leq  2c_1 (\log n)^{p+1} +   2c_2\log n + 2^{q_{min}}.
\end{equation}
\end{lemma}

\begin{proof}
For a fixed $q_{min}$ using Assumption \ref{assump:Vn_complexity} for all $\mathcal{U}_n^q$ we get that 
\begin{eqnarray*}
\mathbb{E}(|\mathcal{T}_n|) & \leq & \sum_{q=q_{min}}^{\lfloor \log_2(n-1) \rfloor} \left(c_1 (\log 2^q)^{p} + c_2\right) +  2^{q_{min}} \\
& \leq & \sum_{q=q_{min}}^{\lfloor \log_2(n-1) \rfloor} \left(c_1 (\log n)^{p} + c_2\right) +  2^{q_{min}} \\
& \leq &  2c_1 (\log n)^{p+1} +   2c_2\log n + 2^{q_{min}}.
\end{eqnarray*}
We get the last inequality using the fact that we sum over at most $\log_2(n)$ terms and $1/\log(2) \leq 2$.
\end{proof}

\subsection{Dyadic algorithm: Updating the set of candidates}

To compute $\mathcal{T}_{n+1}$ from scratch we need to apply the hull algorithm on all $\mathcal{U}_{n+1}^q$ for $q\ge q_{min}$. Assuming we have already computed $\mathcal{T}_n$ computing $\mathcal{T}_{n+1}$ is a bit easier. This is because it is often the case that $\mathcal{U}_{n+1}^q =\mathcal{U}_{n}^q$.
Therefore we directly get the indices of  the point on the hull of $\mathcal{U}_{n+1}^q$ as $\mathcal{U}_{n}^q \cap \mathcal{T}_n$ because
\begin{equation*}
    \mathcal{W}_n^q=\text{\textsc{Hull}} \left( \{P(\tau)\}_{\tau \in \mathcal{U}_n^{q} }\right) = \mathcal{U}_n^q \cap \mathcal{T}_n\,. 
\end{equation*}

Algorithm \ref{dyadic_online} formalise this idea and iteratively merge chunks of size $2^q$ to take advantage of Lemma \ref{lemma:bound_fusion}.

\begin{algorithm}[ht]
\caption{Update of the set of candidates}
\label{dyadic_online}
\begin{algorithmic}[1]
\State {\bf Input :} $n$ and $\mathcal T = \mathcal{T}_n$
\State $ \mathcal T \gets \mathcal T \cup \{n \}$ 
    \State $q \gets q_{min}$
     \While{$remainder(n / 2^q) = 0$}
        \State $ \mathcal T' \gets \{ \tau \in \mathcal T | \tau > n - 2^{q}\}$ 
        \State $ \mathcal T'' \gets \Call{Quickhull}{\{P(\tau)\}_{\tau \in \mathcal T'}}$  
        \State $\mathcal T \gets \mathcal T \setminus (\mathcal T' \setminus \mathcal T'')$ \Comment{discard vertices in $\mathcal T'$ not in $\mathcal T''$}
        \State $q \gets q+1$
      \EndWhile
\State \Return $\mathcal{T}= \mathcal{T}_{n+1}$
\end{algorithmic}
\end{algorithm}

\begin{remark}
In Algorithm \ref{dyadic_online} we assume that if the set of points ${\{P(\tau)\}_{\tau \in \mathcal T'}}$ doesn't define a proper volume in dimension $p+1$ then the output of $\Call{Quickhull}{\{P(\tau)\}_{\tau \in \mathcal T'}}$ is simply $\mathcal T'$.
If we choose $q_{min}$ such that $2^{q_{min}}$ is larger than $p+2$ this is unlikely to happen for data drawn from a distribution with a continuous density.
\end{remark}

We now continue with an informal description of Algorithm \ref{dyadic_online}.
Looking at line 4 when we include a new data point $n+1$ two cases can happen.
\begin{itemize}
    \item $n$ is not a multiple of $2^{q_{min}}$. In that case for all $q \ge q_{min}$ we have $\mathcal{U}_{n+1}^q = \mathcal{U}_n^q$. We just need to add $n$ as a possible changepoint to $\mathcal{T}_n$ to recover $\mathcal{T}_{n+1}$. This was done in line 2.
    \item $n$ is a multiple of $2^{q_{min}}$. In that case for some $q \ge q_{min}$ we have that $\mathcal{U}_{n+1}^q$ is different from $\mathcal{U}_{n}^q$. We therefore need to update the set of points using the hull algorithm.
\end{itemize}

For the later case let us define $q_{n} \ge q_{min}$ as the largest integer such that $2^{q_{n}}$ divides $n$. Before we explain the update of lines 5, 6, and 7 note the following facts.
\begin{itemize}
\item In base 2, $n$ ends with a one followed by $q_{n}$ zeros. In other words $\overline{n}_{q_{n}} = 1$ and for all $q < q_{n}$ $\overline{n}_q=0$. Thus all $\mathcal{U}_{n+1}^q$ with $q < q_{n}$ should be empty and $\mathcal{U}_{n+1}^{q_{n}}$ is not.
\item In base $2$, $n-1$  ends with a zero followed by $q_{n}$ ones. In other words $\overline{(n-1)}_{q_{n}} = 0$ and for all $q < q_{n}$ $\overline{(n-1)}_q=1$. Thus all $\mathcal{U}_{n}^q$ with $q < q_{n}$ are not empty. 
\item We know the index of the points on the hull for all sets of indices $\mathcal{U}_{n}^q$ with $q < q_{n}$: 
\begin{equation*}
    \mathcal{W}_n^q = \text{\textsc{Hull}} \left( \{P(\tau)\}_{\tau \in \mathcal{U}_n^{q} }\right) = \mathcal{U}_n^q \cap \mathcal{T}_n.
\end{equation*}
\item Finally, $\mathcal{U}_{n+1}^{q_{n}} = (\cup_{q_{min} \le q < q_{n}} \mathcal{U}_{n}^{q}) \cup \{\tau | n - 2^{q_{min}} < \tau \leq n\}$. 
\end{itemize}

Our goal is to get the point (or indices of these points) on the hull of $P(\tau)$ for $\tau$ in $\mathcal{U}_{n+1}^{q_{n}}$. 
To exploit the merging bound of \eqref{eq:bound_fusion} Algorithm \ref{dyadic_online} iteratively applies the hull algorithm on larger and larger chunks of the data trimming the points that are not on the hull.
In the while loop, from lines 4 to 9, of Algorithm \ref{dyadic_online} we can distinguish two cases.

\begin{enumerate}
\item For $q=q_{min}$ on lines 5 and 6 the algorithm applies the hull algorithm on the $2^{q}$ right-most points: $\text{\textsc{Hull}}(\{P(\tau)\}_{n - 2^{q} < \tau \leq n})$. It then discards on line 7 all indices larger than $n-2^{q}$ that are not on this hull.
\item If $q_{min} < q \le q_{n}$ on lines 5 and 6 the algorithm applies the hull algorithm on indices that are larger than $n-2^{q}$. In essence it is applying the hull algorithm on the points computed during the previous while step (indices larger than $n-2^{q-1}$ and in $\mathcal T$) with those of $\mathcal{W}_n^q = \text{\textsc{Hull}} \left( \{P(\tau)\}_{\tau \in \mathcal{U}_n^{q} }\right) = \mathcal{U}_n^q \cap \mathcal{T}_n$ (indices between $n-2^{q}$ and $n-2^{q-1}$, and in $\mathcal T$). A bit more formally, calling $\mathcal{W}'^{q}_{n+1}$ the output of line 6 ($\mathcal T''$) at step $q$ of the while loop, we are merging $\mathcal{W}_n^q$ and $\mathcal{W}'^{q-1}_{n+1}$ as follows:
\begin{displaymath}
    \mathcal{W}'^{q}_{n+1} = \text{\textsc{Hull}} \left( \{P(\tau)\}_{\tau \in \mathcal{W}_n^{q} \cup \mathcal{W}'^{q-1}_{n+1} }\right).
\end{displaymath}
On line 7 the algorithm is discarding all indices larger than $n-2^{q}$ that are not on this hull $\mathcal{W}'^{q}_{n+1}$.
\end{enumerate}

Proceeding in such a way at each step we merge the points on the hull of two successive chunks of size $2^{q-1}$ to recover the points on the hull of a chunk of size $2^q$. Figure \ref{fig:scheme} is a schematic representation of these successive fusions (applications of the hull algorithm).

\begin{figure*}[ht]%
\centering
\makebox{\includegraphics[width=.9\linewidth]{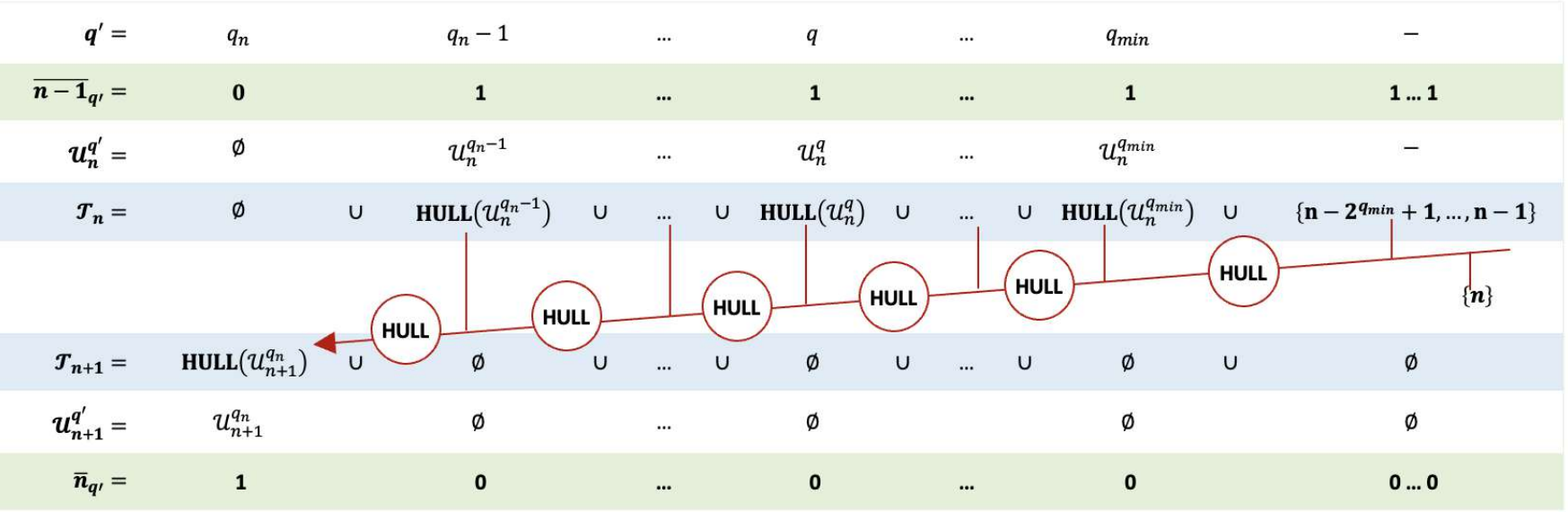}}
\caption{Schematic representation of the while loop of Algorithm \ref{dyadic_online} for $q_{n}$ being the largest $q$ such that $2^q$ divides $n$. First applying the hull algorithm on $\{n - 2^{q_{min}}+1, \ldots, n\}$ (on the right) following the red arrow we iteratively aggregate the result with the hull of points in $\mathcal{U}_n^{q}$  to recover the point on the hull of $\mathcal{U}_{n+1}^{q_{n}}$ (on the left). In such a way we effectively retrieve $\mathcal{T}_{n+1}$ as all $\mathcal{U}_{n+1}^{q}$ are empty for $q < q_{n}$.
}\label{fig:scheme}
\end{figure*}

\begin{example}
In this paragraph we explain how the algorithm proceeds for $n=24$ and $q_{min}=2$.
In base 2, $n-1=10111$. Therefore (see the definition of set $\mathcal{T}_n$)  we already computed the points on the hull of $P(\tau)$ for  $\tau$ in $\mathcal{U}_{24}^4 = \{1, \ldots 16\}$ and $\tau$ in $\mathcal{U}_{24}^2 = \{17, \ldots, 20\}$, and additionally we store all indices $\tau$ larger than $24 - 2^2$ : $\{21, 22, 23\}$.
When we include a new observation $n+1=25$, we need to include a possible change at $n=24$, and here is what we do with the while loop (starting with $q=2$):
\begin{itemize}
    \item for $q=2$: we run \texttt{Quickhull} on points with indices strictly larger than 20 : $\{21, 22, 23, 24\};$
    \item for $q=3$: to recover $\mathcal{U}_{25}^3$ we run \texttt{Quickhull} on points with indices strictly larger than 16 that is the points obtained in the previous step ($q=2$) union those of the hull of $\mathcal{U}_{24}^2$ (that were computed for $n+1=21$ as the point on the hull of $\mathcal{U}_{21}^2$). 
    \item for $q=4$: we stop as the remainder of $24/2^4$ is not 0.
\end{itemize}
\end{example}

We are now ready to give the proof of Lemma \ref{lemma:bound_update_dyadic} in the main text, which we rewrite here to ease reading.
\newcounter{originallemma}
\setcounter{originallemma}{\value{lemma}}
\setcounter{lemma}{0}
\begin{lemma}
Under  Assumptions \ref{assump:Vn_complexity} and \ref{assump:Time_complexity},
the expected time complexity of Algorithm \ref{dyadic_online} is 
   $\mathcal{O}\left(n \log(n)^{p+1}\right)$.
\end{lemma}
\setcounter{lemma}{\value{originallemma}}

\begin{proof}
Considering all time steps up to $n+1$ the algorithm used the \texttt{QuickHull} algorithm on all successive (and nonoverlapping) chunks of size $2^q$ for $q$ larger than $q_{min}$ and smaller than $\lfloor \log_2(n) \rfloor$.

For the smallest scale, $n'=2^{q_{min}},$ we have $n/2^{q_{min}}$ chunks and run \texttt{QuickHull} on all points. Using Assumption \ref{assump:Time_complexity}, the worst time per chunk is less than $c_3 (2^{q_{min}})^2+c_4,$ and summing overall chunks we get 
$ n c_3 2^{q_{min}}+\frac{nc_4}{2^{q_{min}}}$ operations. So at the smallest, scale the computational cost is $\mathcal{O}(n)$.

We will now successively consider all larger scales $n'=2^q$ for $q_{min} < q \le \lfloor \log_2(n) \rfloor$.
The hull of a chunk of size $n'=2^q$ is obtained by combining the hulls of two smaller chunks of size $2^{q-1}$.
For any such chunk, we can thus exploit Lemma \ref{lemma:bound_fusion} and  \eqref{eq:bound_fusion} with $|\mathcal{W}_1|=|\mathcal{W}_2|=n'/2=2^{q-1}$ and we get
\begin{equation*}
   \mathbb{E}\left(Time(|\mathcal{W}_1|+|\mathcal{W}_2|) \right)  \leq   4c_1c_3 2^{q-1} (\log 2^{q-1})^p + 4c_2c_3 2^{q-1}+c_4\,.
\end{equation*}
Multiplying by $n/2^q$ to consider all chunks of size $2^q$, we get that the time to process scale $2^q$ is 
\begin{eqnarray*}
   \frac{n}{2^q} \left(4c_1c_3 2^{q-1} (\log 2^{q-1})^p + 4c_2c_3 2^{q-1}+c_4\right) 
    & \leq &   4c_1c_3 \frac{n}{2} (\log 2^{q-1})^p + 4c_2c_3 \frac{n}{2}+c_4\frac{n}{2^q}
    \\
    & \leq &   4c_1c_3 \frac{n}{2} (\log n)^p + 4c_2c_3 \frac{n}{2}+c_4\frac{n}{2^q}\,.
\end{eqnarray*}

So the computational cost at scale $2^q$ is  $\mathcal{O}(n \log(n)^{p})$. Multiplying by $\log_2(n)$ to consider all scales, we get the desired result.
\end{proof}

\subsection{MdFOCuS with dyadic update}

The MdFOCuS Algorithm implementing this update of the set of candidates is presented in Algorithm \ref{MdFOCuS_algo_dyadic}. 
Combining Lemma \ref{lemma:bound_numberofcandidate} and \ref{lemma:bound_update_dyadic} we get an overall complexity that is $\mathcal{O}(n \log(n)^{p+1})$.

\begin{algorithm}[ht]
\caption{MdFOCuS Dyadic algorithm}
\label{MdFOCuS_algo_dyadic}
\begin{algorithmic}[1]
\State {\bf Input 1:}  $\{x_t\}_{t=1,2,\dots}$ $p$-variate independent time series; $thrs$ threshold; $\eta_1$ pre-change parameter (if known); $q_{min}$ min power 2 scale at which we run \texttt{QuickHull}
\State {\bf Outputs:} stopping time $n$ and maximum likelihood change: $\hat{\tau}(.)$ or $\hat{\tau}(\eta_1)$
    \State $\LLR \gets -\infty,\quad \mathcal T \gets \emptyset, \quad n \gets  1$

\While{($\LLR < thrs$)}
     
     \State $n \gets n+1$
     \State $\mathcal T \gets \mathcal T \cup \{n-1\}$

     \State $\hat{m} \gets \max_{\tau \in \mathcal T} \{\max_{\eta_1, \eta_2} \ell_{\tau,n} (\eta_1,\eta_2)\}$ \Comment{ $\eta_1$ is fixed if pre-change known}
    \State $\LLR \gets \hat{m} - \max_{\eta_1, \eta_2} \ell_{n,n} (\eta_1,\eta_2)$  \Comment{ $\eta_1$ is fixed if pre-change known}
    \State $q \gets q_{min}$
     \While{$\Call{remainder}{\frac{n-1}{2^q}} = 0$ }
        \State $ \mathcal T' \gets \{ \tau \in \mathcal T | \tau > n - 2^{q}\}$ 
        \State $ \mathcal T'' \gets \Call{Quickhull}{\{P(\tau)\}_{\tau \in \mathcal T'}}$ 
        \State $\mathcal T \gets \mathcal T \setminus (\mathcal T' \setminus \mathcal T'')$ \Comment{discard vertices in $\mathcal T'$ not in $\mathcal T''$}
        \State $q \gets q+1$
      \EndWhile
\EndWhile    
\State \Return $n$ and $\LLR$
\end{algorithmic}
\end{algorithm}

We implemented Algorithm \ref{MdFOCuS_algo_dyadic} for a change in the mean of multi-dimensional Gaussian signal in R/C++ using the \texttt{qhull} library.
We evaluate the empirical run times of Algorithm \ref{MdFOCuS_algo_dyadic}
with various values of $q_{min}$ using the profiles of Supplementary Material \ref{append4_Runtime_alpha}. In Figure \ref{GM_qmin_runtime} we report the average run time as a function of the dyadic parameter $q_{min}$ for MdFOCuS for a known value of $\eta_1$ with the dyadic update and values of $q_{min}$ in 
${3,\dots,12}$. We observe a minimum run time for $q_{min}=6$ for $p=1$, $q_{min}=7$ for $p=2$ and $q_{min}=8$ for $p=3$. As a consequence in the rest of the simulations we always set $q_{min}$ to $p+5$.

Using the simulation framework of Section \ref{sec-EmpiricalRunTimeGauss} we evaluated the run times of MdFOCuS with dyadic update for dimension $p=1$ to $p=3$ and compared it with run times of MdFOCuS. The average run times are presented in Figure \ref{GM_DyadicRuntime}. Dyadic MdFOCuS is slower than MdFOCuS, by about a factor of 2.

\begin{figure*}[ht]%
   \begin{minipage}{0.48\textwidth}
     \centering
     \includegraphics[width=0.75\linewidth]{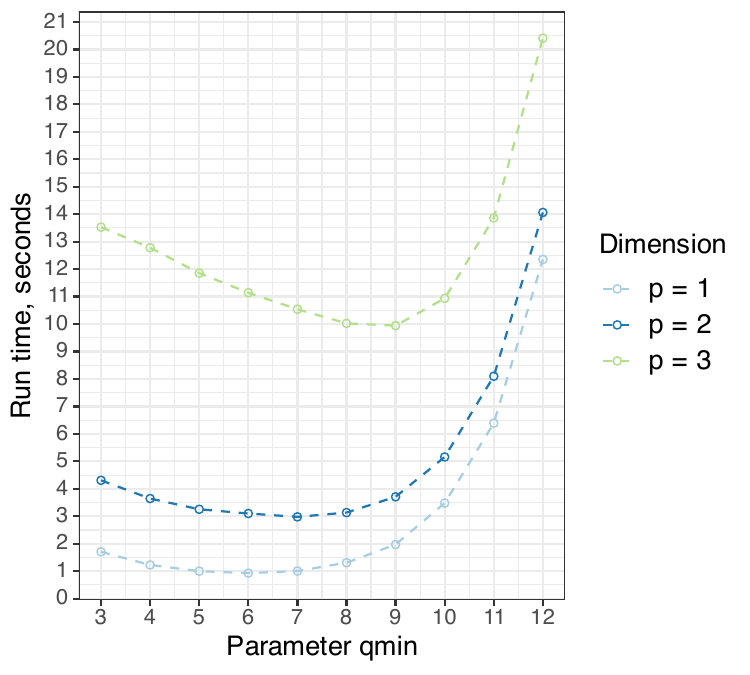}
     \caption{Run time of MdFOCuS with dyadic update as a function of $q_{min}$ for $p$-variate time series with $n=10^5$ data points in dimension $p =1, 2$ and $3$. We simulated $100$ i.i.d. Gaussian data $\mathcal N_p(0,I_p)$ and report the average run time.
}
\label{GM_qmin_runtime}
   \end{minipage} \hfill
   \begin{minipage}{0.48\textwidth}
     \centering
     \includegraphics[width=0.95\linewidth]{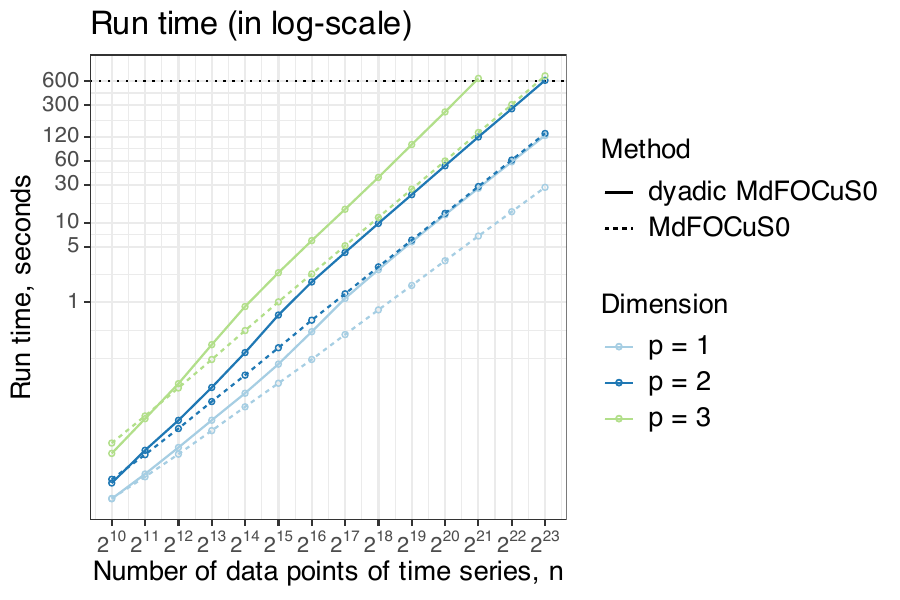}\caption{Run times in seconds of MdFOCuS (dotted) and MdFOCuS with dyadic update (full) both with $q_{min} = 5+p$ in dimension $p=1,2$ and $3$ using time series $x_{1:n}$ (without change, $x_t \sim \mathcal N_p (0,I_p)$ i.i.d.). Run times are averaged over $100$ data sets. We considered a maximum running time of $10$ minutes for the algorithms. That is why we do not report run times for $p=3$ for large signals. }\label{GM_DyadicRuntime}
   \end{minipage}
\end{figure*}

\section{An approximation for large $p$}\label{heuristic_section}

Based on Theorem \ref{th:boundnumbervertexes} and its corollary at step $n$ we expect Algorithm \ref{MdFOCuS_algo} to store $\mathcal{O}(\log^p(n))$ candidates. Considering the calculation of the full and sparse maximum likelihood in lines 7 and 8 should thus have a complexity of  $\mathcal{O}(p\log^p(n)).$
The power $\log(n)^p$ is problematic making the algorithm impractical for $p$ larger than $6$, as exemplified in Figure \ref{GM_Runtime}. 

Here we propose a heuristic solution to this problem. Our idea is to get a subset of the point on the hull considering projections on at most $\tilde{p}$ variates.

For any ordered subset $S=\{S_1, \ldots, S_{|S|}\}$ of $\{1, \ldots, p\}$, with $|S|$ the size of $S$, we define $x_t^{[S]}$ the vector in $\mathbb{R}^{|S|}$ such that the $i-th$ coordinate of $x_t^{[S]}$ is the $S_i$-coordinate of $x_t$. We then define $P^{[S]}(\tau)$ accordingly:
\begin{equation*}
    P^{[S]}(\tau) = \left(\tau, \sum_{t=1}^{\tau} x_t^{[S]}\right).
\end{equation*}
Similarly, we call $\mathcal{T}^{[S]}_n$ the set of index $\tau$ such that $P^{[S]}(\tau)$ is on the hull of $\{P^{[S]}(\tau)\}_{\tau \in\{1, \ldots, n-1\}}.$
We now consider $\mathcal{S}$ a subset of all subsets of $\{1, \ldots, p\}$ (i.e $\mathcal{P}(\{1, \ldots, p\})$ or $2^{\{1, \ldots, p\}}$) and we get an inner approximation of $\mathcal{T}_n$ taking a union over all $\mathcal{T}^{[S]}_n$: $$ \cup_{S \in \mathcal{S}}\mathcal{T}^{[S]}_n \subseteq \mathcal{T}_n.$$

In our simulations as a minimalist implementation of this approximation we considered $\tilde{p} = 2$ and subsets of the form $\{\tilde{p}j+1, \ldots, \tilde{p}j+\tilde{p}\}$ for $j \leq p/\tilde{p}.$ We have $p/\tilde{p}$ such sets and under the condition of Theorem \ref{th:boundnumbervertexes} we thus expect to store $\mathcal{O}(p\log^{\tilde{p}}(n))$ candidate changepoints rather than $\mathcal{O}(\log^p(n))$ with the exact Algorithm \ref{MdFOCuS_algo}.
For $p=100$ and $n=10000$, our approximation with $\tilde{p}=2$, implemented in R runs on average in less than 100 seconds. More generally it can be seen in Figure \ref{fig:runtime_heuristic} that run times of our R implementation of MdFOCuS, and our 2d approximation for larger $p$ are comparable to \texttt{ocd} \cite{chen2020highdimensional}.

\begin{algorithm}[ht]
\caption{MdFOCuS $\tilde{p}$d approximation}
\label{MdFOCuS_approx}
\begin{algorithmic}[1]
\State {\bf Input 1:} $p$-variate independent time series $\{x_t\}_{t=1,2,\dots}$
\State {\bf Input 2:} threshold $thrs$ and $\eta_1$ if known pre-change parameter
\State {\bf Input 3:} $\mathcal{S}$ a subset of all subsets of $\{1, \ldots, p\}$, with for all $S \in \mathcal{S}$ :  $|S| \leq \tilde{p}$
\State {\bf Input 4:} For all $S \in \mathcal{S}$ a list of index $\mathcal{T}^{[S]}$ 
\State {\bf Input 5:} For all $S \in \mathcal{S}$ a $S$ specific pruning index $maxSize^{[S]}$ 
\State {\bf Output:} stopping time $n$ and approximated maximum likelihood change : $\hat{\tau}(.)$ or $\hat{\tau}(\eta_1)$
    \State $\LLR \gets -\infty,\quad \text{For all } $S$ \quad \mathcal{T}^{[S]} \gets  \emptyset, \quad n \gets  0$ 
\While{($\LLR < thrs$)}  
     \State $n \gets n+1$
     \For {all $S \in \mathcal{S}$}
     \State $\mathcal T^{[S]} \gets \mathcal T^{[S]} \cup \{n-1\}$
     \EndFor
     \State $\hat{m} \gets \max_{\tau \in \cup_{S \in \mathcal{S}}\mathcal T^{[S]}} \{\max_{\eta_1, \eta_2} \ell_{\tau,n} (\eta_1,\eta_2)\}$ \Comment{$\eta_1$ is fixed if pre-change known}
    \State $\LLR \gets \hat{m} - \max_{\eta_1, \eta_2} \ell_{n,n} (\eta_1,\eta_2)$  \Comment{$\eta_1$ is fixed if pre-change known}
     \For{all $S \in \mathcal{S}$}
        \If{($|\mathcal{T}^{[S]}| > maxSize^{[S]}$)}
        \State $ \mathcal{T}^{[S]} \gets \Call{Quickhull}{\{P^{[S]}(\tau)\}_{\tau \in \mathcal T^{[S]}}}$
         \State $maxSize^{[S]} \gets   \lfloor \alpha|\mathcal{T}^{[S]}| + \beta \rfloor $ 
        \EndIf
        \EndFor
\EndWhile
\State \Return $n$ and  $\LLR$
\end{algorithmic}
\end{algorithm}


\section{On the probability of false alarms, average-run-length and average detection delay in the multivariate Gaussian case}\label{app:stat_control}

\subsection{Definitions and problem statement} \label{app:stat_control_definition}
In this section, we consider the multivariate Gaussian case. That is, we consider a sequence $\{y_t\}_{t=1,2, \ldots}$ with elements
in $\mathbb{R}^p$, with $y_t= \mu_t + \varepsilon_t$,  and $\varepsilon_t$ i.i.d $\mathcal{N}_p(0, I_p)$.
At time $n$, for any possible change $\tau < n$, we consider two types of statistics: 
\begin{itemize}
    \item Ranked-based statistics, $\RankS{s}{\tau}{n}$, considering the $s$ dimensions with the most evidence for a change at $\tau$;
    \item Threshold-based statistics, $\ThreS{a}{\tau}{n}$, considering, as in \cite{chen2020highdimensional}, only the dimensions for which the evidence for a change exceeds the threshold $a$.
\end{itemize}

To formally present these statistics, we first define the unsorted per-dimension evidence for a change at $\tau$ at time $n$ knowing the pre-change mean to be 0, $\Evid{i}{\tau}{n}(0)$ or not knowing the pre-change mean, $\Evid{i}{\tau}{n}(.)$:
\begin{eqnarray*}
   \Evid{i}{\tau}{n}(0)  =   \frac{1}{\sqrt{n-\tau}}  \sum_{t=\tau+1}^n \varepsilon^i_t  ,\qquad 
   \Evid{i}{\tau}{n}(.)  =   \sqrt{\frac{\tau(n-\tau)}{n}}\left(\frac{1}{n-\tau}  \sum_{t=\tau+1}^n \varepsilon^i_t-\frac{1}{\tau}  \sum_{t=1}^\tau \varepsilon^i_t  \right) ,\\
   \end{eqnarray*}
where $\varepsilon^i_t$ is the $i-$th coordinate of $\varepsilon_t$ in $\RR^p$. We also define the sorted per-dimension evidence for a change at $\tau$ at time $n$, $\RankEvid{j}{\tau}{n}(0)$ as the $j$-th highest value out of all $\Evid{i}{\tau}{n}(0)$. Similarly, for the unknown pre-change mean scenario, we define $\RankEvid{j}{\tau}{n}(.)$ as the $j$-th highest value out of all $\Evid{i}{\tau}{n}(.)$. Based on that, we now define our rank and threshold-based statistics knowing the pre-change mean to be 0 as:
\begin{equation}\label{eq:rankAndThresholdStats}
       \RankS{s}{\tau}{n}(0)  =  \sum_{j=1}^s (\RankEvid{i}{\tau}{n}(0))^2 \quad \text{and} \quad
   \ThreS{a}{\tau}{n}(0) =  \sum_{i=1}^p (\Evid{i}{\tau}{n}(0))^2 \mathbbm{1}_{|\Evid{i}{\tau}{n}(0)| \geq a}\, ,
\end{equation}
where $\mathbbm{1}$ denotes the indicator function.
We define similar statistics if the pre-change mean is not known:
\begin{equation}\label{eq:rankAndThresholdStats2}
       \RankS{s}{\tau}{n}(.)  =  \sum_{j=1}^s (\RankEvid{i}{\tau}{n}(.))^2 \quad \text{and} \quad
   \ThreS{a}{\tau}{n}(.) =  \sum_{i=1}^p (\Evid{i}{\tau}{n}(.))^2 \mathbbm{1}_{|\Evid{i}{\tau}{n}(.)| \geq a}\,.
\end{equation}

For simplicity, when this does not cause confusion, we will omit the $(0)$ and $(.)$ notations. It will sometimes make sense to consider separately the ranked-1 statistics $\RankS{1}{\tau}{n}$. We will refer to $\RankS{p}{m}{n} = \ThreS{a=0}{m}{n}$ as the dense statistic.

First we summarise our results:
\begin{enumerate}
    \item For all ranked and thresholded statistics, we provide, in Section \ref{app-sec:falsealarm}, time-varying thresholds to control their false alarm rate. To be specific, considering a generic ranked or thresholded statistic $\tilde S_{\tau,n}$, we provide a varying threshold $c_n(\alpha)$ for which:
    \[ \mathbb{P}( \exists \ \tau < n \in \mathbb{N} \ \text{such that }\, \tilde S_{\tau,n} \geq  c_n(\alpha) ) \ \leq \ \alpha.\]
    Importantly, our thresholds are also valid if the pre-change mean is unknown.
    As a byproduct, we can calibrate a fixed threshold to control the Average Run Length, $\gamma$. That is, we can define a threshold $c(\gamma)$ such that if all $\mu_t$ are equal to $0$ then $E(\tilde T) \geq \gamma, $ where $\tilde T$ is the first time our generic statistic passes the threshold.
    
    \item Second, in Section \ref{app-sub-ADD}, assuming the pre-change mean is known, and for a fixed threshold $c$, we control using the optional stopping time theorem, the Average Detection Delay (ADD). To be specific, assuming $\mu_t = 0$ for all $t \leq  \tau^*$ and $\mu_t = \delta$ otherwise we get a bound on 
    $$\sup_{y_1,\ldots,y_{\tau^*}} \mathbb{E}( \tilde T-\tau^*|\tilde T\geq \tau^*),$$
    involving only $\lVert \delta \rVert$ and the threshold.
   For thresholded statistics, the same result holds by replacing $c$ by $c+pa^2$ (see remark \ref{rem:addthreshold}).
   
    \item Finally, in Section \ref{app-sub-ADD_Bound}, for all our ranked and thresholded statistics, assuming we know the pre-change mean, we provide probabilistic bounds on the Detection Delay.
    To be specific, again assuming $\mu_t = 0$ for all $t \leq  \tau^*$ and $\mu_t = \delta$ otherwise, we bound with probability $1- \alpha$ the detection time, $\tilde T - \tau^*$, by a term involving only $\lVert \delta \rVert$, the threshold, and $\alpha$.
\end{enumerate}


\subsection{Probability of false alarms for all statistics}\label{app-sec:falsealarm}
\begin{proposition}\label{coro:H0_all}
     For a sequence without change : $y_t=\varepsilon_t$ with $\varepsilon_t$ i.i.d. $\mathcal{N}_p(0, I_p)$.
    Considering a time-varying $x_{n}(\alpha) = 4\log(n) - \log(\alpha),$ we can control the false alarm rate of ranked and thresholded statistics at level $\alpha$. 
    That is for a generic statistic $\GenS{\tau}{n}$ we control:
$$\mathbb{P}( \exists \ \tau < n \in \mathbb{N} \ \text{such that }\, \GenS{\tau}{n} \geq c_n(\alpha) ) \ \leq \ \alpha,$$
using the thresholds in table \ref{tab:statistics_bounds}.

\begin{table}[h]
    \centering
    \renewcommand{\arraystretch}{2} 
    \begin{tabular}{|c|c|}
        \hline
        Statistic $\GenS{\tau}{n}$ & Threshold $c_n(\alpha)$ \\
        \hline
        $\RankS{1}{\tau}{n}$ &
        $2x_{n}(\alpha) + 2 \log(2p)$ \\
        \hline
        $\RankS{s}{\tau}{n}$ &
         $2(x_{n}(\alpha) + \log{p \choose s})   + 2\sqrt{s(x_{n}(\alpha) + \log{p \choose s})}  + s$ \\
        \hline
        $ \ThreS{a}{\tau}{n}$ &
        $4 x_{n}(\alpha) + 6pe^{-a^2/8}$\\
        \hline
    \end{tabular}
    \caption{Probability bounds for our statistics with $x_{n}(\alpha) = 4\log(n) - \log(\alpha)$}
    \label{tab:statistics_bounds}
\end{table}
\end{proposition}

\begin{proof}
We follow Step 1 and Step 2 of the proof of Theorem 1 in \cite{yu2023note} using the peeling concentration bounds of our Lemma \ref{lem:peeling_all_change}.
\end{proof}

Importantly, when using MdFOCuS, one typically considers several statistics. Denoting $\#Stats$ the number of statistics, we control simultaneously all of them at level $\alpha'$ taking $\alpha = \alpha'/\#Stats$ in the thresholds of Table \ref{tab:statistics_bounds}.

\begin{remark}\label{app-rem-ARL}
Using Proposition \ref{coro:H0_all} we can also control the Average Run Length (ARL) using a fixed threshold, $c$. Consider the Rank-1 statistic, $\RankS{1}{\tau}{n}$, as an example. We note that its threshold $2x_{n}(\alpha) + 2 \log(2p)$ increases with $n$. Thus, using a fixed threshold $c = 2x_{N}(\alpha) + 2 \log(p)$ for all time $n$, the probability of stopping before $N$ is less than $\alpha$, and therefore the ARL is larger than $(1-\alpha)N$. To control the ARL at any desired level $\gamma$, it suffices to choose a sufficiently large $N$ and a sufficiently small probability $\alpha$ such that $(1-\alpha)N\geq \gamma$.
\end{remark}

\subsection{Average detection delay for a fixed threshold and knowing the pre-change mean}\label{app-sub-ADD}

In this section, we assume that a change in mean occurs immediately after time $\tau^*$. Assume that we are using the dense statistic $\RankS{p}{\tau}{n}(0) = \ThreS{0}{\tau}{n}(0)$ with a fixed threshold (following Remark \ref{app-rem-ARL}) and let $T$ be the first time that our test statistic exceeds the threshold. The average detection delay is defined as
\begin{equation*}
\mbox{ADD}=\sup_{y_1,\ldots,y_{\tau^*}} \mathbb{E}( T-\tau^*|T\geq \tau^*).
\end{equation*}
That is, it is the worst-case expected time from the change until it is detected, conditional on no detection before $\tau^*$. 

\begin{lemma}\label{lemma:ADD}
For a change of size $\delta$, with $\lVert \delta \rVert^2$ its Euclidean squared norm, the average detection delay for the dense statistic with a known pre-change mean ($\RankS{p}{\tau}{n}(0) = \ThreS{0}{\tau}{n}(0)$) and a fixed threshold $c$ satisfies
\[
\mbox{ADD}\leq \frac{c+\lVert\delta\rVert\sqrt{8/\pi}}{\lVert\delta\rVert^2}+1.
\]
\end{lemma}
\begin{proof}
Let $\tau^*$ be the time of the change, and $T$ the time at which the dense statistic first exceeds $c$. Consider  the test statistic for a change at $\tau$ of size $\delta$, which at time $\tau^*+t$, for $t=1,2,\ldots$ is of the form:

\begin{equation} \label{eq:ADD1}
A_t= \sum_{i=1}^{t} \left(2\delta^Ty_{\tau^*+i}-\lVert \delta \rVert^2\right).
\end{equation}

Note that $\max_{\delta} A_t  = \RankS{p}{\tau^*}{\tau^*+t}$. 
Define $A_0=0$.

Now for any data $y_{1:t+\tau^*}$ we have $A_t\leq \max_{\tau=0}^{\tau+t-1} \RankS{p}{\tau}{t+\tau^*}(0)$. Let $T'$ be the time that $A_t$ first exceeds $c$. It follows that conditional on $T>\tau^*$, we have $T'\geq T-\tau^*$. Thus we can bound the $\mbox{ADD}$ by $\mathbb{E}(T')$.

Denote each term in the sum on the right-hand side of (\ref{eq:ADD1}) by $a_i$. We have these terms are independent, normally distributed with mean $\mathbb{E}(a_i)=\delta^T\delta=\lVert\delta\rVert^2$ and variance $4\lVert\delta\rVert^2$. 
For each $t$ it is trivial to see that $\mathbb{E}(|A_t|)<\infty$, and thus $A_t-t\lVert\delta\rVert^2$ is a Martingale. It is also straightforward to show that $T'$ is a stopping time and $\mathbb{E}(T')<\infty$ so we can apply the optional stopping theorem. This gives
\[
\mathbb{E}\left( A_{T'}-T'\lVert\delta\rVert^2\right)=0.
\]
Thus $\mathbb{E}(T')=\mathbb{E}(A_{T'})/\lVert\delta\rVert^2$. We can bound $\mathbb{E}(A_{T'})$ from above by considering the distribution of $A_{T'}$ conditional on $A_{T'-1}$. We have
\[
\mathbb{E}(A_{T'}|A_{T'-1}=A) = \mathbb{E}\left(A+\lVert\delta\rVert^2+\sqrt{4\lVert\delta\rVert^2}\epsilon \ | \ A+\lVert\delta\rVert^2+\sqrt{4\lVert\delta\rVert^2}\epsilon\geq c\right),
\]
where $\epsilon$ has a standard normal distribution.
This expectation is monotone increasing in $A$, and as $A<c$ we can bound $\mathbb{E}(A_{T'})$ by the value of the conditional expectation in the limit as $A\rightarrow c$ from below. This gives
\[
\mathbb{E}(A_{T'}) \leq \mathbb{E}(c+\lVert\delta\rVert^2+\sqrt{2\lVert\delta\rVert^2}|\epsilon|) = c+\lVert\delta\rVert^2+\sqrt{2\lVert\delta\rVert^2}\sqrt{\frac{2}{\pi}}.
\]
The result follows.
\end{proof}

\begin{remark}\label{rem:addranked}
It is simple to extend this result to get a bound on the ADD of any ranked statistics $\RankS{s}{\tau}{n}(0)$, because the same argument applies to any statistic that considers a given subset of $s$ variates changing.
Following the notation introduced in Appendix G (on the approximation), denoting $R$ the subset of $s$ variates that are changing, we would recover the same formula considering $y^{[R]}$ in $\mathbb{R}^s$ and replacing $\lVert \delta \rVert$ by $\lVert \delta^{[R]} \rVert$.
\end{remark}

\begin{remark}\label{rem:addthreshold}
The proof also applies to any thresholded statistics with threshold $a$.  We replace $c$ by $c+pa^2$ in the proof, note that $\RankS{p}{\tau}{n} - pa^2$ is always smaller than $\ThreS{a}{\tau}{n}$ and the bound follows.
\end{remark}

\subsection{Probabilistic bound on the Detection Delay for a fixed  threshold and knowing the pre-change mean}\label{app-sub-ADD_Bound}

In this section, we again assume that a change in mean occurs immediately after time $\tau^*$, calling $T$ the first time our statistic passes the threshold $c_n$ we give probabilistic bounds on the detection delay $T-\tau^*$. We start with ranked statistics and then proceed with thresholded statistics.

\paragraph{}

\begin{proposition}\label{coro:H1_ranked}
     Consider a change of size $\delta$ at time $\tau^*$, that is $y_t=\varepsilon_t$ for $t \leq \tau^*$ and $y_t=\delta + \varepsilon_t$ for $t > \tau^*$ with $\varepsilon_t \sim \mathcal{N}(0, I_p)$ and $\lVert \delta \rVert > 0$. We denote $T$ the first time the dense statistic, $\RankS{p}{\tau}{n}(0)$, passes the fixed positive threshold $c$. 
     With probability at least $(1-\alpha)$ the detection delay, $T-\tau^*$, is less than
     \begin{equation*}
     \frac{2 c - 8 \log(\alpha/2)}{\lVert \delta \rVert^2}
     \end{equation*}

     In addition, if $c>p$, then with probability at least $(1-\alpha)$, 
     \begin{equation*}
     T-\tau^*<\frac{2 (c-p) + 2\sqrt{p\log(3/\alpha)} - 8 \log(\alpha/3)}{\lVert \delta \rVert^2}
     \end{equation*}
\end{proposition}

\begin{proof}
   We consider the statistic $\RankS{p}{\tau^*}{n}(0)$ for increasing $n$ and apply Lemma \ref{lemma:H1_proba_bound} and find that for any time larger than the provided bound, the statistic would pass the threshold with probability at least $(1-\alpha)$. Hence, the detection delay would be smaller with probability $(1-\alpha)$.
\end{proof}

\begin{remark} \label{rem:DDtos}
It is simple to extend this result to get a bound on the detection delay of any ranked statistic $\RankS{s}{\tau}{n}(0)$. This is because the same argument applies to the statistic that considers any subset of $s$ variates changing.
Following the notation introduced in Appendix G (on the approximation), denoting $R$ the subset of $s$ variates that are changing, we would consider $y^{R}$ in $\mathbb{R}^s$ and recover the same formula replacing $\lVert \delta \rVert$ by $\lVert \delta^{[R]} \rVert$.
\end{remark}

    Using Lemma \ref{lemma:H1_proba_bound} we could get similar bounds for a time-varying threshold assuming it is positive, increasing, and concave thresholds.

Before providing bounds on the detection delay of thresholded statistics, we recall the notion of effective sparsity defined in \cite{chen2020highdimensional}. For any $\delta$ in $\mathbb{R}^p$ such that $\lVert \delta \rVert > 0$ the effective sparsity $z(\delta)$ is the smallest integer in $\{ 2^0, 2^1, \ldots ,2^{\lfloor \log_2(p)} \rfloor  \}$ such that the set $$\mathcal{Z}(\delta) = \{ j \ \text{such that} \ |\delta^j| > \lVert \delta \rVert / (z(\delta)\log_2(p))\},$$ has cardinality at least $z(\delta)$.

\begin{proposition}\label{coro:H1_thresholded}
  Consider a change of size $\delta$ at time $\tau^*$, that is $y_t=\varepsilon_t$ for $t \leq \tau^*$ and $y_t=\delta + \varepsilon_t$ for $t > \tau^*$ with $\varepsilon_t \sim \mathcal{N}(0, I_p)$ and $\lVert \delta \rVert > 0$. We denote $T$ the first time the thresholded statistic, $\ThreS{a}{\tau}{n}(0)$, passes the fixed threshold $c$. 
With probability at least $(1-\alpha)$ the detection delay, $T-\tau^*$, is less than than  
\begin{displaymath}
\log_2(p) \max \left\{ \frac{2c - 8 \log(\alpha) + 8\log(2(z(\delta)+1)))}{\lVert \delta \rVert^2 }, \frac{2a^2 - 8 \log(\alpha) + 8\log(2(z(\delta)+1)))}{\lVert \delta \rVert^2 / z(\delta) }\right\} 
\end{displaymath}

\end{proposition}

\begin{proof}
We consider a subset $\mathcal{Z'}$ of size $z(\delta)$ of $\mathcal{Z}(\delta)$. All variates $j$ in this subset validate the effective sparsity constraint: $|\delta_j| > \lVert \delta \rVert / (z(\delta)\log_2(p))$. We first lower bound the contribution to $\ThreS{a}{\tau}{n}$ of all variates not in $\mathcal{Z'}$ by $0$. Then we consider a time large enough to ensure that simultaneously \begin{itemize}
    \item the contribution of each variates in $\mathcal{Z'}$ passes the threshold $a$
    \item the sum of these $z(\delta)$ contributions passes the fixed threshold $c$.
\end{itemize} 
To this end we use Lemma \ref{lemma:H1_proba_bound} exactly $z(\delta)+1$ times:
\begin{itemize}
    \item for each variates in $\mathcal{Z'}$ the lemma applies replacing $\lVert \delta \rVert^2$ by $\frac{\lVert \delta \rVert^2}{z(\delta) \log_2(p)}$ and $\alpha$ by $\frac{\alpha}{z(\delta)+1}$ (in preparation for a union bound).
    \item for the sum the lemma applies replacing $\lVert \delta \rVert^2$ by $\frac{\lVert \delta \rVert^2}{\log_2(p)}$ and $\alpha$ by $\frac{\alpha}{z(\delta)+1}$. This is because, by definition of $z(\delta)$ the squared euclidean norm restricted to the variates in $\mathcal{Z'}$ is such that $\lVert \delta^{[\mathcal{Z'}]} \rVert^2 \geq \frac{\lVert \delta \rVert^2}{\log_2(p)}$.
\end{itemize}
We get the desired result on the detection delay by taking a union bound over these $z(\delta)+1$ events.
\end{proof}

Using Lemma \ref{lemma:H1_proba_bound} we could get similar bounds for a time-varying threshold assuming it is positive, increasing, and concave thresholds.

\subsection{Some concentration bounds}

In this subsection, we provide concentration bounds on our ranked and threshold-based statistics, respectively $\RankS{s}{\tau}{n}$ and $\ThreS{a}{\tau}{n}$ defined in \eqref{eq:rankAndThresholdStats} and \eqref{eq:rankAndThresholdStats2}. These bounds are useful to control the run-length and the detection delay

First, exploiting the peeling argument of \cite{yu2023note} we provide bounds considering all possible values of $\tau$ and $n$ which are useful to control the probability of false alarms (see Proposition \ref{coro:H0_all}). These bounds and their proof are valid for both the known and unknown pre-change mean scenarios. This is because if there is no change $\Evid{i}{\tau}{n}(0)$ and $\Evid{i}{\tau}{n}(.)$ are both Gaussian with mean zero and variance one.

\begin{lemma}\label{lem:peeling_all_change} 
Assume $\varepsilon_t$ are i.i.d. Gaussian with identity variance matrix. For $\alpha$ in $(0, 1)$, taking $x_{n}(\alpha) = 4\log(n) - \log(\alpha),$ 
\begin{itemize}
\item {\bf Rank-1 statistic}
\[
    \mathbb{P}\left( \ \exists \ \tau < n \in \mathbb{N} \quad \text{s.t.} \quad \RankS{1}{\tau}{n} > 2x_{n}(\alpha) + 2\log(2p) \ \right) \leq \alpha.
\]
\item {\bf Rank-p statistic}
\[
    \mathbb{P}\left( \ \exists \ \tau < n \in \mathbb{N} \quad \text{s.t.} \quad \RankS{p}{\tau}{n} >  2x_{n}(\alpha)  + 2\sqrt{px_{n}(\alpha)}  + p \ \right) \leq \alpha\,;
\] 
\item {\bf Thresholded statistic}
\[
    \mathbb{P}\left( \ \exists \ \tau < n \in \mathbb{N} \quad  \text{s.t.} \quad \ThreS{a}{\tau}{n} > 4 x_{n}(\alpha) + 6pe^{-a^2/8} \right)\leq \alpha.
\]
\end{itemize}
\begin{itemize}
\item {\bf Rank-s statistic}
\[
     \mathbb{P}\left( \exists \ \tau < n \in \mathbb{N} \ \text{s.t.} \ \RankS{s}{\tau}{n} >  2(x_{n}(\alpha) + \log{p \choose s})  + 2\sqrt{s(x_{n}(\alpha) + \log{p \choose s})}  + s  \right) \leq \alpha\,;
\]
\end{itemize}
\end{lemma}

\begin{proof}
For all statistics, the proof follows the peeling argument of Lemma 8 in \cite{yu2023note}. In detail, for some generic statistic $\GenS{\tau}{n}$ and some non-decreasing $b_n$, we can bound the probability of the event $\exists \ \tau < n \in \mathbb{N} \ \text{such that} \ \ \GenS{\tau}{n} > b_n$
using two union bounds:

\begin{eqnarray*}
    \mathbb{P}\left(\ \exists \ \tau < n \in \mathbb{N} \ \text{such that} \ \GenS{\tau}{n} > b_n\right) & \leq & \sum_{j=1}^{\infty} \mathbb{P}\left(\ \exists \ \tau < n \in [2^j, 2^{j+1}) \ \text{such that} \ \GenS{\tau}{n} > b_{n}\right) \\
    & \leq & 3 \sum_{j=1}^{\infty} 2^{2j-1} \max_{\begin{array}{c@{}c@{}c@{}c@{}c@{}} 2^j & \leq & n  & < & 2^{j+1} \\ 1  & \leq  & \tau & < & n\end{array}} \mathbb{P}\left(  \GenS{\tau}{n} > b_{n}\right)\\
      & \leq & 3 \sum_{j=1}^{\infty} 2^{2j-1} \max_{1 \leq   \tau  <  2^j} \mathbb{P}(  \GenS{\tau}{2^j} > b_{2^j}).
\end{eqnarray*}
We get the second inequality bounding the number of couples $(\tau,n)$ such that $\tau < n$ with $n$ in $[2^j, 2^{j+1})$, there are:
$$\frac{(2^{j+1}-2^j)(2^{j+1}-2^j-1)}{2} + (2^{j+1}-2^j)^2 = 2^{j-1}(3\times2^j-1)\,.$$
We obtain the last equality using the fact that $n\to b_n$ is an increasing bound and that all the $\GenS{\tau}{n}$ have the same distribution, meaning that $n \to \mathbb{P}(  \GenS{\tau}{n} > b_{n})$ is non-increasing.
To continue, we consider some appropriate concentration bound on all $\GenS{\tau}{n}$ of the form
$\mathbb{P}\left( \GenS{\tau}{n} \geq c(x) \right) \ \leq \ e^{-x}$.
To be specific:
\begin{itemize}
    \item For the rank-1 statistic we take
    \begin{equation} 
P\left( \RankS{1}{\tau}{n} \geq c(x) \right) \ \leq \ e^{-x},
\end{equation}
    with $c(x) = 2x + 2\log(p) + 2\log(2)$ which we obtain using a union bound over $p$ univariate Gaussian random variable.
 \item For the rank-p and rank-s statistics we take
 \begin{equation} 
P\left( \RankS{s}{\tau}{n} \geq  c(x) \right) \ \leq \ e^{-x},
\end{equation}
with $c(x) = s + 2\sqrt{s(x+\log{p \choose s})} + 2(x+\log{p \choose s})$ which we get from Lemma 1 of \cite{laurent2000adaptive} and a union bound over all ${p \choose s}$ choices of $s$ dimensions out of $p$.
    \item For the threshold statistic we take
    \begin{equation} 
P\left( \ThreS{a}{\tau}{n} \geq c(x) \right) \ \leq \ e^{-x},
\end{equation}
with $c(x) = 6p e^{-a^2/8} + 4x$ which we get from Lemma 17 of \cite{chen2020highdimensional}.
\end{itemize}

Setting $b_{2^j} = c(\tilde x_{2^j}(\alpha))$, we remove the probability term using these concentration bounds to write:
$$\mathbb{P}\left(\ \exists \ \tau < n \in \mathbb{N} \ \text{s.t.} \ \GenS{\tau}{n} > b_n\right) \leq 3 \sum_{j=1}^{\infty} 2^{2j-1} \exp(-\tilde x_{2^j}(\alpha))\,.$$
Considering $\exp(\tilde x_{2^j}(\alpha)) = \frac{3}{\alpha}2^{2j-1}j(j+1)$ we get:
$$3 \sum_{j=1}^{\infty} 2^{2j-1} \exp(-\tilde  x_{2^j}(\alpha)) = \sum_{j=1}^{\infty} \frac{\alpha}{j(j+1)} = \alpha\,.$$
and we have the value for evaluating $c$ given by:
$$\tilde x_n(\alpha) = 2\log(n) - \log(\alpha) + \log (\log_2(n)(\log_2(n) + 1)) + \log(\frac{3}{2}) \le 4\log(n) - \log(\alpha) = x_n(\alpha)\,,$$
using relation $2^{2j} \ge \frac{3}{2}j(j+1)$ on integers for the inequality. As the function $c$ is increasing, we can use the simpler value $x_n(\alpha)$ which leads to the proposed result.


\end{proof}

We now provide some bounds to control the detection delay, assuming we know the pre-change mean and the threshold is fixed or concave increasing.

\begin{lemma}\label{lemma:H1_proba_bound}
    Consider a sequence with a change $\delta$ at $\tau^*$, that is $y_t=\varepsilon_t$ for $t \leq \tau^*$ and $y_t=\delta + \varepsilon_t$ for $t > \tau^*$ and $\lVert \delta \rVert > 0$. 
    \begin{itemize}
        \item Suppose the threshold $c_n=c$ is fixed and positive. Then for any $$n - \tau^* \geq \frac{2 c - 8 \log(\alpha/2)}{\lVert \delta \rVert^2} = \Delta(\delta,c,\alpha)$$ we have
    \begin{displaymath}
         \mathbb{P}\Big( \lVert \frac{\sum_{t=\tau^*+1}^n y_t}{\sqrt{n-\tau^*}} \Vert^2 < c\Big) \leq \alpha.
    \end{displaymath}
    
    \item Suppose the threshold $c_n=c$ is fixed, positive and larger than $p$. Then for any $$n - \tau^* \geq \frac{2 (c-p) + 4\sqrt{p\log(3/\alpha)}- 8 \log(\alpha/3)}{\lVert \delta \rVert^2} = \Delta'(\delta,c,\alpha, p)$$ we have
    \begin{displaymath}
         \mathbb{P}\Big( \lVert \frac{\sum_{t=\tau^*+1}^n y_t}{\sqrt{n-\tau^*}} \Vert^2 < c\Big) \leq \alpha.
    \end{displaymath}
     \item Suppose the threshold $c_n$ is increasing, concave and its derivative $c'_n$ well-defined for all integers $n$. Then for any $n - \tau \geq \Delta(\delta,c_{\tau^*},\alpha) + 2 \frac{c'_{\tau^*}}{\lVert \delta \rVert^2}$ we have
    \begin{displaymath}
         \mathbb{P}\Big( \lVert \frac{\sum_{t=\tau^*+1}^n y_t}{\sqrt{n-\tau^*}} \Vert^2 < c_n\Big) \leq \alpha. 
    \end{displaymath}

    \end{itemize}
\end{lemma}

\begin{proof}
We define $u_n = \sum_{t=\tau+1}^n \langle \frac{\delta}{\lVert \delta \rVert}, \varepsilon_t \rangle$. Note that $u_n$ is $\mathcal{N}(0, 1)$.

We get
    \begin{eqnarray*}
    \left\lVert \frac{\sum_{t=\tau+1}^n y_t}{\sqrt{n-\tau}} \right\Vert^2 & = &\left\lVert \sqrt{n-\tau}\delta + \frac{\sum_{t=\tau+1}^n \varepsilon_t}{\sqrt{n-\tau}} \right\rVert^2 \\
    & = & (n-\tau) \lVert\delta\rVert^2 + 2 \lVert \delta \lVert \sum_{t=\tau+1}^n \langle \frac{\delta}{\lVert \delta \rVert}, \varepsilon_t \rangle + \left\lVert \frac{\sum_{t=\tau+1}^n \varepsilon_t}{\sqrt{n-\tau}} \right\rVert^2 \\
    & \geq & (n-\tau) \lVert\delta\rVert^2 + 2 \lVert \delta \lVert u_n \\
    & \geq & (\sqrt{n-\tau}\lVert\delta\rVert + u_n)^2 - u_n^2\,.
\end{eqnarray*}

Thus, we have:
\begin{equation}\label{eq:H1_1st_imply}
\begin{aligned}
   (n-\tau)\lVert\delta\rVert^2 \ge \left(\sqrt{c_n + u_n^2} - u_n\right)^2 & \iff & (\sqrt{n-\tau}\lVert\delta\rVert + u_n)^2 - u_n^2 \geq c_n \\
   & \implies & \left\lVert \frac{\sum_{t=\tau+1}^n y_t}{\sqrt{n-\tau}} \right\Vert^2 \geq c_n\,.
   \end{aligned}
\end{equation}
Using the fact that for any $(a, b)$ in $\mathbb{R}^2$, $2a^2 + 2b^2 \geq (a-b)^2$ we get
$$(n-\tau)\lVert\delta\rVert^2\ge 2c_n + 4u_n^2 \implies (n-\tau)\lVert\delta\rVert^2 \ge \left(\sqrt{c_n + u_n^2} - u_n\right)^2\,.$$
In terms of probability events, we are looking for a delay $(n-\tau)$ such that the probability of detecting a change is greater than $1-\alpha$, this can be written as:
$$1-\alpha \le \mathbb{P}\Big((n-\tau)\lVert\delta\rVert^2\ge 2c_n + 4u_n^2\Big) \le \mathbb{P}\Big( \left\lVert \frac{\sum_{t=\tau+1}^n y_t}{\sqrt{n-\tau}} \right\Vert^2 \ge c_n\Big)\,.$$
It is known that $u_n$ is $\mathcal{N}(0,1)$ and that $1-\alpha \le \mathbb{P}(-2\log(\alpha/2) > u_n^2)$ so we solve for a fixed threshold $c_n = c$: $-8\log(\alpha/2) \le (n-\tau)\lVert\delta\rVert^2- 2c$, leading to the relation:
$$ n-\tau \ge  \frac{2 c - 8\log(\alpha/2)}{\lVert \delta \rVert^2}\,.$$

For the second result, we proceed similarly but including a lower bound on $\lVert \frac{\sum_{i=\tau+1}^n \varepsilon_i}{\sqrt{n-\tau}} \rVert^2$ using \cite{laurent2000adaptive} 
$$\left\lVert \frac{\sum_{i=\tau+1}^n \varepsilon_i}{\sqrt{n-\tau}} \right\rVert^2 \geq p - 2\sqrt{p\log(3/\alpha)}\,. $$

For the third result, we use the concavity property: $c_n \le c_{n_0} + (n-n_0)c'_{n_0}$ in order to write:
$$-8\log(\alpha/2)+2c_n \le -8\log(\alpha/2)+2( c_{n_0}+(n-n_0)c'_{n_0}) \le (n-\tau)\lVert\delta\rVert^2\,.$$
If the second inequality is true, the full line is true, which leads to the result with $n_0 = \tau$. Notice that the bound $c_n$ is often chosen proportional to $\log(n)$ such that the delay found with constant $c_n$ is shifted by a constant when $\lVert\delta\rVert^2 \ge \frac{1}{n_0}\frac{n-n_0}{n-\tau}$.

\end{proof}

\section{Acknowledgements and Funding}
This work was supported by EPSRC grants EP/Z531327/1, EP/W522612/1 and EP/N031938/1 and an ATIGE grant from Génopole. The IPS2 benefit from the support of Saclay Plant Sciences-SPS [ANR-17-EUR-0007]

\end{appendices}

\end{document}